\documentclass[a4paper,onecolumn,10pt,unpublished]{quantumarticle}
\pdfoutput=1
\input{preamble.tex}
\hypersetup{pdftitle={Uniform Hiding and Two Routes to Relative Accuracy in Gaussian Boson Sampling}}
\begin{document}
\title{{Uniform Hiding and Two Routes to Relative Accuracy in Gaussian Boson Sampling}}
\author{Hongru Zhao}
\affiliation{School of Statistics, University of Minnesota, Minneapolis, Minnesota, USA}
\email{zhao1118@umn.edu}
\begin{abstract}
{Gaussian boson sampling requires control of how closely finite optical
matrices follow Gaussian reference laws. We prove an explicit total
variation bound of order $N^2/M$ between a rescaled Haar transpose Gram
block and its Gaussian transpose Gram counterpart, where $N$ is the detected
photon count and $M$ is the number of optical modes. The bound establishes
quantitative product hiding uniformly over every number of squeezed
inputs. The proof combines Stiefel recursion, centered circular orthogonal
ensemble scores, and a rectangular entropy estimate. Applications combine
hiding with local hafnian bounds to obtain relative probability guarantees
and connect them to sampler error through exact photon-sector normalization.}
\end{abstract}
\maketitle
\begingroup\section{Introduction}
\label{sec:introduction}
Gaussian boson sampling (GBS) turns squeezed light, passive interferometry,
and photon counting into a sampling problem whose probabilities contain
hafnians \cite{HamiltonEtAl2017GBS,KruseEtAl2019GBS}. This makes it a
natural setting for studying quantum computational advantage, but the
complexity argument requires more than worst case hardness of a single
probability. A finite optical instance must be related to a random matrix
estimation problem, and additive probability accuracy must be compared with
the size of the probability itself
\cite{AaronsonArkhipov2013,DeshpandeEtAl2022,HangleiterEisert2023}.
Experiments with fixed and programmable photonic architectures make the
parameter dependence of these reductions relevant
\cite{ZhongEtAl2020Photons,ZhongEtAl2021PhaseProgrammable,MadsenEtAl2022Programmable,DengEtAl2023PseudoPNR}.

For equal squeezing, there are two Gaussian reference ensembles. One retains
the finite transpose Gram product $GG^{\T}$ of a rectangular Gaussian
matrix. The other replaces that product by an independent complex symmetric
Gaussian matrix. Kruse et al.\ already separated transpose Gram and
independent symmetric formulations of GBS \cite[Secs.~V.A and V.B]{KruseEtAl2019GBS}.
Here we compare two reference laws for the \emph{same} equal-squeezing Haar
interferometer model, rather than changing the preparation architecture.
The distinction matters when the number of squeezed inputs is not much
larger than the square of the detected photon number.

Our first result is a uniform finite hiding estimate. If $N$ output rows and
$K$ input columns are selected from a Haar unitary of size $M$, then
\begin{equation}
 \begin{aligned}
 &d_{\rm TV}\!\left(\Law(MU_{N,K}U_{N,K}^{\T}),
                   \Law(G_{N,K}G_{N,K}^{\T})\right)\\
 &\qquad\le \min\{1,C_\ast N^2/M\},
 \end{aligned}
 \label{eq:intro-hiding}
\end{equation}
where $C_\ast=615172$ and the bound holds for every $1\le N,K\le M$.
This resolves Conjecture~1 (Formal) of Ehrenberg et al., stated in
Supplemental Eq.~(S62), with its original $m\ge n^2/\delta$ scaling
and uniformly over $1\le k\le m$
\cite{EhrenbergEtAl2025Transition}. The $K<N$ case is included in the proof
of Theorem~\ref{thm:uniform-product-hiding};
Corollary~\ref{cor:all-input-hiding} and the accompanying notation
dictionary give the conjecture's quantitative specialization.
The large explicit constant is not
optimized; the structural conclusion is the uniform $N^2/M$ dependence.

This result should be distinguished from the recent resolution of a
polynomial hiding formulation by Shou et al.
\cite{ShouEtAl2026ArbitrarySqueezers}. Their Theorem~1.1 gives an
$O(N/\sqrt K)$ comparison with an independent symmetric Gaussian matrix.
Their Remark~1.1 also proves product-law hiding, with a sufficient
$M\ge N^4/(c\delta^4)$ condition in the range $N\le K\le M$.
Our contribution is the explicit quadratic ambient-size estimate for the
finite product law, not the first identification or proof of a product
reference. Their approximate instance-generation construction is a separate
algorithmic result that is not supplied by Eq.~\eqref{eq:intro-hiding}.

The second result is a comparison of the two routes to relative probability
accuracy. The companion paper proves local small-ball estimates for both
Gaussian hafnian ensembles \cite{ZhaoAnticoncentration2026}. Composing
those estimates with hiding gives two finite-Haar lower-tail bounds. Their
probability scales differ by an explicit product, so a valid comparison must
use one estimator and one \emph{absolute} additive threshold. Our main
comparison theorem separates additive estimation error, local
anticoncentration, and matrix replacement, and permits optimization over
the threshold. In the window
\begin{equation}
 cN^2/\log N\le K\ll N^2,\qquad M\gg N^2,
 \label{eq:intro-window}
\end{equation}
the product route can give a vanishing certificate while the currently
available independent-reference hiding envelope is trivial. This is a
separation of proved bounds, not a lower bound on the true error of the
independent-reference approximation.

The main results and the optical normalization are presented before the
random matrix proof. Section~\ref{sec:two-route-results} contains the two
routes and their common-threshold comparison.
Section~\ref{sec:sampling-interface} connects the scales to the photon-number
sector and to total variation error of a proposed sampler. The later
sections prove hiding through Stiefel recursion, centered circular orthogonal
ensemble scores, and a rectangular entropy estimate. The companion supplies
the Gaussian small-ball theorems through a different Fourier and conditional
Wishart argument; neither proof uses the other paper's main theorem.

We work with ideal pure, zero-displacement GBS and a prescribed collision-free
pattern, or a pattern label chosen independently of the interferometer.
The results give rigorous components of approximate sampling reductions.
Average-case hardness, efficient prescribed-instance embedding, and
precision-controlled approximate counting must still be supplied in any
claim of computational hardness. We return to these distinctions in
Section~\ref{sec:discussion}.
\endgroup
\begingroup\section{Uniform hiding for the optical matrix}
\label{sec:hiding-main}

\begingroup
There are $M$ optical modes. The first $K$ inputs are single-mode squeezed
vacua with a common strength $\xi>0$, and all remaining inputs are vacuum.
For a fixed set $S\subseteq[M]$ with $|S|=N=2n$, put
$A_S=U_{S,[K]}U_{S,[K]}^{\T}$. Photon-number-resolving detection assigns the
collision-free occupation pattern with one photon on $S$ the probability
\begin{equation}
 p_S(U;\xi)=\frac{\tanh(\xi)^N}{\cosh(\xi)^K}|\haf(A_S)|^2.
 \label{eq:gbs-probability}
\end{equation}
This is the probability of the full output pattern, not the probability
conditioned on the total photon number \cite{HamiltonEtAl2017GBS,KruseEtAl2019GBS}.
The transpose is ordinary transpose. Conjugate transpose would give a
different matrix and a different optical model.
\endgroup

{Let $U$ be Haar distributed in $\mathrm U(M)$.  By Haar bi-invariance,
every deterministic $N\times K$ subblock selected by rows and columns has
the same law as the upper block.  We therefore write
$U_{N,K}:=U_{[N],[K]}$ in the labeled
theorem below; any fixed sets $S,T\subseteq[M]$ with $|S|=N$ and $|T|=K$
follow by deterministic row and column permutations.  The physical
specialization above takes $T=[K]$.}
Let $G_{N,K}$ have independent entries distributed as
{$\CN(0,1)$, the standard circular complex Gaussian law with
density $z\mapsto\pi^{-1}e^{-|z|^2}$ on $\C$; equivalently, its real and
imaginary parts are independent $\mathcal{N}(0,1/2)$ variables, so
$\E|g|^2=1$ and $\E g^2=0$.}  We use
\[
 d_{\rm TV}(\mu,\nu):=\sup_E|\mu(E)-\nu(E)|
\]
for probability measures.  Every matrix law is regarded as a measure on the full ambient
space $\C^{N\times N}$; the laws in the theorem are supported on its
measurable subset $\Sym_N(\C)$.  For a finite signed measure $\eta$ on this
ambient space, write
$\norm{\eta}_{\rm var}:=|\eta|(\C^{N\times N})$ for its variation norm.  Thus
$d_{\rm TV}(\mu,\nu)=\tfrac12\norm{\mu-\nu}_{\rm var}$; the two conventions
will not be conflated below.  Scaling both random matrices by the same
nonzero constant leaves this distance unchanged.
{Here and throughout, ${}^{\T}$ denotes ordinary transpose and
${}^{*}$ denotes conjugate transpose.}
{Set the universal constant $C_\ast:=615172$ and define}
\begin{equation}
 \delta_{M,N}:=\min\!\left\{1,{C_\ast}\frac{N^2}{M}\right\}.
 \label{eq:hiding-remainder}
\end{equation}

\begin{samepage}
\begin{theorem}[Uniform transpose Gram hiding]
\label{thm:uniform-product-hiding}
For every $1\le N\le M$ and $1\le K\le M$,
\begin{equation}
 \begin{aligned}
 d_{\rm TV}\!\Bigl(&\Law\bigl(\tfrac{M}{\sqrt K}
 U_{N,K}U_{N,K}^{\T}\bigr),\\[-2pt]
 &\Law\bigl(\tfrac1{\sqrt K}G_{N,K}G_{N,K}^{\T}\bigr)\Bigr)
 \le \delta_{M,N}.
 \end{aligned}
 \label{eq:uniform-product-hiding}
\end{equation}
Equivalently,
\begin{equation}
 \begin{aligned}
 d_{\rm TV}\!\bigl(&\Law(MU_{N,K}U_{N,K}^{\T}),\\[-2pt]
 &\Law(G_{N,K}G_{N,K}^{\T})\bigr)\le\delta_{M,N}.
 \end{aligned}
 \label{eq:uniform-product-hiding-unscaled}
\end{equation}
\end{theorem}
\end{samepage}

\begin{proof}
The proof of Theorem~\ref{thm:uniform-product-hiding} is given in
Sec.~\ref{subsec:uniform-theorem-proof}.
\end{proof}

{A Haar entry has variance $1/M$, so a typical entry of
$U_{N,K}U_{N,K}^{\T}$ has fluctuation scale $\sqrt K/M$.  The prefactor
$M/\sqrt K$ therefore matches the unit scale of
$K^{-1/2}G_{N,K}G_{N,K}^{\T}$.}

The theorem is finite, explicit, and uniform in $K$.  For a target error
$0<\delta<1$, a sufficient ambient dimension is
\begin{equation}
 M\ge \frac{{C_\ast}}{\delta}N^2.
 \label{eq:certified-ambient-size}
\end{equation}
The constant is mathematically explicit but large; Eq.~\eqref{eq:certified-ambient-size}
is a rigorous finite guarantee and an asymptotic $M\gg N^2$ criterion, not
an optimized experimental threshold.  Appendix~\ref{app:external-inputs}
explains the four literature axioms imported by the Lean formalization
and their correspondence with the cited sources.

\begingroup
Taking $N=2n$ in Theorem~\ref{thm:uniform-product-hiding} gives the following
immediate consequence.
\begin{corollary}[Quantitative hiding conjecture]
\label{cor:all-input-hiding}
Let $n\ge1$, $N=2n\le M$, $1\le K\le M$, and $\delta>0$.
If $M\ge n^2/\delta$, then the total variation
distance in Eq.~\eqref{eq:uniform-product-hiding-unscaled} is at most
$4C_\ast\delta$.
\end{corollary}

To match the conjecture literally, Ehrenberg et al.\ use a $k\times2n$
block and a Gaussian factor with entry variance $1/m$. Since $U^{\T}$ is
also Haar, their block has the same law as $U_{2n,K}^{\T}$ when $m=M$
and $k=K$. Setting $X=M^{-1/2}G_{2n,K}^{\T}$ gives
$X^{\T}X=M^{-1}G_{2n,K}G_{2n,K}^{\T}$. Scaling both product laws by $M$
therefore identifies their distance with ours.
Thus Corollary~\ref{cor:all-input-hiding} resolves Conjecture~1 (Formal),
Supplemental Eq.~(S62), in its stated form: the total variation distance
is at most $4C_\ast\delta$ under $m\ge n^2/\delta$, uniformly in $1\le k\le m$
\cite[Conjecture~1 and Supplemental Eq.~(S62)]{EhrenbergEtAl2025Transition}.
This is a distributional assertion. It does not construct a Haar
interferometer conditional on a prescribed product instance.
\endgroup

\begin{table}[!htbp]
\centering\small
\caption{{The reference ensembles and the claims compared in this paper.
The cited $C'$ estimate is asymptotic; $C_\ast$ is explicit. The
product-hiding conjecture is distinct from an instance-generation algorithm.}}
\label{tab:literature-comparison}
\begin{tabular}{@{}>{\raggedright\arraybackslash}p{.18\textwidth}>{\raggedright\arraybackslash}p{.19\textwidth}>{\raggedright\arraybackslash}p{.23\textwidth}>{\raggedright\arraybackslash}p{.30\textwidth}@{}}
\toprule
Result & Reference law & Quantitative guarantee & Scope\\
\midrule
Ehrenberg et al.\ formal conjecture \cite{EhrenbergEtAl2025Transition}
& Finite $X^{\T}X$
& $m\ge n^2/\delta$ implies $O(\delta)$ in total variation
& All input counts $1\le k\le m$, with $2n$ selected modes\\
Shou et al.\ Theorem~1.1 \cite{ShouEtAl2026ArbitrarySqueezers}
& Independent $G_N^{\rm sym}$
& $O(N/\sqrt K)$
& $N\le K\le M$; approximate instance generation treated separately there\\
Shou et al.\ Remark~1.1 \cite{ShouEtAl2026ArbitrarySqueezers}
& Finite $GG^{\T}/\sqrt K$
& Sufficient $M\ge N^4/(c\delta^4)$
& Product-law comparison in their polynomial hiding formulation\\
Theorem~\ref{thm:uniform-product-hiding} and Corollary~\ref{cor:all-input-hiding}
& Finite $GG^{\T}/\sqrt K$
& $\min\{1,C_\ast N^2/M\}$
& All $1\le N,K\le M$\\
\bottomrule
\end{tabular}
\end{table}
\endgroup
\begingroup\section{Two routes to relative probability accuracy}
\label{sec:two-route-results}
\subsection{Gaussian inputs and their probability scales}
Let $N=2n$ and, throughout the common comparison, assume $4n\le K\le M$.
For $G\in\C^{N\times K}$ with independent $\CN(0,1)$ entries, set
\begin{align}
 H_{K,n}&=\haf(GG^{\T}),\nonumber\\
 \sigma_{K,n}^2&=\E|H_{K,n}|^2
 =(2n-1)!!\prod_{q=0}^{n-1}(K+2q).
 \label{eq:gram-hafnian-variance}
\end{align}
The exact second moment is due to Ehrenberg et al.
\cite[Theorem~1 and Supplemental Sec.~S3]{EhrenbergEtAl2025Transition};
see also the independent derivation in Ref.~\cite{ZhaoExactMoments2026}.
The companion's Gaussian theorem gives, for all $z\in\C$ and
$\varepsilon\ge0$,
\begin{equation}
 \Pp\{|H_{K,n}-z|\le\varepsilon\sigma_{K,n}\}
 \le\min\{1,B_{K,n}\varepsilon^2\},
 \label{eq:imported-anticoncentration}
\end{equation}
where
\begin{align}
 b_n&=\frac{2\Gamma(n+1/2)}{\sqrt\pi\Gamma(n)},\nonumber\\
 B_{K,n}&=b_n\frac{K}{K-1}
       \prod_{r=2}^{n}\frac{K+2r-2}{K-4r+1}.
 \label{eq:anticoncentration-coefficients}
\end{align}
This is a local bound at shrinking radii, rather than only a moment-ratio
criterion for weak anticoncentration
\cite{EhrenbergEtAl2025SecondMoment,KolarovszkiEtAl2026Framework,ZhaoAnticoncentration2026}.

The second reference is $G_N^{\rm sym}$, whose entries above the diagonal
are independent $\CN(0,1)$ variables and whose independent diagonal entries
have complex variance two. Put
$H_n^{\rm sym}=\haf(G_N^{\rm sym})$ and
$(\sigma_n^{\rm sym})^2=(2n-1)!!$. Diagonal entries do not enter the
hafnian. The companion also proves
\begin{equation}
 \Pp\{|H_n^{\rm sym}-z|\le\varepsilon\sigma_n^{\rm sym}\}
 \le\min\{1,b_n\varepsilon^2\}.
 \label{eq:symmetric-anticoncentration}
\end{equation}
Here $b_n\le2\sqrt{n/\pi}$. This establishes polynomial lower-tail control for the independent symmetric
Gaussian hafnian appearing in the original GBS anticoncentration proposals
\cite{HamiltonEtAl2017GBS,KruseEtAl2019GBS,ZhaoAnticoncentration2026}.
It does not settle the separate average-case estimation-hardness assumption.

For Route~2, let $\delta_2$ be any certified bound satisfying
\begin{equation}
 d_{\rm TV}\!\left(\Law\!\left(\frac M{\sqrt K}A_S\right),
                         \Law(G_N^{\rm sym})\right)\le\delta_2\le1.
 \label{eq:route-two-input}
\end{equation}
The value one is always admissible. In the asymptotic regime of
Shou et al., Theorem~1.1, one can use
\begin{equation}
 \overline\delta_2=\min\{1,C'N/\sqrt K\},
 \label{eq:symmetric-hiding-rate}
\end{equation}
where $C'$ is the universal constant implicit in their estimate as
$K\to\infty$ with $N\le K\le M$ \cite{ShouEtAl2026ArbitrarySqueezers}.
Our finite statements use the certified input
Eq.~\eqref{eq:route-two-input}; Eq.~\eqref{eq:symmetric-hiding-rate} is used
only in its cited asymptotic range. Write $\delta_1=\delta_{M,N}$ for the
explicit Route~1 error.

The natural probability scales are
\begin{align}
 p_1&=\frac{\tanh(\xi)^N}{M^N\cosh(\xi)^K}\sigma_{K,n}^2,
 \label{eq:reference-probability}\\
 p_2&=\frac{\tanh(\xi)^N}{M^N\cosh(\xi)^K}K^n(2n-1)!!.
 \label{eq:symmetric-reference-probability}
\end{align}
Their exact ratio is
\begin{equation}
 p_1=\mathcal R^{\mathrm{scale}}_{K,n}p_2,\qquad
 \mathcal R^{\mathrm{scale}}_{K,n}=\prod_{q=0}^{n-1}\left(1+\frac{2q}{K}\right).
 \label{eq:reference-scale-ratio}
\end{equation}
Neither scale is asserted to equal the finite-Haar mean of $p_S$.
In particular, total variation control alone does not transfer a
relative expectation estimate for a growing, unbounded Gaussian observable.

\begin{figure}[t]
 \centering
 \includegraphics[width=.62\textwidth]{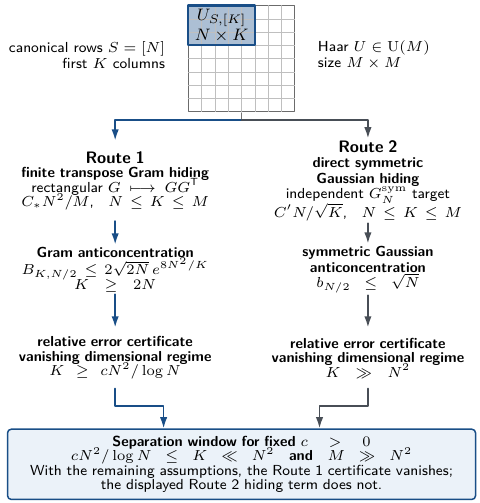}
 \caption{{Two reference routes for the same Haar interferometer block.
 Route~1 retains the dependent finite transpose Gram law; Route~2 uses the
 independent symmetric Gaussian law. The additive threshold is common, but
 the reference probability scales differ by $\mathcal R^{\mathrm{scale}}_{K,n}$.
 The lower strip concerns the currently available error envelopes, not an
 impossibility result for Route~2. The common comparison also assumes
 $K\ge4n=2N$.}}
 \label{fig:two-routes}
\end{figure}

\subsection{Finite-Haar local anticoncentration}
\begin{theorem}[Two finite-interferometer lower tails]
\label{thm:two-haar-tails}
Let $n\ge1$, $N=2n$, $4n\le K\le M$, and let $S$ be a fixed
collision-free pattern. For every $z\in\C$ and $\varepsilon\ge0$,
\begin{align}
 &\Pp_U\{|\haf(MA_S)-z|\le\varepsilon\sigma_{K,n}\}
 \nonumber\\[-2pt]
 &\hspace{1cm}\le\min\{1,B_{K,n}\varepsilon^2+\delta_1\},
 \label{eq:finite-haar-composition}\\
 &\Pp_U\{|\haf(MA_S/\sqrt K)-z|
                  \le\varepsilon\sigma_n^{\rm sym}\}
 \nonumber\\[-2pt]
 &\hspace{1cm}\le\min\{1,b_n\varepsilon^2+\delta_2\}.
 \label{eq:symmetric-haar-composition}
\end{align}
Consequently, for every $t\ge0$,
\begin{align}
 \Pp_U\{p_S\le t p_1\}&\le\min\{1,B_{K,n}t+\delta_1\},
 \label{eq:small-denominator}\\
 \Pp_U\{p_S\le t p_2\}&\le\min\{1,b_nt+\delta_2\}.
 \label{eq:symmetric-small-denominator}
\end{align}
The same conclusions hold after averaging over any independent choice of
pattern label with size $N$.
\end{theorem}
\begin{proof}
The inverse image of a disk under the hafnian is measurable.
Theorem~\ref{thm:uniform-product-hiding} transfers
Eq.~\eqref{eq:imported-anticoncentration} with error $\delta_1$.
The certified comparison in Eq.~\eqref{eq:route-two-input} transfers
Eq.~\eqref{eq:symmetric-anticoncentration} with error $\delta_2$.
Hafnian homogeneity gives
$p_S/p_1=|\haf(MA_S)|^2/\sigma_{K,n}^2$ and
$p_S/p_2=|\haf(MA_S/\sqrt K)|^2/(\sigma_n^{\rm sym})^2$.
Set $z=0$ and $\varepsilon=\sqrt t$.
Haar row invariance gives the same bound for every fixed label, which can
then be averaged against a label law independent of $U$.
\end{proof}

\subsection{A common-threshold comparison}
Fix one measurable randomized estimate $\widetilde p$ of $p_{\bm S}(U;\xi)$,
where $\bm S$ is fixed or chosen independently of $U$. The estimator may
use internal randomness depending on $(U,\bm S)$.
For $\Delta p=\widetilde p-p_{\bm S}$ and $\tau\ge0$, define its
\emph{joint} additive tail by $\gamma(\tau)=\Pp\{|\Delta p|>\tau\}$.
No additive guarantee is inferred from Haar invariance.

\begin{theorem}[Relative accuracy at one absolute additive threshold]
\label{thm:fair-comparison}
Under the assumptions of Theorem~\ref{thm:two-haar-tails}, for $\rho>0$
and every $\tau\ge0$, let
$F_\rho=\Pp\{|\Delta p|>\rho p_{\bm S}\}$. Then
\begin{align}
 F_\rho&\le\min\{E_1(\tau),E_2(\tau)\},
 \label{eq:fair-budget}\\
 E_1(\tau)&=\min\!\left\{1,\gamma(\tau)
             +\frac{B_{K,n}\tau}{\rho p_1}+\delta_1\right\},\nonumber\\
 E_2(\tau)&=\min\!\left\{1,\gamma(\tau)
             +\frac{b_n\mathcal R^{\mathrm{scale}}_{K,n}\tau}{\rho p_1}+\delta_2\right\}.
 \nonumber
\end{align}
In particular, with $\mathcal E_i=\inf_{\tau\ge0}E_i(\tau)$,
\begin{equation}
 F_\rho\le\min\{\mathcal E_1,\mathcal E_2\}.
 \label{eq:optimized-budget}
\end{equation}
\end{theorem}
\begin{proof}
For each $\tau\ge0$,
\begin{equation}
 \{|\Delta p|>\rho p_{\bm S}\}
 \subseteq\{|\Delta p|>\tau\}\cup\{p_{\bm S}\le\tau/\rho\}.
 \label{eq:relative-failure-inclusion}
\end{equation}
Apply Theorem~\ref{thm:two-haar-tails} and the union bound, cap by one,
and substitute $p_1=\mathcal R^{\mathrm{scale}}_{K,n}p_2$. Since $F_\rho$ is a lower bound
for every element of each nonempty set $\{E_i(\tau):\tau\ge0\}$,
it is at most each infimum. No attainment of the infimum is required.
\end{proof}

The comparison does not say that Route~1 is always better. In fact,
\begin{equation}
 \frac{B_{K,n}}{b_n\mathcal R^{\mathrm{scale}}_{K,n}}
   =\frac K{K-1}\prod_{r=2}^{n}\frac K{K-4r+1}>1.
 \label{eq:coefficient-tradeoff}
\end{equation}
Thus Route~2 has the smaller local-denominator coefficient at the common
absolute threshold, while Route~1 can have the smaller hiding error.
Comparing at equal normalized errors $\tau/p_i$ would change the estimation
task and obscure this tradeoff.

\subsection{A regime separating the available certificates}
The companion coefficient expansion implies, uniformly when
$N^2/K\le\kappa\log N$ for fixed $\kappa>0$,
\begin{align}
 B_{K,n}&=O(N^{1/2+3\kappa/4}),\nonumber\\
 \mathcal R^{\mathrm{scale}}_{K,n}&=O(N^{\kappa/4}),\qquad
 b_n\mathcal R^{\mathrm{scale}}_{K,n}=O(N^{1/2+\kappa/4}).
 \label{eq:polynomial-coefficients}
\end{align}
Indeed, $\log(B_{K,n}/b_n)=3n^2/K+o(1)$ in this range, and
$\log\mathcal R^{\mathrm{scale}}_{K,n}\le n(n-1)/K$.

In the regime $N^2/K\to0$, this same inequality and
$\mathcal R^{\mathrm{scale}}_{K,n}\ge1$ give $\mathcal R^{\mathrm{scale}}_{K,n}\to1$.
Thus the two reference scales agree asymptotically in the regime where
the cited independent-reference hiding envelope vanishes.

\begin{corollary}[Separation of the current error envelopes]
\label{cor:separation-window}
Fix $c>0$ and consider a sequence with $N=2n\to\infty$,
$cN^2/\log N\le K\ll N^2$, and $N^2/M\to0$.
Suppose that for some $a>1/2+3/(4c)$ there are thresholds with
$\tau/(\rho p_1)=N^{-a}$ and $\gamma(\tau)\to0$.
Then $\mathcal E_1\to0$.
If the Route~2 certificate uses
$\delta_2=\overline\delta_2$ from Eq.~\eqref{eq:symmetric-hiding-rate},
then $E_2(\tau)=1$ for every threshold eventually, and hence
$\mathcal E_2=1$ eventually.
\end{corollary}
\begin{proof}
The first conclusion follows by substituting
Eq.~\eqref{eq:polynomial-coefficients} with $\kappa=1/c$ into
Eq.~\eqref{eq:fair-budget}, together with $\delta_1\to0$.
For the second, $K/N^2\to0$ implies $C'N/\sqrt K\to\infty$,
so the displayed certified remainder equals one. All other terms are
nonnegative, irrespective of threshold selection.
\end{proof}

This corollary compares upper bounds. It does not prove that the true
Route~2 total variation distance stays positive, nor does it exclude a
better independent-reference theorem. It also leaves the estimator's
additive guarantee as an explicit premise.
\endgroup
\begingroup\section{Photon-sector normalization and sampling error}
\label{sec:sampling-interface}
The hardness argument for approximate boson sampling relates total
variation sampling error to additive error at a typical output label, and
uses anticoncentration to pass to relative accuracy
\cite[Secs.~5.2 and~7]{AaronsonArkhipov2013}. For completeness, we develop
the corresponding normalization and sampling-error conversion for GBS
\cite{HamiltonEtAl2017GBS,KruseEtAl2019GBS}. Here the total photon number
fluctuates, so the reference scales must account for both the number of
collision-free labels and the probability of the chosen photon sector.
Throughout, the Gaussian reference expectation is kept distinct from the
finite-Haar law.

\subsection{An exact reference-scale identity}
Let $D_{M,N}=\binom MN$ and write $(a)_n=\prod_{j=0}^{n-1}(a+j)$.
For $K$ identical squeezed inputs, the probability of total photon number
$N=2n$, before conditioning on collisions, is
\begin{equation}
 W_{K,n}(\xi)=\frac{(K/2)_n}{n!}
            \tanh(\xi)^{2n}\cosh(\xi)^{-K}.
 \label{eq:sector-probability}
\end{equation}
This is the negative-binomial photon-pair law
\cite[Eq.~(12)]{HamiltonEtAl2017GBS}.
It also follows directly by raising the single-mode pair-count generating
function $(1-\tanh^2(\xi)z)^{-1/2}/\cosh(\xi)$ to the $K$th power.

\begin{proposition}[Probability scale and the size of the label set]
\label{prop:sector-scale}
For {$n\ge1$,} $M\ge N=2n$, $K\ge1$, and $\xi>0$,
\begin{equation}
 D_{M,N}p_1=W_{K,n}(\xi)\prod_{j=0}^{N-1}(1-j/M).
 \label{eq:sector-scale-identity}
\end{equation}
The squeezing choice $\tanh^2\xi=N/(K+N)$, equivalently
$K\sinh^2\xi=N$, maximizes $W_{K,n}(\xi)$. At that choice, uniformly for
$K\ge4n$ as $n\to\infty$,
\begin{equation}
 W_{K,n}(\xi)=\frac{1+O(n^{-1})}
                  {\sqrt{2\pi n(1+2n/K)}}.
 \label{eq:sector-stirling}
\end{equation}
Consequently $D_{M,N}p_1=\Theta(N^{-1/2})$ if also $N^2/M\to0$.
\end{proposition}
\begin{proof}
Using $(2n-1)!!=(2n)!/(2^n n!)$ and
$\prod_{q=0}^{n-1}(K+2q)=2^n(K/2)_n$ in
Eq.~\eqref{eq:reference-probability} gives
\[
 \binom M{2n}\frac{\sigma_{K,n}^2}{M^{2n}}
   =\frac{(K/2)_n}{n!}\frac{M(M-1)\cdots(M-2n+1)}{M^{2n}},
\]
which proves Eq.~\eqref{eq:sector-scale-identity}.
For $x=\tanh^2\xi\in(0,1)$ the variable factor is
$x^n(1-x)^{K/2}$. Its logarithmic derivative vanishes uniquely at
$x=n/(n+K/2)$, which is its maximum.
Put $a=K/2$ and use
\[
 W=\frac{\Gamma(a+n)}{\Gamma(a)\Gamma(n+1)}
       \left(\frac n{a+n}\right)^n
       \left(\frac a{a+n}\right)^a.
\]
Stirling's formula \cite[Sec.~5.11(i)--(ii)]{DLMFStirling} has relative error
$1+O(a^{-1}+n^{-1})=1+O(n^{-1})$ uniformly for $a\ge2n$.
Cancellation of the powers gives
$\sqrt{a/(2\pi n(a+n))}$, proving Eq.~\eqref{eq:sector-stirling}.
Finally,
$\prod_{j=0}^{N-1}(1-j/M)=1+O(N^2/M)$ when $N^2/M\to0$,
by the sum of logarithms, or by the elementary product bound.
\end{proof}

The identity concerns the Gaussian reference scale $p_1$, not the
finite-interferometer mass of collision-free outputs. It is therefore not
an additional assertion of collision suppression. It shows that the scale
appearing in our error budget is inverse polynomial relative to the number
of labels at the physically natural mean-matched squeezing.

\subsection{From total variation sampling error to an additive tail}
Let $P_U$ be the ideal distribution on the full photon-counting outcome
space and let $Q_U$ be another probability distribution on that same space.
Denote their probabilities at a collision-free label by $p_S$ and $q_S$.
Suppose
\begin{equation}
 \E_U d_{\rm TV}(P_U,Q_U)\le\epsilon_{\rm sam}.
 \label{eq:sampler-assumption}
\end{equation}
A per-circuit bound implies this assumption, but is not required.
For an independent uniform label $\bm S$ from the $D_{M,N}$ patterns,
\begin{equation}
 \E_{U,\bm S}|p_{\bm S}-q_{\bm S}|
       \le\frac{2\epsilon_{\rm sam}}{D_{M,N}}.
 \label{eq:sampler-mean}
\end{equation}
This follows by restricting the full $\ell^1$ sum to the chosen sector.
For $0<\zeta<1$, Markov's inequality therefore gives the additive threshold
$2\epsilon_{\rm sam}/(\zeta D_{M,N})$ with failure probability at most
$\zeta$ (with the zero-error case obtained directly).

\begin{corollary}[Relative accuracy of the sampler probabilities]
\label{cor:sampler-interface}
Assume Eq.~\eqref{eq:sampler-assumption} and the common dimension conditions
of Theorem~\ref{thm:two-haar-tails}. For $\rho>0$ and $0<\zeta<1$,
\begin{align}
 &\Pp_{U,\bm S}\{|q_{\bm S}-p_{\bm S}|>\rho p_{\bm S}\}
 \nonumber\\
 &\quad\le\min_{i=1,2}\min\!\left\{1,\zeta+
 \frac{2\beta_i\epsilon_{\rm sam}}
      {\rho\zeta D_{M,N}p_i}+\delta_i\right\},
 \label{eq:sampler-interface}
\end{align}
where $\beta_1=B_{K,n}$ and $\beta_2=b_n$.
Optimizing the elementary upper bound gives
\begin{equation}
 \begin{aligned}
 &\Pp\{|q_{\bm S}-p_{\bm S}|>\rho p_{\bm S}\}\\
 &\quad\le\min_{i=1,2}\min\!\left\{1,\delta_i+
 2\sqrt{\frac{2\beta_i\epsilon_{\rm sam}}{\rho D_{M,N}p_i}}\right\}.
 \end{aligned}
 \label{eq:sampler-optimized}
\end{equation}
\end{corollary}
\begin{proof}
Apply Theorem~\ref{thm:fair-comparison} to the estimate $q_{\bm S}$ and
the Markov threshold. For a fixed route put
$A=2\beta_i\epsilon_{\rm sam}/(\rho D_{M,N}p_i)$.
If $0<A<1$, choose $\zeta=\sqrt A$. If $A=0$, let $\zeta\downarrow0$.
If $A\ge1$, the capped bound in Eq.~\eqref{eq:sampler-optimized} is one
and follows without optimization.
\end{proof}

These are statements about the probabilities of a candidate sampler,
not about their efficient numerical evaluation. If a classical sampler
runs in polynomial time using finitely many random bits, approximate
counting supplies an additional probability-estimation interface under its
usual oracle assumptions \cite{Stockmeyer1983,AaronsonArkhipov2013}.
Its own estimation error must then be included in $\gamma(\tau)$ or
combined by a triangle inequality. Nothing in
Corollary~\ref{cor:sampler-interface} assumes that querying $q_S$ is free.

At mean-matched squeezing, Proposition~\ref{prop:sector-scale} shows that
inverse polynomial $\epsilon_{\rm sam}$ can make the Route~1 expression
vanish whenever $B_{K,n}$ is polynomial and $N^2/M\to0$. For example, in
Corollary~\ref{cor:separation-window}, a fixed $\rho>0$ and
$\epsilon_{\rm sam}=o(N^{-1-3/(4c)})$ suffice for this probability
comparison. This makes the normalization explicit, but it does not replace
the ensemble-specific average-case hardness assumption.
\endgroup
\begingroup\section{Transpose Gram structure and congruence equivariance}
\label{sec:structure}

For $Y\in\C^{N\times K}$ define
\begin{equation}
 \begin{aligned}
 \Upsilon(Y)&:=YY^{\T},&
 \Upsilon_K(Y)&:=K^{-1/2}YY^{\T},\\
 \rho_g(A)&:=gAg^{\T}.&&
 \end{aligned}
 \label{eq:hide-projection}
\end{equation}
The maps $\Upsilon$ and $\Upsilon_K$ take values in
$\Sym_N(\C)=\{A:A^{\T}=A\}$.  Their elementary covariance is the structural
reason the proof closes after taking the product.

Here transpose Gram refers specifically to the
ordinary transpose product $YY^{\T}$, not to the Hermitian Gram product
$YY^{*}$.  Consequently $YY^{\T}$ is complex symmetric, but it is not in
general Hermitian or positive semidefinite and is not the usual complex
Wishart matrix.

For every $Y\in\C^{N\times K}$ and $g\in\mathrm{GL}_N(\C)$,
\begin{equation}
 \Upsilon_K(gY)=\rho_g(\Upsilon_K(Y))
 =g\Upsilon_K(Y)g^{\T}.
 \label{eq:hide-covariance}
\end{equation}
Direct multiplication gives
$K^{-1/2}(gY)(gY)^{\T}=g(K^{-1/2}YY^{\T})g^{\T}$.

\endgroup
\begingroup\section{Uniform hiding after the transpose Gram map}
\label{sec:uniform-hiding}

This section proves the matrix hiding theorem stated in
\cref{thm:uniform-product-hiding}.  The theorem compares product matrices, not the underlying
rectangular blocks.  This distinction is essential when $K$ is comparable
with $M$, because entrywise Gaussian approximation of the whole block is then
unavailable.

Figure~\ref{fig:uniform-hiding-geometry} summarizes the four stages: exact
ambient recursion, radial transport to a square COE corner, a summable
centered estimate, and a rectangular comparison.  Their details are in
Appendices~\ref{app:haar-stiefel} through~\ref{app:rectangular-kl}; the
external mathematical inputs are stated in Appendix~\ref{app:external-inputs}.

\begin{figure*}[!tp]
  \centering
  \includegraphics[width=0.98\textwidth]{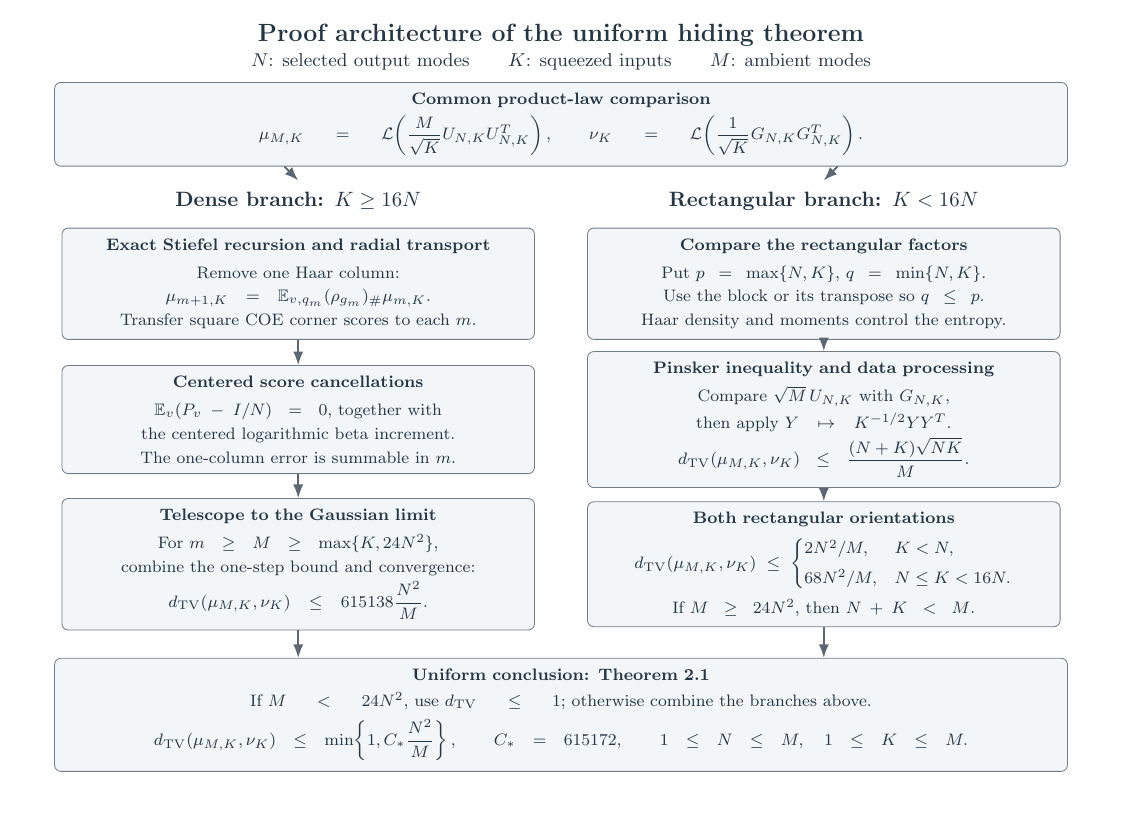}
  \caption{Proof architecture.  The dense branch combines exact one column
  recursion, radial transport, centered COE scores, and telescoping.  The
  rectangular branch uses a Haar block relative entropy bound and data
  processing.  Their combination proves
  Theorem~\ref{thm:uniform-product-hiding}.}
  \label{fig:uniform-hiding-geometry}
\end{figure*}

\subsection{Exact ambient recursion}

Classical Haar Stiefel column deletion and beta radial laws
\cite{BourgadeEtAl2008,Mezzadri2007,ZyczkowskiSommers2001,
ZyczkowskiSommers2000Truncations} yield the exact one column recursion below.
A related positive congruence factorization appears in
Ref.~\cite[Lemma~2.2]{ShouEtAl2026ArbitrarySqueezers}.

For $m\ge\max\{N,K\}$, a \emph{complex Haar {row Stiefel} matrix} is a random
$V\in\C^{N\times m}$ with $VV^*=I_N$ whose law is the invariant probability
measure on
$\{W\in\C^{N\times m}:WW^*=I_N\}$; equivalently, $V$ consists of the first
$N$ rows of a Haar matrix in $\mathrm U(m)$
~\cite{Mezzadri2007,ZyczkowskiSommers2000Truncations}.  Let $V_K$ contain the
first $K$ columns of $V$, and put
\begin{equation}
 {Z_{m,K}}:=\frac{m}{\sqrt K}V_KV_K^{\T},
 \qquad \mu_{m,K}:=\Law({Z_{m,K}}).
 \label{eq:hide-mu-definition}
\end{equation}
Thus $\mu_{m,K}$ is a probability measure on $\C^{N\times N}$ supported on
the complex symmetric matrices.  If $\Lambda$ is a random probability
measure on this ambient space, we write $\E\Lambda$ for its barycenter, defined by
$(\E\Lambda)(E):=\E[\Lambda(E)]$ for every Borel set $E$.
Also define
\begin{equation}
 B_K:=\frac{1}{\sqrt K}G_{N,K}G_{N,K}^{\T},
 \qquad \nu_K:=\Law(B_K).
 \label{eq:hide-nu-definition}
\end{equation}
{For a measurable map $f$, the symbol $f_\#\eta$ denotes the
pushforward of a measure $\eta$.  In particular, $\rho_g(A)=gAg^{\T}$ was
defined in Eq.~\eqref{eq:hide-projection}.}

\begin{lemma}[{One column} {Haar Stiefel} recursion]
\label{lem:hide-recursion}
For $1\le N\le K\le m$, set $g_m:=\exp(a_mI+b_mP_v)$, and let
$A_m\sim\mu_{m,K}$ be independent of $(v,q_m)$.  Then the recursion, written
directly as an identity of probability measures on $\C^{N\times N}$, is
\begin{equation}
 \mu_{m+1,K}
 =\E_{v,q_m}\big[(\rho_{g_m})_\#\mu_{m,K}\big]
 =\Law(g_mA_mg_m^{\T}).
 \label{eq:hide-recursion-measure}
\end{equation}
Equivalently, expanding the barycenter and pushforward, for every Borel set
$E\subset\C^{N\times N}$,
\begin{equation}
 \begin{aligned}
 \mu_{m+1,K}(E)
 &=\E_{v,q_m}\!\left\{
   \big[(\rho_{g_m})_\#\mu_{m,K}\big](E)\right\}\\
 &=\E_{v,q_m}\!\int
   \1_E(g_mAg_m^{\T})\,\mu_{m,K}(\dd A).
 \end{aligned}
 \label{eq:hide-recursion}
\end{equation}
where $P_v=vv^*$, $v$ is uniform on the unit sphere of $\C^N$, and
\begin{equation}
 \begin{aligned}
 a_m&=\frac12\log\!\left(1+\frac1m\right),&
 b_m&=\frac12\log q_m,\\
 q_m&\sim\operatorname{Beta}(m-N+1,N).&&
 \end{aligned}
 \label{eq:hide-ab}
\end{equation}
The inner integral samples $A\sim\mu_{m,K}$ independently of $(v,q_m)$,
whereas the outer expectation samples {$v$ and $q_m$
independently of each other}.  Equivalently,
Eq.~\eqref{eq:hide-recursion-measure} is a positive congruence Markov step
with scalar part $a_mI$ and rank one part $b_mP_v$.
\end{lemma}

\begin{proof}
Take an $N$ by $m+1$ complex Haar {row Stiefel} matrix and
delete its last column $u$.  Let $R$ be the remaining $N$ by $m$ matrix.
{Row orthonormality gives
\[
 RR^*+uu^*=I_N,\qquad RR^*=I_N-uu^*.
\]}
{Writing $u=rv$, unitary invariance makes $v$ uniform and
independent of $r$, while the classical truncated Haar radial law
\cite[Sec.~III, Eq.~(7)]{ZyczkowskiSommers2000Truncations} gives}
\[
 r^2\sim\operatorname{Beta}(N,m+1-N),\qquad
 v\sim\operatorname{Unif}(S^{2N-1}).
\]
{Since $r^2<1$ almost surely, $I_N-uu^*$ is positive definite.
Define
\[
 \begin{aligned}
 W&:=(I_N-uu^*)^{-1/2}R,\\
 WW^*&=(I_N-uu^*)^{-1/2}RR^*(I_N-uu^*)^{-1/2}=I_N.
 \end{aligned}
\]
Then $R=(I_N-uu^*)^{1/2}W$.  For every fixed $Q\in\mathrm U(m)$,
\[
 \begin{bmatrix}R&u\end{bmatrix}
 \begin{pmatrix}Q&0\\0&1\end{pmatrix}
 =\begin{bmatrix}RQ&u\end{bmatrix}
 \law\begin{bmatrix}R&u\end{bmatrix},
\]
by right unitary invariance of the original Haar row Stiefel law.
The last column is unchanged, so $(R,u)\law(RQ,u)$ and
$\Law(RQ\mid u)=\Law(R\mid u)$ almost surely.
Since $WW^*=I_N$, the conditional law of $W$ given $u$ is a
right-unitarily invariant probability law on the row Stiefel manifold,
hence the Haar row Stiefel law.  This conditional law does not depend on
$u$, so $W$ is independent of $u$.}

{Returning to the radial variable, set $q_m=1-r^2$.
The beta law of $r^2$ above gives
$q_m\sim\operatorname{Beta}(m-N+1,N)$.}
If $R_K$ and $W_K$ denote the first $K$ columns
{of $R$ and $W$, respectively}, then
\[
 \begin{aligned}
 \frac{m+1}{\sqrt K}R_KR_K^{\T}
 &=g_m\left(\frac m{\sqrt K}W_KW_K^{\T}\right)g_m^{\T},\\
 g_m&=\sqrt{\frac{m+1}{m}}(I-uu^*)^{1/2}
      =e^{a_mI+b_mP_v}.
 \end{aligned}
\]
Equation~\eqref{eq:hide-covariance} now gives the measure identity
Eq.~\eqref{eq:hide-recursion-measure}; expanding its barycenter gives
Eq.~\eqref{eq:hide-recursion}.
\end{proof}

For fixed $1\le N\le K$, every starting index $m_0$, and every
$\varepsilon>0$, Appendix~\ref{app:haar-stiefel} proves
\begin{equation}
 \exists n\in\mathbb N:\qquad
 d_{\rm TV}(\mu_{m_0+n,K},\nu_K)\le\varepsilon.
 \label{eq:hide-target-convergence}
\end{equation}
The dense argument uses the ambient recursion because its first-order drift
cancels exactly.

\subsection{Radial decomposition and orbital commutation}

{By the unitary invariance of the Haar and Gaussian blocks
and Eq.~\eqref{eq:hide-covariance}, both $\mu_{m,K}$ and
$\nu_K$ are invariant under unitary congruence.  Conditional on the Takagi
singular values, each has the same angular law induced by Haar measure on the
corresponding orbit.  Only their radial Takagi distributions can therefore
differ, motivating the radial factorization below.}

{Radial and angular} factorizations and matrix beta variables are standard in
multivariate Wishart theory, while complex projective moments and
multiplicative spherical transforms are standard Haar tools
~\cite{OlkinRubin1964,CollinsSniady2006,KieburgKosters2019}.  Their role here is
different: radial factorization transports the resulting single column
probability bound from one square COE corner to every ambient dimension,
and projective averaging isolates the
centered congruence directions that act on the transpose Gram product.

For a probability law $\xi$ on $\mathrm{GL}_N(\C)$ and a probability
law $\eta$ on $\C^{N\times N}$, let $G\sim\xi$ and $X\sim\eta$ be independent and
write
\begin{equation}
 \mathcal A_\xi\eta:=\Law(GXG^{\T}).
 \label{eq:hide-congruence-action}
\end{equation}
Thus $\mathcal A_\xi$ is a Markov congruence action.  If $\lambda$ is a
probability law on $\mathrm{GL}_N(\C)$ supported on
$\operatorname{Herm}^{++}_N(\C)$, we also write
$\mathcal T_\lambda:=\mathcal A_\lambda$.  Here ``orbital'' refers to
averaging over the unitary orbit of a fixed rank one direction.  With
$v\sim\operatorname{Unif}(S^{2N-1})$, define
\begin{equation}
 \begin{aligned}
 P_v&=vv^*,\qquad Q_v=P_v-I/N,\\
 \mathcal K_s\eta&=\E_v(\rho_{\exp(sQ_v)})_\#\eta.
 \end{aligned}
 \label{eq:hide-orbital-kernel}
\end{equation}

\begin{lemma}[Spherical commutation and variation transport]
\label{lem:hide-commute}
Let $N\ge1$.  Let $\eta$ be a probability law on $\C^{N\times N}$ invariant
under unitary congruence, and let $\lambda$ be a probability law on
$\mathrm{GL}_N(\C)$, supported on $\operatorname{Herm}^{++}_N(\C)$ and
invariant under unitary conjugation.  Then
$\mathcal T_\lambda\eta$ is invariant under unitary congruence and, for every
$s\in\R$,
\begin{equation}
 \mathcal K_s\mathcal T_\lambda\eta
 =\mathcal T_\lambda\mathcal K_s\eta.
 \label{eq:hide-commute}
\end{equation}
The operator $\mathcal T_\lambda$ has a bounded linear extension to finite
signed Borel measures, and its variation norm satisfies
\begin{equation}
 \bigl\|\mathcal T_\lambda\sigma\bigr\|_{\rm var}
 \le \|\sigma\|_{\rm var}.
 \label{eq:hide-commute-contraction}
\end{equation}
Let $I\subset\R$ be open and let $r\in\mathbb N_0$.  If
$s\mapsto\mathcal K_s\eta$ is $C^r$ on $I$ in the variation norm, then, for
every $0\le j\le r$ and $s\in I$,
\begin{equation}
 \begin{aligned}
  \partial_s^j(\mathcal K_s\mathcal T_\lambda\eta)
  &=\mathcal T_\lambda\!\left[\partial_s^j(\mathcal K_s\eta)\right],\\
  \bigl\|\partial_s^j(\mathcal K_s\mathcal T_\lambda\eta)\bigr\|_{\rm var}
  &\le \bigl\|\partial_s^j(\mathcal K_s\eta)\bigr\|_{\rm var}.
 \end{aligned}
 \label{eq:hide-commute-derivative}
\end{equation}
\end{lemma}

\begin{proof}
The complete proof is given in Appendix~\ref{app:haar-stiefel},
Sec.~\ref{subsec:hide-commute-direct-proof}, where the commutation,
variation contraction, and derivative transport in
Eqs.~\eqref{eq:hide-commute}--\eqref{eq:hide-commute-derivative}
are established in turn.
Figure~\ref{fig:haar-lift-transport} summarizes its Haar lift,
convolution, Markov contraction, and derivative transport steps.
\end{proof}

\begin{figure*}[!tp]
\centering
\includegraphics[width=0.92\textwidth]{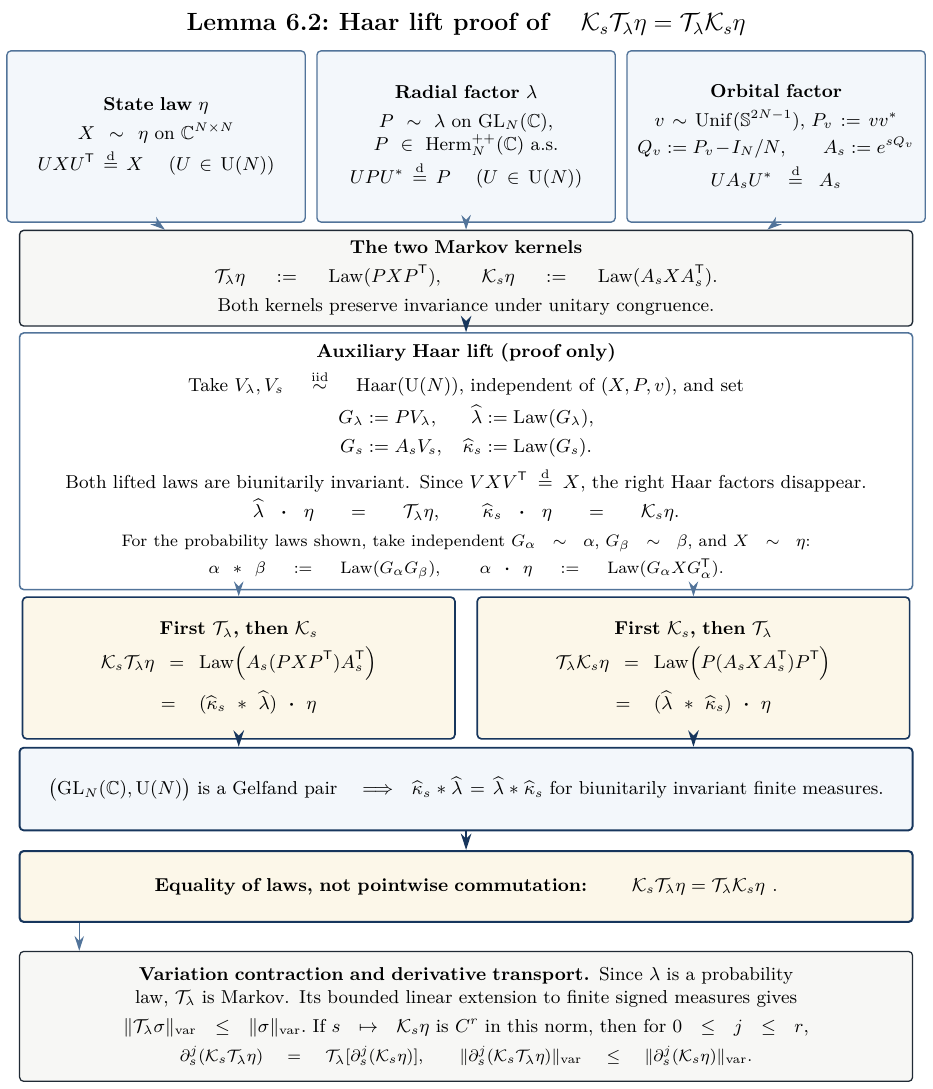}
\caption{Haar lift and variation transport in Lemma~\ref{lem:hide-commute}.
The auxiliary Haar factors identify the two congruence kernels with
convolutions of biunitarily invariant laws. Their commutation and Markov
contraction give the stated transport of variation derivatives.}
\label{fig:haar-lift-transport}
\end{figure*}

For the square COE base used below, Appendix~\ref{app:coe-density} constructs
the transported density derivatives through order four as $L^1$ curves.  The
variation norm of a measure with density is its $L^1$ norm, so those derivatives
supply the $C^4$ variation regularity required by Lemma~\ref{lem:hide-commute}.

We make the radial transport used in the proof explicit.  Let
\[
 \mathsf R_{m,N}\eta
 :=\E_{v,q_m}(\rho_{\exp(a_mI+b_mP_v)})_\#\eta
\]
be the one column Markov kernel of Lemma~\ref{lem:hide-recursion}.  Define the
finite chain by
\[
 \mathcal T_{K,K}:=\operatorname{Id},\qquad
 \mathcal T_{m+1,K}:=\mathsf R_{m,N}\mathcal T_{m,K}\quad(m\ge K).
\]
Iterating the exact recursion gives
\begin{equation}
 \mu_{m,K}=\mathcal T_{m,K}\mu_{K,K},\qquad m\ge K.
 \label{eq:hide-radial}
\end{equation}
The base law in this chain is
\begin{equation}
 \mu_{K,K}=\Law(\sqrt K\,{C_{N,K}}),
 \label{eq:hide-base}
\end{equation}
where $W$ is Haar in $\mathrm U(K)$,
$\operatorname{COE}(K):=\Law(WW^{\T})$, and
{$C_{N,K}:=(WW^{\T})_{[N],[N]}$} is its upper $N$ by
$N$ corner.
{Each single column kernel in $\mathcal T_{m,K}$ is a mixture
of positive definite matrices with conjugation invariant law, so
Lemma~\ref{lem:hide-commute} applies successively.  This is commutation of
averaged laws, not pointwise commutation.

{We use the following consequence for probability laws.  For
$1\le N\le K\le m$ and every integer $r\ge N$, let
$\mathsf R_{r,N}$ be the single column kernel with
$q_r\sim\operatorname{Beta}(r-N+1,N)$.  Both beta parameters are positive
in this range.  The concrete radial commutation theorem gives}
\begin{equation}
 \mathsf R_{r,N}\mu_{m,K}
 =\mathcal T_{m,K}(\mathsf R_{r,N}\mu_{K,K}).
 \label{eq:hide-radial-one-column-commute}
\end{equation}
Every $\mathcal T_{m,K}$ is Markov.  Its eventwise total variation
contraction says that, for
every $\delta\ge0$,
\begin{equation}
 \begin{aligned}
 &d_{\rm TV}\!\left(\mu_{K,K},\mathsf R_{r,N}\mu_{K,K}\right)\le\delta\\
 &\quad\Longrightarrow\quad
 d_{\rm TV}\!\left(\mu_{m,K},\mathsf R_{r,N}\mu_{m,K}\right)\le\delta .
 \end{aligned}
 \label{eq:hide-radial-TV}
\end{equation}
Thus the score calculation is needed only for the square COE base law; its
resulting single column probability bound, rather than a signed derivative
measure, is what is transported to every ambient dimension.}

\subsection{Centered COE score estimate}

The square COE corner is the base law for the radial transport.  Its density
and {beta prime} representation are classical; the new technical estimate is
explicit control of the \emph{averaged centered} score.  The
projector $Q_v=P_v-I/N$ cancels the {first order} orbital response, and the
remaining bounds of second and {third order} have precisely the sizes needed for
a summable ambient step.  We have not found these simultaneous orbital and
scalar {first order} cancellations in earlier {Haar product} hiding arguments.
The density and boundary calculation is in
Appendix~\ref{app:coe-density}; the projective tensor contractions and moment
calculations are in Appendix~\ref{app:projective-scores}.

For a Borel set $E\subset\C^{N\times N}$ define the central and centered orbital
event paths at the square base law by
\begin{equation}
 \begin{aligned}
 S_E(t)&:=((\rho_{e^{tI}})_\#\mu_{K,K})(E),\\
 O_E(s)&:=(\mathcal K_s\mu_{K,K})(E).
 \end{aligned}
 \label{eq:hide-event-paths}
\end{equation}

\begin{lemma}[Uniform event score bounds]
\label{lem:hide-scores}
For $N\ge1$, $K\ge16N$, every Borel $E$, every $t\in\R$, and every
$\abs{s}\le1/N$, the paths $S_E$ and $O_E$ are respectively $C^2$ and
$C^4$, and
\begin{gather}
 \abs{S_E'(0)}\le 5N,\qquad
 \abs{S_E''(t)}\le 44N^2,
 \label{eq:hide-scale-score}\\[2pt]
 O_E'(0)=0,\qquad
 \abs{O_E''(0)}\le573,\qquad
 \abs{O_E'''(s)}\le306840N.
 \label{eq:hide-orbital-score}
\end{gather}
\end{lemma}

\begin{proof}
Appendix~\ref{app:coe-density} proves the required regularity.  Appendix
~\ref{app:projective-scores}, Sec.~\ref{subsec:hide-scores-direct-proof},
derives all five evaluated bounds; Fig.~\ref{fig:event-score-flow} shows
how its density scores, inverse Wishart moments, and projective
averages fit together.  The finite coefficient identities and closing
inequalities follow by exact rational arithmetic from the four cited
mathematical results A1 through A4 stated in
Appendix~\ref{app:external-inputs}.
\end{proof}

\begin{figure*}[!tp]
\centering
\includegraphics[width=\textwidth]{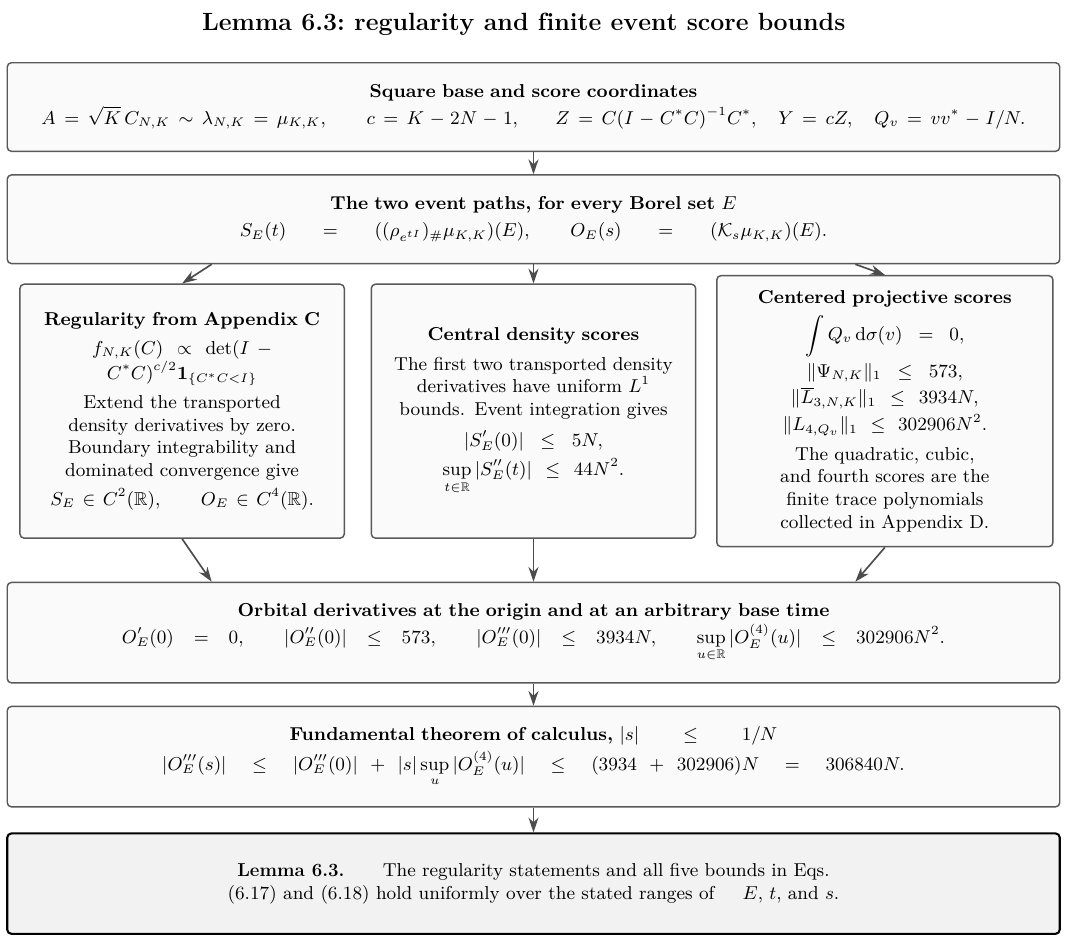}
\caption{The eventwise score bounds of Lemma~\ref{lem:hide-scores}, for
$K\ge16N$. Central density scores give the scalar estimates, while
centered projective contractions give the orbital estimates. The bound
away from the origin follows by integrating the fourth derivative over
$|s|\le1/N$.}
\label{fig:event-score-flow}
\end{figure*}

For a fixed beta sample in the single column update, the central step merely
replaces $E$ in the orbital path by its measurable inverse image.  Hence the
uniform event bounds in Lemma~\ref{lem:hide-scores} apply to the correlated
scalar and orbital path using the same beta sample and give the square base
one column probability estimate.  By Eqs.~\eqref{eq:hide-radial} and
\eqref{eq:hide-radial-one-column-commute}, the two ambient laws are the images
of the two square base laws under the same Markov operator
$\mathcal T_{m,K}$.  Total variation contraction for that operator gives
Eq.~\eqref{eq:hide-radial-TV} and transports the completed probability
estimate to $\mu_{m,K}$.

\subsection{One step estimate and dense telescope}

Two first order cancellations drive the new one step estimate.  Projective
centering gives $\E_vQ_v=0$ and hence $O_E'(0)=0$ for every Borel $E$, while
the scalar beta cumulants center $c_m=a_m+b_m/N$ so that its mean is of order
$N/m^2$ rather than $1/m$.  The Taylor expansion and final telescope are standard
replacement devices, analogous in spirit to a Lindeberg telescope
~\cite{Chatterjee2006}, once these cancellations and the averaged score bounds
are available.

Split the exponent in Eq.~\eqref{eq:hide-recursion} as
\begin{equation}
 a_mI+b_mP_v=c_mI+b_mQ_v,\qquad c_m=a_m+b_m/N.
 \label{eq:hide-split}
\end{equation}
The exact {log beta} cumulants and the scalar and orbital Taylor
remainders, and the {exceptional event} tail are evaluated in
Appendix~\ref{app:projective-scores}.  Together with
Lemma~\ref{lem:hide-scores}, those calculations give the following {one step}
statement, including its summable remainder.

The estimates enter with fixed coefficients, so the result can be stated
with the universal constant $C_\ast$ rather than asymptotic notation.
\begin{proposition}[Dense {one step} estimate]
\label{prop:hide-one-step}
For $N\ge1$, $K\ge16N$, and $m\ge\max\{K,24N^2\}$,
\begin{equation}
 d_{\rm TV}(\mu_{m+1,K},\mu_{m,K})
 \le C_\ast\frac{N^2}{m(m+1)}.
 \label{eq:hide-one-step}
\end{equation}
\end{proposition}
\begin{proof}
The complete eventwise proof is given in
Appendix~\ref{app:projective-scores},
Sec.~\ref{subsec:hide-one-step-direct-proof}.  It shows explicitly how
Lemma~\ref{lem:hide-scores} controls the scalar Taylor term and the centered
orbital Taylor term while retaining the same beta sample, and then combines
those terms with the exceptional event estimate and radial total variation
contraction.  Figure~\ref{fig:one-step-taylor-flow} displays this
three term decomposition.
\end{proof}

\begin{figure*}[!tp]
\centering
\includegraphics[width=\textwidth]{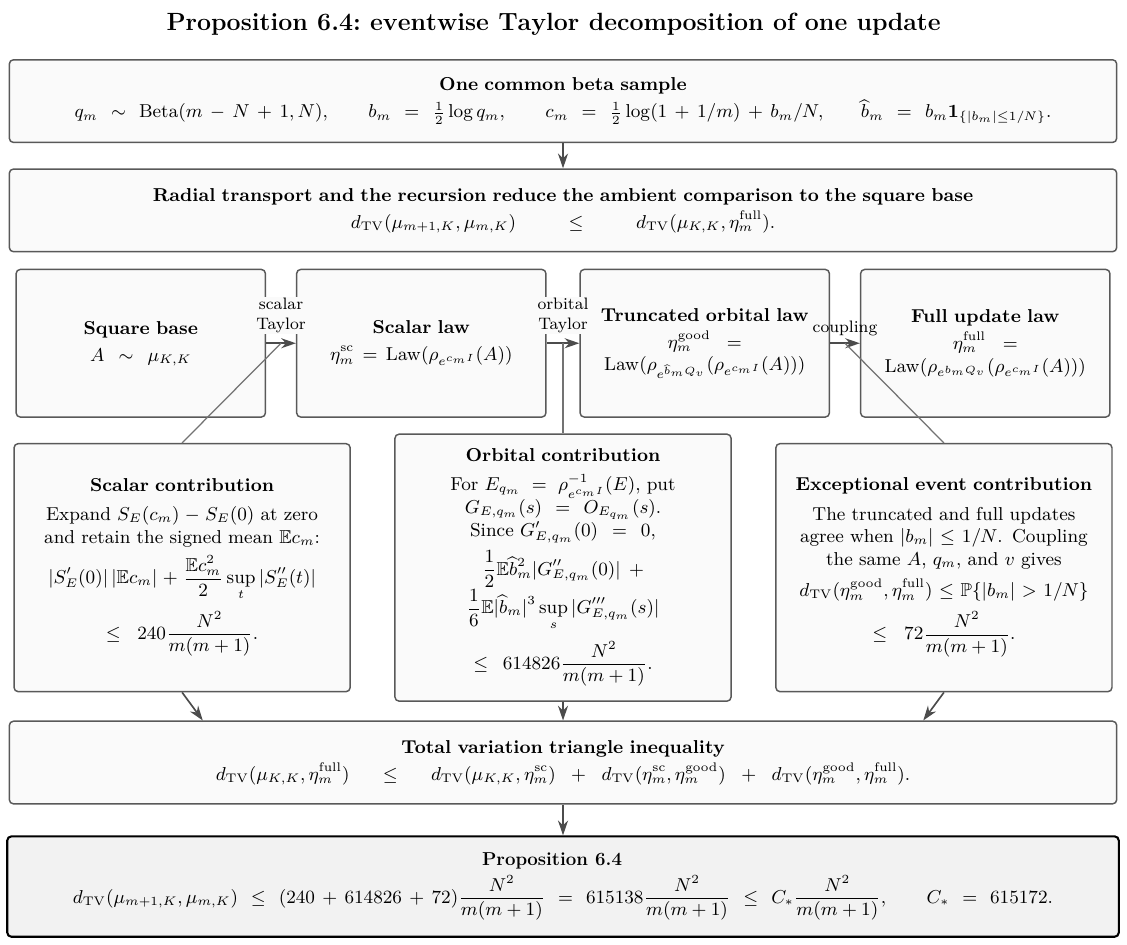}
\caption{The one-column comparison in Proposition~\ref{prop:hide-one-step}.
The same beta sample drives the scalar and centered orbital increments.
The scalar, orbital, and exceptional-event contributions sum to
$615138N^2/[m(m+1)]$, which is bounded by the stated estimate with
$C_\ast=615172$.}
\label{fig:one-step-taylor-flow}
\end{figure*}
The passage to the Gaussian target can be kept finite.  For $R\ge M$, the
triangle inequality gives
\[
 \begin{aligned}
 d_{\rm TV}(\mu_{M,K},\nu_K)
 &\le \sum_{m=M}^{R-1}
 d_{\rm TV}(\mu_{m+1,K},\mu_{m,K})\\
 &\quad+d_{\rm TV}(\mu_{R,K},\nu_K).
 \end{aligned}
\]
Given $\varepsilon>0$, choose $R$ from
Eq.~\eqref{eq:hide-target-convergence} so that the last term is at most
$\varepsilon$.  Since
$\sum_{m=M}^{R-1}[m(m+1)]^{-1}=M^{-1}-R^{-1}$,
Eq.~\eqref{eq:hide-one-step} gives the dense estimate after letting
$\varepsilon\downarrow0$:
\begin{equation}
 d_{\rm TV}(\mu_{M,K},\nu_K)\le C_\ast\frac{N^2}{M}
 \label{eq:dense-product-hiding}
\end{equation}
whenever the dense branch hypotheses above hold.

\subsection{Finite rectangular entropy branch}

The complementary branch compares the whole rectangular block before taking
its transpose Gram product.  {The required Gaussian
approximation for strict size Haar blocks is known
\cite{Jiang2009HaarBlocks}, but the cited result does not provide an explicit
universal constant for the finite upper bound needed here.  We therefore
derive that constant directly from the exact block density.  An elementary
second order bound for the logarithmic determinant, together with the exact
first and second Haar Gram moments, yields the stated relative entropy
estimate.  Pinsker's inequality converts this estimate to total variation,
and data processing under the transpose Gram map gives the matrix product
comparison.}

\begin{lemma}[Finite rectangular comparison]
\label{lem:hide-sparse}
For probability laws $P\ll Q$, use the oriented convention
$D_{\rm KL}(P\Vert Q)=\int\log(\dd P/\dd Q)\,\dd P$.
For random matrices $X,Y$, the abbreviation $D_{\rm KL}(X\Vert Y)$ below
means $D_{\rm KL}(\Law(X)\Vert\Law(Y))$.
Suppose
\begin{equation}
 1\le q\le p,\qquad p+q<M.
 \label{eq:hide-sparse-hyp}
\end{equation}
If $U_{p,q}$ is an upper $p$ by $q$ block of a Haar unitary and $G_{p,q}$ is
standard complex Gaussian, then
\begin{equation}
 D_{\rm KL}(\sqrt M U_{p,q}\Vert G_{p,q})
 \le\frac{3pq(p+q)^2}{4M^2},
 \label{eq:hide-sparse-KL}
\end{equation}
and
\begin{equation}
 d_{\rm TV}(\sqrt M U_{p,q},G_{p,q})
 \le\frac{(p+q)\sqrt{pq}}{M}.
 \label{eq:hide-finite-sparse}
\end{equation}
Ordinary transposition transports the estimate to the corresponding wide
$q$ by $p$ Haar and Gaussian blocks; {applying the measurable map
$Y\mapsto YY^{\T}$ to those wide blocks does not increase the total variation distance.}
\end{lemma}

\begin{proof}
Appendix~\ref{app:rectangular-kl} starts from the exact density ratio
in Eq.~\eqref{eq:rectangular-rn-density} and the entropy identity
in Eq.~\eqref{eq:rectangular-kl-identity}.  The logarithmic determinant bound
and the two Haar Gram moments give the cancellation in
Eq.~\eqref{eq:rectangular-kl-cancellation}; Pinsker's inequality and data
processing then give the total variation and transpose Gram conclusions.
\end{proof}

\subsection{Proof of Theorem~\ref{thm:uniform-product-hiding}}
\label{subsec:uniform-theorem-proof}

\begin{proof}

Assume $1\le N\le M$ and $1\le K\le M$, as in the theorem.
The dense calculation has
coefficient $615138$. The retained common constant
$C_\ast=615172$ is larger than this dense coefficient, the rectangular
coefficients $68$ and $2$, and the trivial-range threshold $24$.
If $M<24N^2$, the right hand side of
Eq.~\eqref{eq:uniform-product-hiding} is one because the chosen value of
{$C_\ast$} dominates the fixed threshold, so the universal bound
$d_{\rm TV}\le1$ proves the claim.

Now suppose $M\ge24N^2$. If $K<N$, put $p=N$ and $q=K$.
Then $1\le q\le p$ and $N+K<2N\le24N^2\le M$, so
Lemma~\ref{lem:hide-sparse} applies. Data processing under
$Y\mapsto K^{-1/2}YY^{\T}$ gives
\begin{align*}
 d_{\rm TV}(\mu_{M,K},\nu_K)
 &\le\frac{(N+K)\sqrt{NK}}{M}\\
 &\le 2\frac{N^2}{M}\le C_\ast\frac{N^2}{M}.
\end{align*}
Here $\sqrt{NK}\le N$ because $K<N$.

In the dense branch $K\ge16N$, every integer
$m\ge M$ satisfies $m\ge\max\{K,24N^2\}$.  Thus
Proposition~\ref{prop:hide-one-step} applies to every term of the ambient
telescope and supplies Eq.~\eqref{eq:hide-one-step}.  Combining those terms
with the target convergence in Eq.~\eqref{eq:hide-target-convergence} gives
the dense estimate in Eq.~\eqref{eq:dense-product-hiding}.  Its assumptions
are therefore exactly $1\le N\le K\le M$, $K\ge16N$, and $M\ge24N^2$.

In the remaining branch $N\le K<16N$, put $p=K$ and $q=N$.
The theorem assumptions give $1\le q\le p$.  Moreover,
$K+N<17N\le24N^2\le M$, so $p+q<M$.  These are precisely the dimensional
hypotheses of Lemma~\ref{lem:hide-sparse}.  Its total variation estimate,
followed by transposition and the transpose Gram data processing map, gives
\begin{equation}
 d_{\rm TV}(\mu_{M,K},\nu_K)
 \le\frac{(K+N)\sqrt{KN}}{M}
 \le 68\frac{N^2}{M}.
 \label{eq:hide-sparse-stitch}
\end{equation}
Since $2,68<C_\ast$, the $K<N$ estimate, the dense estimate in
Eq.~\eqref{eq:dense-product-hiding}, and the rectangular estimate in
Eq.~\eqref{eq:hide-sparse-stitch} together give
$d_{\rm TV}(\mu_{M,K},\nu_K)\le C_\ast N^2/M$.  Combining this estimate with
the universal probability bound $d_{\rm TV}\le1$ gives exactly the minimum in
Eq.~\eqref{eq:uniform-product-hiding}; when $M<24N^2$, that minimum is one and
the same universal bound applies.  Finally, the measurable bijection
$A\mapsto\sqrt K\,A$ sends the two normalized laws to the two laws in
Eq.~\eqref{eq:uniform-product-hiding-unscaled}.  Total variation is invariant
under a common measurable bijection, so the unscaled statement follows.
\end{proof}

The proof reveals why the uniform result is possible.  The sparse branch
approximates the rectangular Haar block itself and gives a uniform
$O(N^2/M)$ bound throughout $K<16N$. The dense branch works after $\Upsilon_K$, where
Eq.~\eqref{eq:hide-covariance} turns the ambient change into a congruence
perturbation.  The centered projective score bounds then make the one column
error summable, and the final rate loses all dependence on $K$.
\endgroup
{Further applications to measurable observables, jointly preselected patterns, and order statistics are in Appendix~\ref{sec:applications}.}
\begingroup\section{Relation to earlier results and remaining questions}
\label{sec:discussion}
Theorem~\ref{thm:uniform-product-hiding} and
Corollary~\ref{cor:all-input-hiding} resolve Conjecture~1 (Formal) of
Ehrenberg et al., Supplemental Eq.~(S62), with the conjectured
$m\ge n^2/\delta$ scaling and all $1\le k\le m$
\cite{EhrenbergEtAl2025Transition}.
Shou et al.\ prove hiding in a broader reduction framework, including a
product-law comparison and approximate generation of prescribed instances
\cite[Theorem~1.1 and Remark~1.1]{ShouEtAl2026ArbitrarySqueezers}.
Our finite estimate sharpens the sufficient ambient-size dependence for
the product law. The two-route comparison explains why retaining the finite
Gaussian factor can be useful before the independent symmetric limit is
quantitatively controlled. The statements are complementary, not a claim
that one method dominates in every parameter regime.

Three questions remain particularly relevant. First, can an efficient
conditional instance-generation procedure be obtained with the same
$M=O(N^2/\delta)$ error scale? Distributional closeness does not itself
answer that algorithmic question. Second, what is the actual total
variation behavior of the independent symmetric reference when
$K$ is between $N^2/\log N$ and $N^2$? The nonvanishing envelope used here
is not a matching lower bound. Third, can local anticoncentration remain
polynomial for substantially smaller $K$? The companion proves a sufficient
$K\gtrsim N^2/\log N$ regime, not its necessity.

The average-case complexity of the resulting hafnian estimation problem
also remains separate. The original GBS proposals distinguish approximate
estimation assumptions from their analytic lower-tail inputs
\cite{HamiltonEtAl2017GBS,KruseEtAl2019GBS}; related formulations can
place the remaining conjecture directly at additive accuracy
\cite{DeshpandeEtAl2022,ShouEtAl2026ArbitrarySqueezers}.
Neither the hiding theorem nor the two-route comparison is a proof of
average-case $\#\mathrm P$ hardness or of quantum advantage for a noisy
experiment. Unequal squeezing, losses, displacements, and labels sampled
adaptively from the circuit require additional analysis. Classical
simulation remains a separate benchmark
\cite{QuesadaArrazola2020ExactSimulation,QuesadaEtAl2022QuadraticSpeedup,BulmerEtAl2022Boundary}.

The companion structure follows the mathematics. The present paper proves
finite Haar replacement and develops its optical and estimation
consequences. The companion proves the Gaussian local theorems from
Fourier cofactor compression and a preserved-coordinate Wishart identity.
Together they give the explicit finite-interferometer lower tails in
Theorem~\ref{thm:two-haar-tails}, with no appeal to an unproved
anticoncentration conjecture in the stated parameter range.
\endgroup

\begingroup
\section*{Author contributions and use of artificial intelligence}
Hongru Zhao is responsible for the mathematical content, its presentation,
and the accompanying source materials. OpenAI's GPT-5.6 Sol and GPT-6 Astra Ultra models assisted with
literature searches, manuscript restructuring and editing, algebraic checks, figure generation,
and preparation of LaTeX and Lean source materials.

\section*{Data and code availability}
\begingroup
No experimental data were created or analyzed. Verification of
Theorem~\ref{thm:uniform-product-hiding},
Corollary~\ref{cor:all-input-hiding}, and Route~1 of
Theorem~\ref{thm:fair-comparison} is maintained in the GitHub
repository~\cite{Zhao2026UniformHidingGitHub}.
Route~1 uses the companion anticoncentration
formalization~\cite{Zhao2026ComplexGramHafniansGitHub}.
The broader formalization and paper-to-code correspondence are archived
on Zenodo~\cite{Zhao2026UniformHidingLean}.
The hiding and Route~1 formal proofs depend on the four literature inputs in
Appendix~\ref{app:external-inputs}. Coverage of auxiliary results is
partial, and the Route~2 hiding input is not formalized. The accompanying
verification documentation specifies the exact statements covered.
\endgroup

\endgroup

\appendix
\begingroup\clearpage
\begingroup

\section{The four literature axioms used in the Lean formalization}
\label{app:external-inputs}

This appendix documents the four literature axioms A1 through A4 accepted
in the Lean formalization.  We import these mathematical results rather
than reprove them inside Lean.  The purpose is to justify precisely those
imported assumptions by comparison with the cited sources.
\textbf{Part (1) of each subsection is the exact mathematical translation
of the Lean axiom's hypotheses and conclusions in this paper's notation,
not Lean code or a quotation from the source.}
Parts (2)--(4) give, respectively, the source statement, the justification
for any differences, and the notation dictionary.  These source-to-axiom
justifications explain why the Lean assumptions are legitimate; they are
not presented as separately formalized proofs.  The four subsections
correspond to the four axioms.

All finite-dimensional spaces carry their Borel structures.  For a
measurable map $f$, $f_\#\mu$ denotes the pushforward measure,
$(f_\#\mu)(B)=\mu(f^{-1}(B))$.  A permutation invariant map on
$\mathbb R^r$ satisfies $F(x\circ\pi)=F(x)$ for every coordinate
permutation $\pi$.  Table~\ref{tab:external-component-audit} summarizes the
four inputs; their full quantified statements follow.

{\small
\setlength{\tabcolsep}{4pt}
\renewcommand{\arraystretch}{1.16}
\begin{table}[!ht]
\color{black}
\caption{Imported mathematical content and the source-to-axiom conversion.}
\label{tab:external-component-audit}
\centering
\begin{tabular}{@{}>{\raggedright\arraybackslash}p{.08\textwidth}>{\raggedright\arraybackslash}p{.27\textwidth}>{\raggedright\arraybackslash}p{.27\textwidth}>{\raggedright\arraybackslash}p{.32\textwidth}@{}}
\hline
Axiom & Imported conclusion & Cited source & Conversion explained below\\
\hline
A1 & Haar principal COE corner equals its normalized determinant density.
& Friedman--Mello, Eqs.~(1.2), (3.7), Appendix~1.
& Haar transpose convention; independent symmetric coordinates; normalization.\\
A2 & One positive finite radial constant, global spectrum measurability,
and all measurable symmetric-test pushforwards.
& FitzGerald--Warren, Sec.~6, p.~165; An--Wang--Yan, Thm.~4.2 and remark.
& Flat Jacobian to measurable integration; selected spectrum; arbitrary targets.\\
A3 & Global measurability of squared GSVD coordinates and their
beta Jacobi law under all measurable symmetric tests.
& Edelman--Sutton, Definition~1.1 and Proposition~1.2.
& Hermitian representative; values on singular samples; measurable selector.\\
A4 & Paired Wishart and inverse Wishart tensor moments for a supplied
Wishart probability law.
& Matsumoto, {arXiv v3,} Thm.~3; Eq.~(4.10), Lemma~5, and Eq.~(5.2).
& Laplace-transform law convention; inverse matching-Gram versus zonal
Weingarten coefficient.\\
\hline
\end{tabular}
\end{table}}

\subsection{COE principal block law (A1)}

\paragraph{(1) Exact mathematical translation of the Lean axiom.}
Let $N,K$ be integers with $N\ge1$ and $2N\le K$, and let $h_K$ be Haar
probability measure on $\mathrm U(K)$.  Set
\[
 C_{N,K}(U)=[UU^{\T}]_{1:N,1:N},\qquad
 \mathcal S_N=\{C\in\mathbb C^{N\times N}:C^{\T}=C\},\qquad
 \dd C=\prod_{i\le j}\dd\Re C_{ij}\,\dd\Im C_{ij}.
\]
Embed this coordinate measure in the full square-matrix space, supported
on $\mathcal S_N$, and define
\[
 w_{N,K}(C)=
 \begin{cases}
 \det(I_N-C^*C)^{(K-2N-1)/2},& I_N-C^*C>0,\\
 0,&\text{otherwise},
 \end{cases}
 \qquad Z_{N,K}=\int_{\mathcal S_N}w_{N,K}(C)\,\dd C.
\]
The imported law identity is
\begin{equation}
 (C_{N,K})_\#h_K
   =Z_{N,K}^{-1}w_{N,K}(C)\,\dd C.
 \label{eq:external-A1}
\end{equation}
The right-hand side is defined by normalizing the raw measure by its own
total mass.  No closed gamma or pi prefactor is part of the axiom, and no
separate differentiability or moment conclusion is imported.  On the
support, the determinant is positive and real, so its real power is exactly
the weight used in the formalization.  The use of $N,K$ here is a renaming
of the axiom's general block and ambient dimensions, not a restriction of
its scope.

\paragraph{(2) Statement in the cited source.}
Write $\mathcal U\in\mathrm U(n)$ for Friedman and Mello's Haar unitary,
to distinguish it from the axiom's $U$.  Their Eq.~(1.2) defines
$S=\mathcal U^{\T}\mathcal U$.  For its leading $m\times m$ block $s$, their
Eq.~(3.7) writes the marginal, in redundant matrix coordinates, as
\[
 p_0(s)\ \propto\ \delta(s-s^{\T})
              [\det(I_m-s^\dagger s)]^{(n-2m-1)/2}.
\]
The delta factor imposes complex symmetry.  The matrix-ball support comes
from the corner construction and the derivation; it is not printed as a
separate indicator in that display.  The stated range is $m\le n/2$;
Appendix~1 explicitly includes $n\ge2m$.  Thus the endpoint $n=2m$, with
exponent $-1/2$, is included
\cite[Eqs.~(1.2), (3.7) and Appendix~1]{FriedmanMello1985}.
The modern principal-corner density in
Ref.~\cite[Theorem~2.1 and Eq.~(2.1)]{ShouMillerGalitski2025} is
corroboration, not an additional input in A1.

\paragraph{(3) Difference from the source and its justification.}
There are two changes of representation.  First, the source has
$\mathcal U^{\T}\mathcal U$, whereas the axiom has $UU^{\T}$.
Transposition preserves Haar probability: for fixed $A\in\mathrm U(n)$,
$A\mathcal U^{\T}=(\mathcal U A^{\T})^{\T}$, so right invariance of the
law of $\mathcal U$ makes the law of $\mathcal U^{\T}$ left invariant.
With $n=K$, Haar uniqueness gives $\mathcal U^{\T}\law U$, hence
$\mathcal U^{\T}\mathcal U\law UU^{\T}$.  Taking the same principal block
preserves equality in law.  This is not a pointwise equality of the two
products.

Second, the axiom uses independent symmetric coordinates and their
embedding, rather than the source's redundant coordinates and symmetry
delta.  Removing the redundant coordinates changes only a fixed volume
factor, which disappears on normalization.  The source density therefore
has exactly the support and exponent in Eq.~\eqref{eq:external-A1}, and its
probability normalization gives $0<Z_{N,K}<\infty$.  These convention
conversions are included in the imported law identity; no stronger
regularity statement is being attributed to Friedman and Mello.

\paragraph{(4) Notation correspondence.}
The dimension and matrix dictionary is
\begin{equation}
 n_{\mathrm{source}}=K,\qquad
 m_{\mathrm{source}}=N,\qquad
 s_{\mathrm{source}}\law C_{N,K}.
 \label{eq:external-A1-dictionary}
\end{equation}
The source's $\dagger$ is this paper's $*$; both mean conjugate transpose.
Its $S=\mathcal U^{\T}\mathcal U$ corresponds in law to our $UU^{\T}$.  Its independent
corner-entry volume becomes $\dd C$, and its unspecified normalizer becomes
$Z_{N,K}^{-1}$.  With this dictionary the imported Lean law is the density
in Eq.~\eqref{eq:hide-coe-density}.  Boundary regularity and score estimates
are not part of A1.

\subsection{Flat Takagi and Weyl integration (A2)}

\paragraph{(1) Exact mathematical translation of the Lean axiom.}
For an integer $N\ge1$, use the Borel structures on Euclidean spaces and let
\[
 \begin{aligned}
 \mathcal M_N&=\C^{N\times N},\qquad
 \mathcal S_N=\{C\in\mathcal M_N:C^{\T}=C\},\\
 \dd C&=\prod_{1\le i\le j\le N}
       \dd\Re C_{ij}\,\dd\Im C_{ij}.
 \end{aligned}
\]
We regard $\dd C$ as a measure on $\mathcal M_N$ supported on
$\mathcal S_N$, by embedding these independent upper triangular coordinates
as a symmetric matrix.  On all of $\mathcal M_N$ define
\[
 \lambda(C)=1-\operatorname{eig}(I-C^*C),
 \qquad
 \Delta(\lambda)=\prod_{i<j}(\lambda_j-\lambda_i),
\]
where $\operatorname{eig}$ uses a fixed ordering of the Hermitian
eigenvalues, with the fixed reindexing used by the formalization.
The coordinates of $\lambda(C)$ are the squared singular values as a
multiset.  Put
\[
 \rho_N(\dd\lambda)=
 \1_{(0,\infty)^N}(\lambda)\abs{\Delta(\lambda)}\,\dd\lambda.
\]
Then $\lambda:\mathcal M_N\to\mathbb R^N$ is measurable, and there is
\emph{one} finite constant $c_N>0$, depending only on $N$, such that for
\emph{every} measurable space $\mathcal Y$ and every measurable map
$F:\mathbb R^N\to\mathcal Y$ satisfying
$F(\lambda\circ\pi)=F(\lambda)$ for every $\lambda\in\mathbb R^N$
and every coordinate permutation $\pi$,
\begin{equation}
 (F\circ\lambda)_\#(\dd C)=c_N F_\#\rho_N.
 \label{eq:external-A2prime}
\end{equation}
The quantifier for $c_N$ precedes those for $\mathcal Y$ and $F$.
No standard Borel assumption on $\mathcal Y$ or finite-mass assumption is
needed.  These are exactly the data asserted by A2 in the formalization:
a positive finite constant, measurability of the selected spectrum on all
square matrices, and the displayed equality for arbitrary measurable
permutation invariant tests.

\paragraph{(2) Statement in the cited source.}
FitzGerald and Warren display the flat Jacobian
\[
 \dd X\ \propto
 \prod_{i<j}|\lambda_i-\lambda_j|\,\dd\lambda\,\dd\Omega
\]
in Section~6, printed p.~165, immediately after Eq.~(70)
\cite[Sec.~6, p.~165]{FitzGeraldWarren2020}.
Their matrix is complex symmetric, their flat coordinates are its independent
complex entries, and their $\lambda_i$ are the eigenvalues of $X^*X$.
Thus the matrix space, coordinate volume, squared singular values, and
Vandermonde power agree with those above.  We use this geometric Jacobian,
not the subsequent density with Gaussian parameters.  The formula is
unnumbered; Eq.~(70) locates it but is not itself the Jacobian.
For the general integration theorem and its measurable-integrand form we
also use An, Wang, and Yan, Theorem~4.2 and its following remark, printed
p.~13 \cite[Thm.~4.2 and following remark]{AnWangYan2006}.
Their notation is $G$ for the group, $X$ for the integration manifold,
$Y$ for its closed section, and $K$ for the section's common stabilizer.
Write $K_{\mathrm{AWY}}$ for this subgroup to distinguish it from our
ambient dimension.  Under their ensemble conditions (invariant measures,
orbit coverage, transversality, isotropy dimension and orthogonality),
and a finite $d$-sheeted covering $G/K_{\mathrm{AWY}}\times Y'\to X'$,
their Eq.~(4--4) states
\[
 \int_X f(x)p(x)\,\dd x
 =\frac1d\int_Y
       \left[\int_{G/K_{\mathrm{AWY}}}
                    f(\sigma_g(y))\,\dd\mu([g])\right]\dd\nu(y).
\]
Here $\sigma_g$ is the action, $p(x)\dd x$ is the invariant measure, and
$\dd\nu$ includes the section Jacobian.  The theorem states this for
smooth nonnegative or integrable $f$; its following remark replaces
smoothness by measurability, retaining nonnegativity or integrability.
Neither source states the arbitrary-target pushforward identity in
Eq.~\eqref{eq:external-A2prime} verbatim.

\paragraph{(3) Difference from the source and its justification.}
The imported assertion includes the specific spectrum map's measurability
on all square matrices, one constant before all tests, and arbitrary
measurable targets.  The source supplies the geometric integration formula.
The following deductions connect that formula to every imported clause.
They are mathematical explanations of the bundled A2 input, not separate
Lean proofs of those clauses.

\smallskip\noindent\emph{The selected spectrum is globally measurable.}
For Hermitian matrices $A,B$, the min--max formula gives, for consistently
ordered eigenvalues,
\[
 \max_i|\operatorname{eig}_i(A)-\operatorname{eig}_i(B)|
 \le \|A-B\|_{\mathrm{op}}.
\]
The same bound holds after any fixed reindexing.  Since
$C^*C-D^*D=C^*(C-D)+(C-D)^*D$, it follows for all
$C,D\in\mathcal M_N$ that
\[
 \max_i|\lambda_i(C)-\lambda_i(D)|
 \le (\|C\|_{\mathrm{op}}+\|D\|_{\mathrm{op}})
                      \|C-D\|_{\mathrm{op}}.
\]
Thus the exact selector used in A2 is continuous on all square matrices,
including at spectral multiplicities.  Also,
$(UCU^{\T})^*(UCU^{\T})=\overline U(C^*C)U^{\T}$ for unitary $U$,
so its spectrum is unchanged by unitary congruence.  Every nonnegative
Borel permutation invariant $h$ therefore gives a measurable invariant
integrand $C\mapsto h(\lambda(C))$.  No choice of Takagi vectors is needed
for this assertion.

\smallskip\noindent\emph{Regular coordinates and their uniqueness.}
Let $\mathcal S_N^{\mathrm{reg}}$ consist of the symmetric matrices whose
squared singular values are positive and pairwise distinct, and put
\[
 \mathcal W=\{s\in\mathbb R^N:0<s_1<\cdots<s_N\},\qquad
 H_N=\{\operatorname{diag}(\varepsilon_1,\ldots,\varepsilon_N):
                         \varepsilon_i\in\{-1,1\}\}.
\]
The complement of $\mathcal S_N^{\mathrm{reg}}$ is flat null: it is the
union of the zero sets of $|\det C|^2$ and the discriminant of the
characteristic polynomial of $C^*C$.  Their restrictions to $\mathcal S_N$
are nonzero real polynomials, as a positive diagonal matrix with distinct
entries shows.  For $N=1$ the collision condition is empty.

Takagi factorization makes the smooth map
\[
 \Phi:(\mathrm U(N)/H_N)\times\mathcal W
       \longrightarrow\mathcal S_N^{\mathrm{reg}},\qquad
 \Phi([U],s)=U\operatorname{diag}(s)U^{\T},
\]
surjective.  It is injective as well.  The ordered positive diagonal is
fixed by the singular values.  If a unitary $V$ stabilizes
$D=\operatorname{diag}(s)$, then $VDV^{\T}=D$ implies
$VD^2V^*=D^2$.  Because $D^2$ has distinct diagonal entries, $V$ is
diagonal; the first equality then gives $V_{ii}^2=1$.  Hence the stabilizer
is exactly $H_N$, and the angular coordinate is unique modulo this group.

\smallskip\noindent\emph{The differential and the integration theorem.}
At $([I],s)$, a real diagonal variation $E$ and a unitary tangent $X^*=-X$ give
\[
 \dd C=E+XD+DX^{\T}=E+XD-D\overline X.
\]
For $i<j$, write $X_{ij}=x_{ij}+\mathrm i y_{ij}$.  The corresponding entry
of the orbit tangent is
$(s_j-s_i)x_{ij}+\mathrm i(s_j+s_i)y_{ij}$.
Writing $X_{ii}=\mathrm i t_i$, the diagonal variation is
$E_{ii}+2\mathrm i s_i t_i$.  These real coordinate blocks are all
invertible on $\mathcal W$.  Thus $\Phi$ is a local diffeomorphism, and its
bijectivity makes it a global diffeomorphism.  This proves the covering
condition with one sheet; no further sign or permutation multiplicity is
left in these coordinates.

Here is the precise specialization of the cited integration theorem.
Take its integration manifold to be $\mathcal S_N^{\mathrm{reg}}$, its
section to be $\{\operatorname{diag}(s):s\in\mathcal W\}$, its group to be
$\mathrm U(N)$, and $p=1$.  The exceptional sets within these manifolds
are empty, so $X'=X$ and $Y'=Y$ in the source's notation.
The section is closed relative to
$\mathcal S_N^{\mathrm{reg}}$: a limit there of positive ordered diagonals
still has positive, distinct, ordered entries.  Its real diagonal tangent
is orthogonal to the orbit tangent in the real Frobenius metric, and the
displayed differential gives their direct sum.  The stabilizers equal the
finite group $H_N$.  Thus orbit coverage, transversality, the isotropy
dimension condition, orthogonality, and the covering condition all hold.
The Frobenius volume is $2^{N(N-1)/2}\dd C$, so $\dd C$ itself is also
invariant under unitary congruence.  Apply the theorem to Frobenius volume
and divide both sides by this fixed factor to obtain the formula for $\dd C$.

Give $\mathrm U(N)/H_N$ its invariant probability measure.  The same
differential gives, with one constant $a_N\in(0,\infty)$ depending only on
the fixed coordinate and angular normalizations, the radial Jacobian
$a_N\prod_i(2s_i)\prod_{i<j}(s_j^2-s_i^2)$.  The constant is independent of
$s$ by the displayed coordinate blocks, and independent of the angular
point by unitary invariance.  The measurable form of the integration
theorem therefore gives, for every nonnegative Borel permutation invariant
$h$, including when the integrals are infinite,
\[
 \int_{\mathcal S_N}h(\lambda(C))\,\dd C
 =a_N\int_{\mathcal W}h(s_1^2,\ldots,s_N^2)
               \prod_i(2s_i)\prod_{i<j}(s_j^2-s_i^2)\,\dd s.
\]
The flat null set removed above does not affect this nonnegative integral.

\smallskip\noindent\emph{Squared coordinates and the common constant.}
Set $x_i=s_i^2$ and
$\mathcal W_x=\{x\in\mathbb R^N:0<x_1<\cdots<x_N\}$.  The factors $2s_i$
cancel exactly against the coordinate differentials:
\[
 \prod_i(2s_i)\prod_{i<j}(s_j^2-s_i^2)\,\dd s
 =\Delta(x)\,\dd x.
\]
The positive orthant, except for its collision hyperplanes, is the disjoint
union of $N!$ coordinate permutations of $\mathcal W_x$.  Permutation
invariance of $h$, $|\Delta|$, and Lebesgue measure gives
\[
 \int_{\mathcal S_N}h(\lambda(C))\,\dd C
 =a_N\int_{\mathcal W_x}h(x)\Delta(x)\,\dd x
 =\frac{a_N}{N!}\int_{(0,\infty)^N}h(x)|\Delta(x)|\,\dd x.
\]
Thus $c_N=a_N/N!$ is fixed before any choice of test or target.
For an independent normalization check, insert $h(x)=e^{-\sum_i x_i}$.
Since $\tr(C^*C)=\sum_i|C_{ii}|^2+2\sum_{i<j}|C_{ij}|^2$, this gives
\[
 c_N=
 \frac{2^{-N(N-1)/2}\pi^{N(N+1)/2}}
 {\displaystyle\int_{(0,\infty)^N}
       e^{-\sum_i x_i}|\Delta(x)|\,\dd x}.
\]
The denominator is positive on an open chamber and finite by polynomial
growth and exponential decay.  For $N=1$ this reduces to $c_1=\pi$,
as also follows directly from planar polar coordinates.

\smallskip\noindent\emph{Arbitrary measurable targets.}
Finally, for any measurable $B\subseteq\mathcal Y$ take
$h(\lambda)=\1_{\{F(\lambda)\in B\}}$ in the scalar integration formula.
Measurability and permutation invariance of $F$ give the corresponding
properties of $h$.  The resulting identity is equality of the two measures
in Eq.~\eqref{eq:external-A2prime} on every measurable $B$, which proves
that equation for an arbitrary measurable target.  This completes the
implication from the cited formula to every part of the stated A2.

\paragraph{(4) Notation correspondence.}
FitzGerald and Warren's matrix $X$ is our $C\in\mathcal S_N$, their size
$n$ is $N$, and their $\lambda_i$ are the eigenvalues of $C^*C$.
Our vector $\lambda(C)=1-\operatorname{eig}(I-C^*C)$ contains this same
multiset in a fixed order.  Their angular variables $\Omega$ become
$[U]\in\mathrm U(N)/H_N$ with invariant probability measure in the
calculation above.  Their proportionality constant becomes the single
$c_N=a_N/N!$ when the ordered squared chamber is replaced by the full
positive orthant.  In the An--Wang--Yan specialization, $G=\mathrm U(N)$,
their $K$ is our $H_N$, $X=\mathcal S_N^{\mathrm{reg}}$, and $Y$ is the
positive ordered diagonal section.  Their action is
$\sigma_g(C)=gCg^{\T}$.  Their $K$ is a subgroup, whereas our $K$ elsewhere
is an ambient dimension; their $d$ counts covering sheets, here one.

In the formalization, A2 is applied with A1 to permutation invariant tests
of the determinant weighted COE law.  Its constant $c_N$ cancels under
probability normalization.  The odds and trace power maps are subsequent
constructions, not clauses imported in A2.  Permutation invariance is
essential to the axiom: an unrestricted equality between one canonically
ordered eigenvalue vector and a measure on the full unordered orthant
would be false.

\subsection{Gaussian GSVD beta Jacobi law (A3)}

\paragraph{(1) Exact mathematical translation of the Lean axiom.}
Let $N\ge1$, $a,b\in\mathbb Z_{\ge0}$, and $\beta\in\{1,2\}$.  On the
common sample space
\[
 \Omega_{N,a,b}=\mathbb C^{(N+a)\times N}
                      \times\mathbb C^{(N+b)\times N},
\]
let $\mu_{N,a,b,\beta}$ be the law of an independent pair $(X_1,X_2)$.
For $\beta=1$ their entries are independent $N(0,1)$ variables embedded in
$\mathbb C$; for $\beta=2$ they are independent
$(G_1+\mathrm iG_2)/\sqrt2$, with $G_1,G_2$ independent $N(0,1)$.
Define on every sample, including singular samples,
\[
 A=X_1^*X_1,\quad B=X_2^*X_2,\quad R=(A+B)^{1/2},\quad
 J=R^{-1}A(R^{-1})^*,\quad
 x_i(X_1,X_2)=\bigl(\sqrt{\operatorname{eig}_i(J)}\bigr)^2.
\]
Here the matrix square root is the positive semidefinite one, the inverse
is the usual inverse on invertible matrices and the zero matrix on singular
matrices, and the eigenvalues use the fixed ordering and reindexing of the
formalization.  This inverse convention is not the Moore--Penrose inverse.
The matrix $J$ is positive semidefinite, so $x_i=\operatorname{eig}_i(J)$.

Let $\mathsf J_{N,a,b,\beta}$ be the measure obtained by normalizing the
following kernel by its own integral $Z_{N,a,b,\beta}$:
\begin{equation}
 \mathsf J_{N,a,b,\beta}(\dd t)=
 \frac{\1_{(0,1)^N}(t)}{Z_{N,a,b,\beta}}
 \prod_{i=1}^N t_i^{\beta(a+1)/2-1}(1-t_i)^{\beta(b+1)/2-1}
 \prod_{i<j}|t_i-t_j|^\beta\,\dd t.
 \label{eq:external-A3-jacobi}
\end{equation}
The imported assertion has two conclusions: $x:\Omega_{N,a,b}\to\mathbb R^N$
is measurable on the whole sample space, and, for every measurable space
$\mathcal Y$ and every measurable permutation invariant
$F:\mathbb R^N\to\mathcal Y$,
\begin{equation}
 (F\circ x)_\#\mu_{N,a,b,\beta}=F_\#\mathsf J_{N,a,b,\beta}.
 \label{eq:external-A3-symmetric-test}
\end{equation}
Collision nullity, project parameter substitutions, and later trace laws
are not additional conclusions of this axiom.

\paragraph{(2) Statement in the cited source.}
Edelman and Sutton's Definition~1.1 uses GSVD cosine coordinates $c_i$,
the diagonal entries of a nonnegative diagonal matrix $C_0$ accompanied by
$S_0$ with $C_0^2+S_0^2=I_n$.  The source explicitly notes that these
coordinates are unique only up to reordering; its $c_i$ are not the ratios
$c_i/s_i$ called generalized singular values in another convention.
Proposition~1.2 states that independent Gaussian matrices $N_1,N_2$ of
sizes $(n+a)\times n$ and $(n+b)\times n$, with the real or complex
standard Gaussian conventions above, have squared GSVD coordinates
$\{c_1^2,\ldots,c_n^2\}$ with the Jacobi law of parameters $a,b$ and
$\beta=1$ or $2$, respectively
\cite[Definition~1.1 and Proposition~1.2]{EdelmanSutton2008}.
The source density is Eq.~\eqref{eq:external-A3-jacobi} with $N=n$.
Its proof identifies the coordinates with the eigenvalues of
\[
 (N_1^*N_1)(N_1^*N_1+N_2^*N_2)^{-1}.
\]
The proposition does not specify our eigenvalue selector or its values and
measurability on every singular sample, and it does not give an ordered
vector the symmetric density on the entire cube.

\paragraph{(3) Difference from the source and its justification.}
The additional interface content is the concrete coordinate map and its
global measurability.  Each Gaussian matrix has full column rank almost
surely: the determinant of its first $N$ rows is a nonzero polynomial, whose
zero set has Gaussian measure zero.  On that event $A,B>0$, and
\[
 RJR^{-1}=A(A+B)^{-1},\qquad 0<J<I_N.
\]
Thus the Hermitian matrix used in the axiom is similar to the source's
matrix without any commutation assumption on $A,B$.  Its eigenvalues are
the same squared GSVD coordinates.  The chosen values on the null singular
set do not affect the probability law.

For measurability on all of $\Omega_{N,a,b}$, the Gram maps and the positive
matrix square root are continuous.  Inversion extended by zero is Borel:
it is the continuous rational map $\operatorname{adj}(R)/\det R$ on the
open invertible set and is constant on its closed complement.  Hence $J$
is measurable.  Continuity of ordered Hermitian eigenvalues, fixed
reindexing, and the real square root proves measurability of $x$.  Since
$J\ge0$ on every sample, squaring its real square-root eigenvalues does not
change them.

Finally, a permutation invariant $F$ has the same value on every ordering
of the squared GSVD coordinates.  Apply the source's unordered law to the
indicator of $F^{-1}(D)$ for each measurable $D\subseteq\mathcal Y$ to
obtain Eq.~\eqref{eq:external-A3-symmetric-test}.  This also covers arbitrary
measurable targets.  These coordinate and measurability deductions explain
the extension from Proposition~1.2; they are bundled into A3, not separately
proved by that imported Lean declaration.  Collision nullity is then
derived by the symmetric collision indicator and the null hyperplanes of
the Jacobi density.

\paragraph{(4) Notation correspondence.}
The source's $n,N_1,N_2,c_i^2$ correspond respectively to this subsection's
$N,X_1,X_2,x_i$ up to ordering.  The Gaussian variances and the parameters
$a,b,\beta$ are unchanged.  The formalization uses the specialization
\begin{equation}
 n_{\mathrm{source}}=N,\qquad a=1,\qquad
 b=K-2N,\qquad\beta=1,\qquad K\ge2N.
 \label{eq:external-A3-dictionary}
\end{equation}
The two matrix sizes become $(N+1)\times N$ and $(K-N)\times N$; the two
individual-coordinate exponents become $0$ and $(K-2N-1)/2$.  This is the
radial kernel obtained from A1 and A2.  Reflection, odds transformation,
and trace power maps give the beta prime trace law only after these inputs
are combined.

\subsection{Wishart and inverse Wishart tensor moments (A4)}

\paragraph{(1) Exact mathematical translation of the Lean axiom.}
Let $N,q\ge1$, $\beta,\gamma\in\mathbb R$, and
$\sigma\in\operatorname{Sym}_N^+(\mathbb R)$, with
\begin{equation}
 \gamma=\beta-\frac{N+1}{2},\qquad \gamma>q-1.
 \label{eq:external-A4-gap}
\end{equation}
Here $q$ is the moment order, and $\beta$ is the Wishart shape parameter,
not the real or complex index in A3.  Let $\mu$ be a probability measure on
$\operatorname{Sym}_N^+(\mathbb R)$ satisfying, for every real symmetric
$\theta$ such that $\sigma^{-1}-\theta>0$,
\[
 \int e^{\tr(\theta w)}\,\mu(\dd w)
       =\det(I_N-\theta\sigma)^{-\beta}.
\]
This supplied probability law and its transform identity are hypotheses
of A4, through the definition of $W\sim W_N(\beta,\sigma;\mathbb R)$.
The axiom does not separately assert existence of a law for arbitrary
shape parameters.  Write $\E_\mu$ for integration against this law.

For arbitrary complex $N\times N$ matrices $m_1,\ldots,m_q$, a permutation
$g\in S_{2q}$, and real symmetric $w$, define
\begin{equation}
 T_g(w;m)=\sum_{j_1,\ldots,j_{2q}=1}^N
       \left(\prod_{r=1}^q(m_r)_{j_{2r-1},j_{2r}}\right)
       \left(\prod_{r=1}^q w_{j_{g(2r-1)},j_{g(2r)}}\right).
 \label{eq:external-A4-Tg}
\end{equation}
Let $\mathcal M(2q)$ be the perfect matchings of $\{1,\ldots,2q\}$ and
$M_0=\{\{1,2\},\ldots,\{2q-1,2q\}\}$.  Represent a matching $M$ by the
permutation $g_M$ that lists each pair increasingly and lists the first
members of pairs increasingly.  Let $\ell(M,L)$ count the connected
components in $M\cup L$ and put $\kappa(g)=\ell(M_0,gM_0)$.
The coefficient used in the axiom is defined by the finite matching matrix:
\[
 G_z(M,L)=z^{\ell(M,L)},\qquad
 \mathrm{Wg}^{\mathrm O}_{\mathrm L}(g;z)
      =G_z^{-1}(M_0,gM_0),\qquad
 \widetilde{\mathrm{Wg}}_{\mathrm L}(g;\gamma)
      =(-1)^q2^q\mathrm{Wg}^{\mathrm O}_{\mathrm L}(g;-2\gamma).
\]
The subscript $\mathrm L$ distinguishes this definition from the source's
definition below.  The matrix inverse is the ordinary inverse at
$z=-2\gamma$; its existence throughout the stated range is justified in
part (3).  The imported conclusion is the conjunction
\begin{equation}
 \E_\mu T_g(W;m)
 =2^{-q}\sum_{M\in\mathcal M(2q)}
          (2\beta)^{\kappa(g^{-1}g_M)}T_{g_M}(\sigma;m),
 \label{eq:external-A4-direct}
\end{equation}
\begin{equation}
 \E_\mu T_g(W^{-1};m)
 =\sum_{M\in\mathcal M(2q)}
       \widetilde{\mathrm{Wg}}_{\mathrm L}(g^{-1}g_M;\gamma)
                              T_{g_M}(\sigma^{-1};m).
 \label{eq:external-A4-inverse}
\end{equation}
These are equalities of complex integrals.  The declaration does not
return separate integrability conclusions or any specialized trace or
score bound.

\paragraph{(2) Statement in the cited source.}
{All source theorem, equation, and page numbers in this
subsection refer to
\href{https://arxiv.org/abs/1004.4717v3}{arXiv:1004.4717v3}
of Ref.~\cite{Matsumoto2012}.}
Matsumoto's Theorem~3 takes $W\sim W_d(\beta,\sigma;\mathbb R)$,
$\gamma=\beta-(d+1)/2>n-1$, arbitrary $d\times d$ matrices
$m_1,\ldots,m_n$, and $g\in S_{2n}$.  It states
\[
 \begin{aligned}
 \E T_g(W;m)
   &=2^{-n}\sum_{M\in\mathcal M(2n)}
                  (2\beta)^{\kappa(g^{-1}g_M)}T_{g_M}(\sigma;m),\\
 \E T_g(W^{-1};m)
   &=\sum_{M\in\mathcal M(2n)}
      \widetilde{\mathrm{Wg}}_{\mathrm{src}}(g^{-1}g_M;\gamma)
                                      T_{g_M}(\sigma^{-1};m).
 \end{aligned}
\]
{These are the paired identities of Theorem~3, p.~20.}  The source's Wishart convention is the
Laplace-transform definition in its Section~1.1; when $2\beta=p$ is an
integer its Gaussian realization is a sum of $p$ independent outer
products with vector covariance $\sigma/2$.

The source defines its coefficient differently from part (1).
Writing its moment order as $q=n$ for the next display, Eq.~(4.10) gives
\[
 \mathrm{Wg}^{\mathrm O}_{\mathrm{src}}(g;z)
   =\frac{1}{(2q-1)!!}
        \sum_{\nu\vdash q}
        \frac{f^{2\nu}\omega^\nu(g)}{C_\nu(z)},\qquad
 C_\nu(z)=\prod_{(i,j)\in\nu}(z+2j-i-1).
\]
Here $f^{2\nu}$ is the irreducible character dimension for $S_{2q}$,
$H_q$ is the stabilizer of $M_0$ with $|H_q|=2^qq!$, and
$\omega^\nu(g)=|H_q|^{-1}\sum_{h\in H_q}\chi^{2\nu}(gh)$ is the
normalized zonal spherical function.  The definition requires every
$C_\nu(z)\ne0$.  Equation~(5.2) sets
$\widetilde{\mathrm{Wg}}_{\mathrm{src}}(g;\gamma)
=(-1)^q2^q\mathrm{Wg}^{\mathrm O}_{\mathrm{src}}(g;-2\gamma)$
\cite[Eqs.~(4.10), (5.2)]{Matsumoto2012}.

\paragraph{(3) Difference from the source and its justification.}
The tensor contractions and the two moment identities are the same after
renaming $d=N$ and $n=q$.  There are two representation issues to check.
First, the supplied probability measure in part (1) is identified with the
source Wishart law by its Laplace transform on a neighborhood of zero in
the real symmetric-matrix coordinates.  Uniqueness of such a transform
therefore identifies the laws.  The gap condition implies
$\beta>(N+1)/2$, so the source's singular Wishart regimes do not arise.

Second, the inverse-Gram coefficient must be identified with the zonal
coefficient with its exact normalization.  For every box $(i,j)$ in a
partition of $q$, $j\le q$ and $i\ge1$, so
\[
 -2\gamma+2j-i-1\le-2\gamma+2q-2<0.
\]
Thus all the zonal denominators are nonzero at $z=-2\gamma$.
Put $Q_z(g)=z^{\kappa(g)}$ and
$V_z(g)=\mathrm{Wg}^{\mathrm O}_{\mathrm{src}}(g;z)$.
Matsumoto's Lemma~5, with its normalized identity element
$|H_q|^{-1}\1_{H_q}$, says for ordinary finite-group convolution that
\[
 Q_z*V_z=|H_q|\1_{H_q}.
\]
Evaluate at $g_M^{-1}g_L$ and sum over the cosets represented by $g_P$.
The two functions are $H_q$-bi-invariant, so the group sum has a common
factor $|H_q|$.  After dividing by that factor, the identity is exactly
\[
 \sum_{P\in\mathcal M(2q)}
      G_z(M,P)V_z(g_P^{-1}g_L)=\delta_{M,L}.
\]
This constructs a right inverse of the finite square matrix $G_z$.
Consequently $G_z$ is invertible and
\[
 \mathrm{Wg}^{\mathrm O}_{\mathrm{src}}(g;z)
      =G_z^{-1}(M_0,gM_0)
      =\mathrm{Wg}^{\mathrm O}_{\mathrm L}(g;z).
\]
The same factor $(-1)^q2^q$ then identifies the two modified coefficients.
No additional factorial or sign is present
\cite[Eq.~(4.10) and Lemma~5]{Matsumoto2012}.

This proves mathematical equivalence of the two definitions in the whole
range used by A4.  The code defines the inverse-Gram coefficient directly
and imports the moment equalities with it; it does not separately formalize
the zonal expansion and this conversion.  Similarly, the source supplies
finite moments in its stated range, but the imported integral equalities
should not be read as separate Lean integrability assertions.

\paragraph{(4) Notation correspondence.}
The source's matrix dimension $d$ is our $N$, and its moment order $n$ is
our $q$.  This $n$ is unrelated to the half-photon count in $N=2n$ elsewhere
in the paper.  The symbols $\beta,\gamma,\sigma,m_r,g$ keep their meanings;
the source's $\widetilde{\mathrm{Wg}}$ equals
$\widetilde{\mathrm{Wg}}_{\mathrm L}$ by part (3).  For the denominator
Wishart calculation, the dictionary is
\begin{equation}
 d=N,\qquad n=q=4,\qquad k=K-N,\qquad
 \beta=\frac{K-N}{2},\qquad\gamma=\frac{K-2N-1}{2}.
 \label{eq:external-A4-dictionary}
\end{equation}
The condition $K\ge2N+8$ gives $\gamma\ge7/2>3=q-1$.
Variance-$1/2$ Gaussian entries give source scale $\sigma=I_N$;
variance-one entries give $\sigma=2I_N$.  The deterministic factor of two
relates these Gram matrices.  Choosing the permutation and insertion
matrices so that $T_g(W^{-1};m)=\tr(W^{-4})$ is a later specialization,
followed by the matching sum, integrability closure, centering, and score
estimates.  None of these specialized estimates is built into A4.

\endgroup
\clearpage
\endgroup
\begingroup\section{Haar/Stiefel disintegration and target convergence}
\label{app:haar-stiefel}

{Haar/Stiefel deletion, beta radial laws, and polar disintegration
are classical {invariant measure} tools
~\cite{Chikuse1991Stiefel,Mezzadri2007,ZyczkowskiSommers2000Truncations,
BourgadeEtAl2008,ZyczkowskiSommers2001}.}  A direct GBS
precedent also factors an arbitrary $K$ Haar transpose Gram product as a
positive random congruence of a square COE product
~\cite[Lemma~2.2]{ShouEtAl2026ArbitrarySqueezers}.  The step used here is
narrower: deleting one ambient coordinate closes exactly as the
rank one beta congruence recursion in Lemma~\ref{lem:hide-recursion}.
The radial Markov chain then transports the resulting centered square COE
single column probability bound without reintroducing a dependence on $K$.

{For the commutation argument, the single column radial factors in
Lemma~\ref{lem:hide-recursion} are positive definite and their laws are
invariant under unitary conjugation.  A finite radial kernel chain is a
composition of these congruence mixtures.  It is therefore enough to prove
commutation for one
arbitrary conjugation invariant positive definite factor: the conclusion
then applies successively to every factor and hence to the whole chain.}
\subsection{Proof of Lemma~\ref{lem:hide-commute}}
\label{subsec:hide-commute-direct-proof}

\begin{proof}[Proof of Lemma~\ref{lem:hide-commute}]
{Figure~\ref{fig:haar-lift-transport} summarizes the four
measure theoretic steps.  First, let
$P\sim\lambda$, with $\lambda$ as in Lemma~\ref{lem:hide-commute}, and let
$V$ be an independent Haar unitary.  The law of $PV$ is invariant under
independent left and right multiplication by elements of $\mathrm U(N)$:
right invariance follows from Haar invariance, while left multiplication by
$W$ may be written
$WPV=(WPW^*)(WV)$ and uses conjugation invariance of $P$ together with Haar
invariance of $WV$.  Thus the lift is a biunitarily invariant finite measure
on $\mathrm{GL}_N(\C)$.

Second, if the input matrix law is invariant under unitary congruence, the
right unitary in the lift has no effect on the induced positive congruence
operator.  Indeed, $PV$ acts by
$X\mapsto PVX(PV)^{\T}$, and the conditional law of
$VXV^{\T}$ is the law of $X$.  Third, the fixed spectrum orbital law of
$\exp(sQ_v)$ has the same double coset lift, so the two Markov operators are
represented by convolution of biunitarily invariant measures.

Fourth, these convolutions commute.  R\"osler and Voit define a Gelfand pair
by commutativity of the bounded biinvariant measure algebra and give the
involutive automorphism criterion
\cite[Definition~3.1 and Lemma~3.2]{RoeslerVoit2008}.  For
$G=\mathrm{GL}_N(\C)$ and $K=\mathrm U(N)$ take
$\theta(g)=(g^*)^{-1}$.  This is a continuous involutive automorphism.  If
$g=u d v^*$ is a singular value decomposition, then
$g^{-1}=(vu^*)\theta(g)(vu^*)$, so
$g^{-1}\in K\theta(g)K$.  The criterion therefore makes $(G,K)$ a Gelfand
pair, and the two lifted convolution measures commute.  Removing the
auxiliary right unitaries gives the claimed probability law commutation.
The same conjugation covariance shows that
$\mathcal T_\lambda\eta$ remains invariant under unitary congruence.

Separately, let $R_\lambda(X,E)$ be the congruence transition kernel obtained
by averaging the indicator of $PXP^{\T}\in E$ over $P\sim\lambda$.
For a finite signed measure $\sigma$, integration of
$R_\lambda(X,E)$ against the Jordan parts of $\sigma$ defines
$\mathcal T_\lambda\sigma$.  Signed integration is linear, so this construction
is linear and agrees with the probability law action.  For a bounded test
function write
$R_\lambda f(X):=\int f(Y)R_\lambda(X,dY)$.  Duality and the Markov bound
$\abs{R_\lambda f}\le1$ for $\abs{f}\le1$ prove
Eq.~\eqref{eq:hide-commute-contraction}.  Thus the extension is bounded.

Now put $F(s)=\mathcal K_s\eta$ as a curve in finite signed measures with the
variation norm.  A bounded linear map preserves $C^r$ curves and commutes
with each derivative through order $r$.  Applying this fact to
$\mathcal T_\lambda F$, then using Eq.~\eqref{eq:hide-commute} to identify it
with $s\mapsto\mathcal K_s\mathcal T_\lambda\eta$, proves
Eq.~\eqref{eq:hide-commute-derivative}; its norm inequality follows from
Eq.~\eqref{eq:hide-commute-contraction}.  The later hiding estimate uses only
the probability consequence
$d_{\rm TV}(\mathcal T_\lambda\alpha,
\mathcal T_\lambda\beta)\le d_{\rm TV}(\alpha,\beta)$, namely
Eq.~\eqref{eq:hide-radial-TV}.}
Thus the Haar lift proves both unitary congruence invariance and
Eq.~\eqref{eq:hide-commute}; the Markov kernel proves the bounded linear
extension and Eq.~\eqref{eq:hide-commute-contraction}; and applying that
extension to the variation differentiable curve proves both assertions in
Eq.~\eqref{eq:hide-commute-derivative}. These are all the conclusions of
Lemma~\ref{lem:hide-commute}, so its proof is complete.
\end{proof}

\subsection{Target convergence}

The {target convergence} input to the dense telescope follows
from the same strict size Haar block density
\cite{Jiang2009HaarBlocks}.  Apply Lemma~\ref{lem:hide-sparse} with
$(M,p,q)=(m,K,N)$, transpose the two rectangular blocks, and push both laws
forward by $Y\mapsto K^{-1/2}YY^{\mathsf T}$.  Total variation data processing
gives, whenever
$1\le N\le K$ and $K+N<m$,
\begin{equation}
 d_{\rm TV}(\mu_{m,K},\nu_K)
 \le \frac{(K+N)\sqrt{KN}}{m}.
 \label{eq:hide-target-rate}
\end{equation}
Given $m_0\in\mathbb N$ and
$\varepsilon>0$, choose $n$ so that
$m=m_0+n>K+N$ and $(K+N)\sqrt{KN}/m\le\varepsilon$.  This proves
Eq.~\eqref{eq:hide-target-convergence}.  {Hence the dense
telescope reaches the Gaussian transpose Gram target from the same strict
size density route.}
\endgroup
\begingroup\section{COE density and fourth order boundary calculus}
\label{app:coe-density}

This appendix identifies the COE trace moments with the Wishart model and
justifies the event derivatives used in Lemma~\ref{lem:hide-scores}.
Appendix~\ref{app:projective-scores} derives the score formulas and bounds,
using the detailed moment calculations and polynomial certificates in
Appendix~\ref{app:score-certificates}. These ingredients are combined in
Sec.~\ref{subsec:hide-scores-direct-proof} to prove Lemma~\ref{lem:hide-scores}.
Throughout the appendix,
$1\le N$ and $16N\le K$; in particular all square roots and inverse Wishart
moments below are in their stated positive parameter range.

Let $C_{N,K}$ be the upper $N$ by $N$ corner of a $K$ dimensional
COE matrix.  Friedman and Mello's density, encoded by external input A1, is
\begin{equation}
 f_{N,K}(C)=\mathcal Z_{N,K}^{-1}
 \det(I-C^*C)^{(K-2N-1)/2}
 \mathbf 1_{\{C^*C<I\}} .
 \label{eq:hide-coe-density}
\end{equation}
It is a density on the independent complex symmetric coordinates
\cite{FriedmanMello1985}; see also
Ref.~\cite[Theorem~2.1 and Eq.~(2.1)]{ShouMillerGalitski2025}.  The square
base law and its unscaled coordinate are
\begin{equation}
 \begin{gathered}
 A=\sqrt K\,C_{N,K},\qquad C=A/\sqrt K,\\
 \lambda_{N,K}:=\Law(A).
 \end{gathered}
 \label{eq:hide-coe-scaling}
\end{equation}

Put
\begin{equation}
 \begin{gathered}
 c=K-2N-1,\qquad Y=cZ,\\
 Z=C(I-C^*C)^{-1}C^*.
 \end{gathered}
 \label{eq:hide-YZ}
\end{equation}
On the open matrix ball $C^*C<I$, define
\begin{equation}
 \Omega=I+Z,\qquad
 \mathcal R=\sqrt c\,\Omega C,\qquad
 \mathcal R\mathcal R^*=Y\Omega .
 \label{eq:hide-omega-R}
\end{equation}
The last identity uses $c>0$, which follows from the standing dense
hypothesis.  This qualification is essential: without positivity of $c$ the
displayed square root identity is not valid.

For the radial transport, let $G_A\in\mathbb R^{(N+1)\times N}$ and
$G_B\in\mathbb R^{(K-N)\times N}$ be independent standard Gaussian matrices,
and set $A_0=G_A^{\mathsf T}G_A$ and $B_0=G_B^{\mathsf T}G_B$.  The A1 through A3
trace vector transport gives, for the four powers actually used in the
score proof,
\begin{equation}
 \begin{aligned}
 \bigl(\tr Z^j\bigr)_{j=1}^{4}
 &\ \law\ 
 \bigl(\tr(B_0^{-1}A_0)^j\bigr)_{j=1}^{4},\\
 A_0&\sim W_N(N+1,I),\\
 B_0&\sim W_N(K-N,I).
 \end{aligned}
 \label{eq:hide-beta-prime}
\end{equation}
This is the trace consequence invariant under permutations of the Takagi and
Weyl integration formula and the matrix beta to beta prime change of variables.  It does
not assert equality of canonically ordered eigenvalue vectors.  Muirhead's
matrix beta representation gives classical corroboration
\cite[Theorem~3.3.1, p.~110]{Muirhead1982}.  The exact first denominator
trace moment used later is
\begin{equation}
 \mathbb E\,\tr(cB_0^{-1})=N .
 \label{eq:hide-inverse-wishart-mean}
\end{equation}

We next record the boundary facts used by the fourth order differentiation
argument.
Write $\alpha=(K-2N-1)/2$. We use only the dense range
$K\ge16N$, where
\begin{equation}
 \alpha\ge13/2>4,\qquad \alpha-j>0\quad(0\le j\le4).
 \label{eq:hide-boundary-exponents}
\end{equation}
In particular, for every $0\le j\le4$,
\begin{equation}
 \int_0^1 u^{\alpha-j}\,\dd u
 =\frac1{\alpha-j+1}<\infty.
 \label{eq:hide-boundary-integrability}
\end{equation}
The argument below also controls simultaneous small matrix gaps.

For $v\in S^{2N-1}$ let $Q_v=vv^*-I/N$.  On the independent symmetric
coordinates, the centered congruence flow has complex Jacobian one:
\begin{equation}
 \det_{\mathbb C}\!\left[H\mapsto
 e^{tQ_v}He^{tQ_v^{\mathsf T}}\right]=1,\qquad t\in\mathbb R .
 \label{eq:hide-centered-jacobian}
\end{equation}
For a Borel set $E\subset\mathbb C^{N\times N}$ define
\begin{equation}
 O_E(t)=\int_{S^{2N-1}}
 \int \mathbf 1_E\!\left(e^{tQ_v}Ae^{tQ_v^{\mathsf T}}\right)
 \,\lambda_{N,K}(\dd A)\,\dd v .
 \label{eq:hide-centered-event-path}
\end{equation}

\begin{lemma}[Fourth order centered COE event regularity]
\label{lem:coe-boundary-regularity}
For every Borel $E\subset\mathbb C^{N\times N}$, the path $O_E$ in
Eq.~\eqref{eq:hide-centered-event-path} is $C^4$.
\end{lemma}

\begin{proof}
Use the independent symmetric coordinates and first remove the fixed
scale $\sqrt K$. Put $F(C)=I-CC^*$ and extend
$f(C)=\det(F(C))^\alpha$ by zero outside $F(C)>0$.
On the support closure, $0\le F(C)\le I$ and $\|C\|_{\rm op}\le1$.
Let $g$ be the smallest eigenvalue of $F(C)$ at an interior point.
Jacobi's formula and the product rule show that every coordinate derivative
of $f$ of order $j\le4$ is a finite sum of terms bounded by
\[
 C_{N,K,j}\det(F(C))^\alpha g^{-j}
 \le C_{N,K,j}g^{\alpha-j}.
\]
Here derivatives of $F$ are bounded on this compact set, derivatives of
order above two vanish, and each inverse factor contributes at most $g^{-1}$.
The determinant inequality uses all eigenvalues being at most one; it
remains valid when several gaps vanish. Since $\alpha>4$, these derivative
bounds tend to zero at every boundary point through order four. Extending
each derivative by zero gives the derivatives of the extension: along a
coordinate line this follows successively from the fundamental theorem of
calculus on its open support intervals and continuity at their endpoints.
Thus the extended density is $C^4$ with compact support.

For $|t|\le T$, the matrices $e^{-tQ_v}$ and their first four time
derivatives are bounded uniformly over the compact unit sphere. The chain
rule consequently bounds every time derivative through order four of
$f(e^{-tQ_v}Ce^{-tQ_v^{\T}})$ by a constant times the indicator of one
compact set containing all transported supports for $|t|\le T$.
This is an integrable majorant in the flat symmetric coordinates,
independent of $t$ and $v$.
\begingroup
To specify the derivative densities, let $\dd A$ denote Lebesgue measure
on the independent complex symmetric coordinates, and restore the scale
and normalization by setting
\[
 p_{N,K}(A)
 :=K^{-N(N+1)/2}\mathcal Z_{N,K}^{-1}f(A/\sqrt K).
\]
Thus $p_{N,K}$ is the zero-extended density of $\lambda_{N,K}$.
For each fixed $v\in S^{2N-1}$, define
\[
 \begin{aligned}
 f_{t,v}(A)&:=p_{N,K}\!\left(e^{-tQ_v}Ae^{-tQ_v^{\T}}\right),\\
 J_{j,v}(t)(A)&:=\partial_t^j f_{t,v}(A),\qquad 0\le j\le4.
 \end{aligned}
\]
Here $j$ denotes derivative order and $t$ denotes time.  By
Eq.~\eqref{eq:hide-centered-jacobian}, the centered flow preserves $\dd A$,
so $f_{t,v}$ is the density of the transported law for this fixed direction.
Regard each $J_{j,v}(t)$ as an element of $L^1(\dd A)$.
The chain rule makes these representatives jointly continuous in $(t,v,A)$.
The common integrable majorant above is preserved by the fixed scaling
and normalization. Dominated convergence therefore proves continuity and
differentiation through order four in this $L^1$ space. The same majorant
justifies Fubini and differentiation of the sphere averages
$\overline J_j(t):=\int_{S^{2N-1}}J_{j,v}(t)\,\dd v$, where $\dd v$
is normalized surface measure.
In particular, the fixed-direction and averaged identities in $L^1(\dd A)$ are
\begin{equation}
 \begin{aligned}
 \left.\frac{\dd}{\dd t}J_{3,v}(t)\right|_{t=0}
   &=J_{4,v}(0),\qquad v\in S^{2N-1},\\
 \left.\frac{\dd}{\dd t}\overline J_3(t)\right|_{t=0}
   &=\overline J_4(0),
 \end{aligned}
 \label{eq:hide-fourth-secant}
\end{equation}
\endgroup
The centered congruence flow has the group identity
\begin{equation}
 e^{(t_0+h)Q_v}A e^{(t_0+h)Q_v^{\mathsf T}}
 =e^{hQ_v}\bigl(e^{t_0Q_v}A e^{t_0Q_v^{\mathsf T}}\bigr)
  e^{hQ_v^{\mathsf T}} .
 \label{eq:hide-centered-flow-group}
\end{equation}
It therefore transports the same local $C^4$ argument from zero to every
base time $t_0$, which proves the stated global regularity in $t$.
Restriction to a measurable event is a continuous linear functional on
$L^1$, and spherical averaging is handled by the same uniform envelope and
Fubini.  Hence the event path is $C^4$.  This argument neither differentiates
a moving boundary hypersurface nor assumes
that the density extended by zero belongs to an ambient $W^{4,1}$ space, and it
does not select singular vectors at multiple singular values.
\end{proof}

For the central path $S_E$ in Eq.~\eqref{eq:hide-event-paths}, the same
extension by zero and dominated convergence argument through order two
applies to $A\mapsto e^{tI}Ae^{tI}$.  Retaining its coordinate Jacobian,
which is constant in the matrix coordinate at each time, proves
$S_E\in C^2$ globally and supplies the regularity asserted in
Lemma~\ref{lem:hide-scores}.

\begingroup
The same construction proves the derivative interchange actually used in
Appendix~\ref{app:projective-scores}.  If $O_{v,E}(t)$ denotes the inner event
probability for fixed $v$ in Eq.~\eqref{eq:hide-centered-event-path}, then for
$0\le j\le4$,
\begin{equation}
 \begin{aligned}
 O_E^{(j)}(t)&=\int_{S^{2N-1}}O_{v,E}^{(j)}(t)\,\dd v,\\
 O_{v,E}^{(j)}(t)&=\int\mathbf1_E(A)J_{j,v}(t)(A)\,\dd A.
 \end{aligned}
 \label{eq:hide-centered-derivative-interchange}
\end{equation}
All hypotheses, including the Borel measurability of $E$, occur explicitly
in the preceding construction.
\endgroup
\endgroup
\begingroup\section{Projective scores and finite constants}
\label{app:projective-scores}

This appendix records the finite identities and estimates used by the dense
one step argument.  Classical inverse Wishart and complex projective formulas are
used only through their fixed degree specializations
\cite{CollinsSniady2006,Matsumoto2012,AidaStroock1994}; no unquantified
operator moment family is needed.
The detailed moment substitutions and exact arithmetic for the constants are given in Appendix~\ref{app:score-certificates}.

\subsection{The exact inverse Wishart recurrence}

Put $D=cB_0^{-1}$ and let
$x=\tr D$, $y=\tr D^2$, $z=\tr D^3$, and $r=\tr D^4$.
For $N\ge1$ and $K\ge16N$, with $n=N$, the concrete inverse Wishart law
satisfies the closed recurrence
\begin{equation}
\begin{aligned}
 c\E x^2&=cn^2+2\E y,\\
 (c-1)\E y&=cn+\E x^2,\\
 c\E x^3&=cn\E x^2+4\E(xy),\\
 c\E(xy)&=cn\E y+4\E z,\\
 (c-2)\E z&=c\E y+2\E(xy),\\
 c\E x^4&=cn\E x^3+6\E(x^2y),\\
 c\E(x^2y)&=cn\E(xy)+2\E y^2+4\E(xz),\\
 (c-1)\E y^2&=c\E(xy)+\E(x^2y)+4\E r,\\
 c\E(xz)&=cn\E z+6\E r,\\
 (c-3)\E r&=c\E z+2\E(xz)+\E y^2.
\end{aligned}
\label{eq:hide-inverse-wishart-recurrence-system}
\end{equation}
Solving these ten coupled equations, followed by the Gaussian Wick calculation
for the numerator and clearing denominators yields the three product law
inequalities in Eqs.~\eqref{eq:hide-fourth-ell-one-fourth-moment} through
\eqref{eq:hide-fourth-ell-one-three-moment} below.

One elementary differential identity used in that calculation is the
following.  If $B_0=R^{\T}R$ with $R$ of full column rank, then
\begin{equation}
 \norm{\nabla_R\tr(cB_0^{-1})}_{F}^{\,2}
 =4c^2\tr(B_0^{-3}).
\label{eq:hide-inverse-gradient}
\end{equation}

For the independent uniform complex direction $v$, write
$P_v=vv^*$ and $Q_v=P_v-I/N$.  The exact centered projective contraction is
\begin{equation}
 \int \tr(Q_vX)\tr(Q_vY)\,\dd v
 =\frac{\tr(XY)-\tr X\,\tr Y/N}{N(N+1)}.
\label{eq:hide-projective-contraction}
\end{equation}
This is the only second projective moment needed for the quadratic score;
the coordinate tensor identity through degree four supplies
the remaining fixed degree contractions.

\subsection{Fourth score}

\begingroup
Let $p_{N,K}$ be the density of $\mu_{K,K}$ on the independent complex
symmetric coordinates, including its support indicator. For $A$ in its
open support, define the transported density, local likelihood ratio,
and logarithmic scores by
\[
 \begin{gathered}
 p_{s,Q_v}(A):=
 p_{N,K}\!\left(e^{-sQ_v}Ae^{-sQ_v^{\mathsf T}}\right),\qquad
 \Lambda_{s,Q_v}(A):=\frac{p_{s,Q_v}(A)}{p_{N,K}(A)},\\
 \ell_j(Q_v;A):=
 \left.\partial_s^j\log\Lambda_{s,Q_v}(A)\right|_{s=0},
 \qquad 1\le j\le4.
 \end{gathered}
\]
The centered congruence has Jacobian one by
Eq.~\eqref{eq:hide-centered-jacobian}, so $p_{s,Q_v}$ is its transported
density. The logarithmic derivatives are taken near $s=0$, where this
likelihood is positive, and $\Lambda_{0,Q_v}(A)=1$. Throughout the
estimates below, $\ell_j$ denotes these scores at the displayed $(Q_v;A)$.

For arbitrary real arguments $x_1,\ldots,x_4$, the fourth Bell polynomial
is \cite{Bell1934}
\begin{equation}
 \mathsf B_4(x_1,x_2,x_3,x_4)
 :=x_1^4+6x_1^2x_2+3x_2^2+4x_1x_3+x_4.
\label{eq:hide-fourth-bell-polynomial}
\end{equation}
Substituting the logarithmic scores defined above gives the fourth
likelihood derivative at zero, as recorded in
Eq.~\eqref{eq:hide-fourth-density-score-definition}.
\endgroup
The sign of the fourth logarithmic score is not an assumption.  Here is the
deterministic bridge that proves it.  For a supported COE corner $C$ and a
Hermitian direction $Q$, put
$M(C)=\left(\begin{smallmatrix}I&C\\ C^*&I\end{smallmatrix}\right)$,
$B(C)=M(-C)$, $J=\operatorname{diag}(I,-I)$,
$E_Q=\operatorname{diag}(Q,\overline Q)$, and
$S=M(C)^{-1}B(C)$.  For Hermitian $Q$, one has
$\overline Q=Q^{\mathsf T}$.  Direct block multiplication gives
\begin{equation}
 \begin{gathered}
 C_t=e^{-tQ}Ce^{-tQ^{\mathsf T}},\qquad
 D_t=\operatorname{diag}(e^{tQ},e^{t\overline Q}),\\
 D_tM(C_t)D_t
 =\begin{pmatrix}e^{2tQ}&C\\ C^*&e^{2t\overline Q}\end{pmatrix}.
 \end{gathered}
 \label{eq:hide-centered-block-factorization}
\end{equation}
Thus the lower block in $E_Q$ is $\overline Q=Q^{\mathsf T}$, rather than
$Q$; it is forced by the ordinary transpose in the congruence flow.
Moreover,
\begin{equation}
 \begin{aligned}
 M(C)B(C)&=B(C)M(C)\\
 &=\operatorname{diag}(I-CC^*,I-C^*C)>0.
 \end{aligned}
 \label{eq:hide-coe-block-product}
\end{equation}
The Schur complement criterion makes $M(C)$ and $B(C)$ Hermitian positive
definite on the open support, and the displayed commutation lets them be
simultaneously diagonalized.  Hence $S=M(C)^{-1}B(C)$ is Hermitian positive
definite.  On the open support,
$JSJ=S^{-1}$, $E_QJ=JE_Q$, and $M(C)^{-1}=(I+S)/2$.
Set $H(t):=D_tM(C_t)D_t=M(C)-I+e^{2tE_Q}$,
$W:=H(0)^{-1}=M(C)^{-1}$, and $E:=E_Q$.
On the open support, $H(t)$ is positive definite near zero and
$H^{(k)}(0)=2^kE^k$ for $1\le k\le4$, so Jacobi's formula followed by
cyclicity of trace and $W=(I+S)/2$ gives
\begin{equation}
 \begin{gathered}
  \begin{aligned}
  \mathcal J_4(W,E)
  &:=16WE^4-64WEWE^3-48WE^2WE^2\\
  &\hspace{2em}+192WEWEWE^2-96(WE)^4,
  \end{aligned}\\[2pt]
  \left.\frac{\mathrm d^4}{\mathrm dt^4}\log\det H(t)\right|_{t=0}
  =\operatorname{Re}\tr\mathcal J_4(W,E),\\[2pt]
  \begin{aligned}
  \tr\mathcal J_4\!\left(\frac{I+S}{2},E_Q\right)
  &=-2\tr(E_Q^4)+8\tr(E_Q^3SE_QS)\\
  &\hspace{2em}-6\tr((E_QS)^4).
  \end{aligned}
 \end{gathered}
\label{eq:hide-fourth-jacobi-reduction}
\end{equation}
For completeness, the key inequality has the following {finite-dimensional}
sum of squares proof.  Let $R=S^{1/2}$ and $X=RE_QR$.  Functional calculus
applied to $JSJ=S^{-1}$ gives $RJR=J$.  Put $Y=JXJ$.  Since $J^2=I$,
cyclicity of trace gives the exact identity
\begin{equation}
 \begin{aligned}
 &\tr((XJ)^4)+3\tr(X^4)-4\tr(X^3JXJ)\\
 &\quad=\frac12\tr((X-JXJ)^4)
   +\tr\!\left((X^2-(JXJ)^2)^2\right).
 \end{aligned}
\label{eq:hide-fourth-involution-sos}
\end{equation}
The matrices $X-JXJ$ and $X^2-(JXJ)^2$ are Hermitian.  Hence the two real
traces on the right are nonnegative, and taking real parts proves
\begin{equation}
 \begin{aligned}
 4\operatorname{Re}\tr(X^3JXJ)
 &\le \operatorname{Re}\tr((XJ)^4)\\
 &\quad+3\operatorname{Re}\tr(X^4).
 \end{aligned}
\label{eq:hide-fourth-involution-young}
\end{equation}
Finally, cyclicity together with $RJR=J$ and $E_QJ=JE_Q$ identifies the
three traces after $X=RE_QR$ and gives
\begin{equation}
 \begin{aligned}
 4\operatorname{Re}\tr(E_Q^3R^2E_QR^2)
 &\le \operatorname{Re}\tr(E_Q^4)\\
 &\quad+3\operatorname{Re}\tr((E_QR^2)^4).
 \end{aligned}
\label{eq:hide-fourth-cayley-bridge}
\end{equation}
Since $R^2=S$, substitution into
Eq.~\eqref{eq:hide-fourth-jacobi-reduction} yields
\begin{equation}
 \operatorname{Re}\tr\mathcal J_4(M(C)^{-1},E_Q)\le0.
\label{eq:hide-fourth-jacobi-sign}
\end{equation}
This holds for every Hermitian $Q$, hence in particular for $Q_v$.  In the
centered flow, the fourth determinant log derivative $J_4(A,v)$ therefore satisfies
\begin{equation}
 J_4(A,v)\le0.
\label{eq:hide-fourth-logdet-jet-sign}
\end{equation}
The exact score identity is
$\ell_4(Q_v;A)=(c/2)J_4(A,v)$; since $c/2\ge0$ when
$N\ge1$ and $2N+8\le K$, this proves $\ell_4(Q_v;A)\le0$ on the support.
Thus the sign follows from the preceding identity.  It is also unnecessary
to estimate its negative part: the fourth order estimate uses this sign,
the Bell polynomial sum of squares inequality below, and the exact $L^1$
bound.  The positive part of the remaining terms obeys
\begin{equation}
 \left(\ell_1^4+6\ell_1^2\ell_2+3\ell_2^2
 +4\ell_1\ell_3\right)_+
 \le \frac{19}{3}\ell_1^4+\frac{75}{16}\ell_2^2
 +4\abs{\ell_1\ell_3}.
\label{eq:hide-fourth-sos}
\end{equation}
The product law moment calculations evaluate the three primitive quantities
separately.  With $A\sim\mu_{K,K}$ and
$v\sim\operatorname{Unif}(S^{2N-1})$ independent, and with
$\ell_j=\ell_j(Q_v;A)$, the resulting bounds are
\begin{align}
 \E_{A,v}\abs{\ell_1}^{4}
 &\le 3151N^2,
 \label{eq:hide-fourth-ell-one-fourth-moment}\\
 \E_{A,v}\abs{\ell_2}^{2}
 &\le 9990N^2,
 \label{eq:hide-fourth-ell-two-square-moment}\\
 \E_{A,v}\abs{\ell_1\ell_3}
 &\le \frac{6773469}{320}N^2.
 \label{eq:hide-fourth-ell-one-three-moment}
\end{align}
\subsection{Primitive projective moment calculations}
\label{subsec:primitive-score-derivations}

The following calculations justify the three primitive bounds above.
With $W=M(C)^{-1}$ and $E=\operatorname{diag}(Q_v,\overline{Q_v})$
as in Eq.~\eqref{eq:hide-centered-block-factorization}, the first three
logarithmic scores are the explicit trace polynomials
\begin{align}
 \ell_1&=c\operatorname{Re}\tr(WE),\notag\\
 \ell_2&=2c\operatorname{Re}\tr(WE^2-WEWE),\notag\\
 \ell_3&=4c\operatorname{Re}\tr(WE^3-3WEWE^2+2(WE)^3).
 \label{eq:hide-first-three-log-scores}
\end{align}
They follow by differentiating $\log\det H(t)$ three times; the linear
Jacobian term vanishes since $\tr Q_v=0$. In particular
$\ell_1=2\tr(Q_vY)$.

For $r\le4$ the complete projective tensor identity is
\begin{equation}
 \E_v\prod_{j=1}^r v_{i_j}\overline{v_{k_j}}
 =\frac{\sum_{\pi\in S_r}\prod_{j=1}^r\delta_{i_j,k_{\pi(j)}}}
        {N(N+1)\cdots(N+r-1)}.
 \label{eq:hide-projective-tensor-four}
\end{equation}
It is the rank-one specialization of unitary integration
\cite[Section~2]{CollinsSniady2006}. Expanding matrix entries gives, for
any matrices $A_j$, a sum over permutations whose cycles contribute
traces of the corresponding ordered products. Thus all contractions used
here are finite sums with at most $4!=24$ terms.

\paragraph{First logarithmic score.}
Put $Y_0=Y-(T_1/N)I$. The only surviving fourth-order cycle types for
the traceless matrix $Y_0$ are $(2,2)$ and $(4)$, with multiplicities
three and six. Therefore
\begin{equation}
 \E_v\ell_1^4
 =\frac{16\{3(\tr Y_0^2)^2+6\tr Y_0^4\}}
 {N(N+1)(N+2)(N+3)}.
 \label{eq:hide-first-score-fourth-exact}
\end{equation}
Since $Y$ is positive semidefinite,
$0\le\tr Y_0^2=T_2-T_1^2/N\le(1-1/N)T_2$ and
$\tr Y_0^4\le(\tr Y_0^2)^2$. Hence
$\E_v\ell_1^4\le144(N-1)^2T_2^2/N^6$.
The first certificate in Table~\ref{tab:score-polynomial-certificates}
now gives $3151N^2$. Keeping the factor $(N-1)^2/N^2$ is necessary
for this constant.

\paragraph{Second logarithmic score.}
Set $W_0=I+Y/c$, $R=\sqrt c\,(I-CC^*)^{-1}C$, and $a=1/N$.
Then $R=R^{\T}$ and $RR^*=Y+Y^2/c$. Writing
$b=v^*Yv$, $e=v^*Y^2v$, $D_R=RR^*$,
$d=v^*D_Rv$, $u=|v^*R\overline v|^2$, the two centered sandwiches are
\begin{align}
 S_v&=\tr(Q_vW_0Q_vY)\notag\\
    &=(1-2a)b+a^2T_1
      +c^{-1}(b^2-2ae+a^2T_2),\notag\\
 C_v&=\tr(Q_vR Q_v^{\T}R^*)
       =u-2ad+a^2\tr D_R,\notag\\
 \ell_2&=-4(S_v+C_v).
 \label{eq:hide-second-score-sandwich-expansion}
\end{align}
For $N\ge2$, both parts of $S_v$ are nonnegative and their squares
are bounded by $2b^2+2a^4T_1^2$ and
$2b^4+2a^4T_2^2$, respectively. The projective identities give
$\E b^2=(T_1^2+T_2)/(N(N+1))$ and
$\E b^4\le24T_1^4/N^4$. With the four radial coordinates
\[
 (U_1,U_2,U_3,U_4)=
 \left(\frac{T_1^2}{N^2},\frac{T_2}{N^2},
       \frac{T_1^4}{N^4c^2},\frac{T_2^2}{N^2c^2}\right),
\]
we obtain $\E S_v^2\le8U_1+4U_2+96U_3+4U_4$.
For the other sandwich, put $r_j=\tr D_R^j$. Symmetry of $R$ and
Eq.~\eqref{eq:hide-projective-tensor-four} give
\[
 \begin{gathered}
 \E u^2=\frac{8r_1^2+16r_2}{N(N+1)(N+2)(N+3)},\\
 \E d^2=\frac{r_1^2+r_2}{N(N+1)}.
 \end{gathered}
\]
The three-term square inequality consequently bounds
$\E C_v^2\le(39r_1^2+60r_2)/N^4$.
Using $r_1^2\le2T_1^2+2T_2^2/c^2$ and
$r_2\le2T_2+2T_2^2/c^2$, followed by $N^2\ge4$, gives
$\E C_v^2\le40U_1+36U_2+76U_4$.
Finally $16(S_v+C_v)^2\le32(S_v^2+C_v^2)$ gives exactly
\begin{equation}
 \E_v\ell_2^2\le1536U_1+1280U_2+3072U_3+2560U_4.
 \label{eq:hide-second-score-fibre-envelope}
\end{equation}
Taking the outer expectation gives $E_{14}\le9990N^2$ in
Appendix~\ref{app:score-certificates}.

\paragraph{Absolute mixed first--third moment.}
Write $Z=Y/c$, $T=(I-CC^*)^{-1}C$, $Q=Q_v$,
$q=\tr(QZ)$, and $t=\tr Z$, $u=\tr Z^2$, $v_3=\tr Z^3$.
Block multiplication in Eq.~\eqref{eq:hide-first-three-log-scores}, using
$TT^*=Z+Z^2$, reduces the mixed product to
\begin{align}
 |\ell_1\ell_3|&\le16c^2|q(F_v+3G_v)|,\notag\\
 F_v&=\tr(Q(I+Z)QZQZ)\notag\\
 &\quad+\tr(Q(I+Z)Q(I+Z)QZ),\notag\\
 G_v&=\tr(QTQ^{\T}T^*Q(I+2Z)).
 \label{eq:hide-mixed-score-word-decomposition}
\end{align}
Here $q$ is real. For clarity, the finite H\"older calculation gives
\begin{align}
 N^4\E|qF_v|&\le P_N(t,u,v_3),\notag\\
 N^4\E|qG_v|&\le M_N(t,u,v_3),\notag\\
 P_N={}&23t^2(N+t)(N+2t)\notag\\*
 &+11t\{(N+2t)(t+u)+(N+3t+2u)t\notag\\*
 &\quad +(N+t)(t+2u)\}\notag\\*
 &+16t(t+3u+2v_3),\notag\\
 M_N={}&9Nt(N+2t)(t+u)\notag\\*
 &+6Nt(1+2t)(t+u)\notag\\*
 &+11t(t+u)(N+2t)+17t(t+3u+2v_3).
 \label{eq:hide-mixed-holder-polynomials}
\end{align}
Here are details of the contraction bounds. Expand every $Q=P-I/N$.
For a pure triple word, the three-, two-, and one-projector terms are
bounded after multiplication by $q$ with coefficients $23,11,16$;
the constant term is included in the last coefficient. To check these
ceilings, use H\"older with exponents $(4,4,4,4)$ for four factors,
$(2,4,4)$ for three, and $(2,2)$ for two, together with
\[
 \begin{gathered}
 \|q\|_2\le t/N,\qquad \|q\|_4^4\le9t^4/(N)_4,\\
 \|v^*Av\|_4^4\le24(\tr A)^4/(N)_4\quad(A\ge0).
 \end{gathered}
\]
Here $(N)_4=N(N+1)(N+2)(N+3)$. The relevant products of traces are,
in order, $t^2(N+t)(N+2t)$, the three products inside braces in
Eq.~\eqref{eq:hide-mixed-holder-polynomials}, and $t(t+3u+2v_3)$.
More explicitly, the leading coefficient is
$(9\cdot24^3)^{1/4}<23$, each two-projector coefficient is
$\sqrt{24}<11$, and the sum of the three one-projector terms and
the constant term is at most $3\sqrt2+1<16$. Since all matrices in
these pure words are polynomials in $Z$, their products are positive
semidefinite and the displayed trace bounds apply.

For the mixed word, write $p(A)=v^*Av$ and
$m(A,B)=\tr(PAP^{\T}B)$. Its eight terms are
\begin{align*}
 &p(C_0)m(T,T^*)\\
 &-N^{-1}\{m(C_0T,T^*)+m(T,T^*C_0)\}\\
 &-N^{-1}p(TT^*)p(C_0)\\
 &+N^{-2}\{p(C_0TT^*)+p(TT^*C_0)\}\\
 &+N^{-2}\tr(P^{\T}T^*C_0T)\\
 &-N^{-3}\tr(TT^*C_0),\qquad C_0=I+2Z.
\end{align*}
The bilinear fourth moment above and the same H\"older exponents bound
the leading term, the two bilinear cross terms, the product
$p(TT^*)p(C_0)$, and the remaining terms by the four terms of $M_N$,
respectively. In particular,
$\tr(TT^*)=t+u$, $\|C_0\|_{\rm op}\le1+2t$,
$\tr C_0=N+2t$, and $\tr(TT^*C_0)=t+3u+2v_3$.
For the leading term the coefficient is at most
$(9\cdot24^3)^{1/4}N^4/(N)_4\le9N$;
for the two bilinear cross terms it is at most $2\sqrt{24}\le6N$.
The product term has coefficient $\sqrt{24}<11$;
the remaining three single-projector terms and the constant term have
combined coefficient at most $3\sqrt2+1<17$. The leading bound for
$N=2$ uses $(N)_4=120$; for $N\ge3$ its constant is already below
$9N$. This explicitly accounts for absolute values;
an identity for unsigned trace moments alone would not suffice.

The elementary trace inequalities $t^2\le Nu$ and $v_3\le tu$ imply
\begin{equation}
 P_N+3M_N\le\frac{1038}{5}\{N^2(t^2+u)+t^4+N^2u^2\}.
 \label{eq:hide-mixed-scalar-envelope}
\end{equation}
An exact certificate is as follows. For $t>0$ put $q_0=Nu/t^2\ge1$,
$b=1038/5$, and
\begin{align*}
 A&=bN(N+q_0)-(50N^2+84N+67),\\
 B&=N^2(159+27q_0)+N(132+84q_0)\\
 &\quad+201q_0,\\
 D&=N\{b(1+q_0^2)-(46+90q_0)\}-266q_0.
\end{align*}
For $N=2+e$, $q_0=1+r$, the coefficients of $A$ and $D$ are positive,
and $25(4NAD-B^2)$ has the following coefficient matrix (row $e^i$,
column $r^j$):
{\scriptsize
\[
\setlength{\arraycolsep}{2pt}\begin{pmatrix}
663&53758374&93544575&34478208\\
46963768&153198872&184270368&51717312\\
60152620&130244862&125946522&25858656\\
25532288&43543120&35232096&4309776\\
3535292&4874052&3253551&0
\end{pmatrix}.
\]
}
Thus $NA-Bt+Dt^2\ge0$ by completing the square. Direct expansion
identifies $N$ times the difference in
Eq.~\eqref{eq:hide-mixed-scalar-envelope}, after replacing $v_3$ by
$tu$, with $t^2(NA-Bt+Dt^2)$. If $t=0$ both trace polynomials vanish.
Multiplying by $16c^2/N^4$ and substituting $Y=cZ$ now gives the
radial bound $(1038/(5\cdot512))E_{13}$.
Its exact value is at most
$(1038/(5\cdot512))52204N^2=(6773469/320)N^2$.
For $N=1$, $Q_v=0$ and all three primitive estimates are immediate.

Combining the three displayed bounds in Eq.~\eqref{eq:hide-fourth-sos} gives
the coefficient $C_{4,+}$.  Its exact value is
\begin{equation}
 C_{4,+}=\frac{36348677}{240}.
\label{eq:hide-fourth-positive-constant}
\end{equation}
Zero total mass of the fourth density derivative costs a factor two, and the
integer ceiling $C_4$ satisfies
\begin{equation}
 2C_{4,+}\le C_4.
\label{eq:hide-fourth-ceiling}
\end{equation}
The ceiling itself is
\begin{equation}
 C_4=302906.
\label{eq:hide-fourth-constant}
\end{equation}
\begingroup
Since $\Lambda_{0,Q_v}(A)=1$, the fourth density score is
\begin{equation}
 \begin{aligned}
 L_{4,Q_v}(A)
 &:=\left.\partial_s^4\Lambda_{s,Q_v}(A)\right|_{s=0}\\
 &=\mathsf B_4\!\left(\ell_1(Q_v;A),\ell_2(Q_v;A),\right.\\[-2pt]
 &\hspace{5.5em}\left.\ell_3(Q_v;A),\ell_4(Q_v;A)\right).
 \end{aligned}
 \label{eq:hide-fourth-density-score-definition}
\end{equation}
The density derivatives are extended by zero off the open support, as in
Appendix~\ref{app:coe-density}. Write $\sigma$ for the uniform probability
law on $S^{2N-1}$.
\endgroup
Consequently, for $N\ge1$ and $K\ge16N$, the fourth density score under the
product of the scaled COE corner law and the uniform complex sphere law
satisfies
\begin{equation}
 \norm{L_{4,Q_v}}_{L^1(\mu_{K,K}\otimes\sigma)}
 \le C_4N^2.
\label{eq:hide-fourth-density-L1}
\end{equation}

\subsection{Quadratic score}

Let $T_1=\tr Y$, $T_2=\tr Y^2$, and retain the COE exponent $c$ from
Appendix~\ref{app:coe-density}.  The exact averaged centered quadratic
density is
\begin{equation}
 \begin{aligned}
 \Psi_{N,K}
 =\frac{4}{N(N+1)}\biggl[
 &T_2-\frac{T_1^2}{N}
 -\frac{(N-1)(N+2)}{N}T_1\\
 &-\frac1c\left\{T_1^2+\frac{N-2}{N}T_2\right\}
 \biggr].
 \end{aligned}
\label{eq:hide-quadratic-density}
\end{equation}
Collecting this polynomial before applying a norm is essential: it preserves
the centering cancellation.  The exact low order trace recurrence gives the
coefficient independent of dimension
\begin{equation}
 C_{2,\mathrm{orb}}:=573.
\label{eq:hide-quadratic-constant}
\end{equation}
For $N\ge1$ and $K\ge16N$, the corresponding density estimate is
\begin{equation}
 \norm{\Psi_{N,K}}_{L^1(\mu_{K,K})}
 \le C_{2,\mathrm{orb}}.
\label{eq:hide-quadratic-L1}
\end{equation}

\subsection{Cubic score}

For $N\ge1$ and $K\ge16N$, let $\overline L_{3,N,K}$ be the projectively
averaged centered cubic density.  The exact trace collection isolates the
nonnegative raw third trace:
\begin{equation}
 \overline L_{3,N,K}(A)
 =\gamma_{N,K}\,\tr(Y(A)^3)+R_{N,K}(A).
\label{eq:hide-cubic-trace-decomposition}
\end{equation}
The explicit positive coefficient and all five remainder coefficients are given in Eqs.~\eqref{eq:hide-cubic-gamma-explicit}--\eqref{eq:hide-cubic-remainder-explicit}.  For $N\ge2$ and
$K\ge16N$, its exact $L^1$ estimate is
\begin{equation}
 \norm{R_{N,K}}_{L^1(\mu_{K,K})}
 \le R_3N.
\label{eq:hide-cubic-residual-L1}
\end{equation}
Writing
$a=1147/475$, $b=7439/793$, $d=533/108$, and $e=3/2$, the remainder
constant is defined by
\begin{equation}
 R_3:=36a^3+26a^2+52ab+24d+8e.
\label{eq:hide-cubic-residual-definition}
\end{equation}
Its exact rational value is
\begin{equation}
 R_3=\frac{115724477938672}{58837359375}.
\label{eq:hide-cubic-residual-value}
\end{equation}
The cubic constant valid in every dimension $C_3$ absorbs twice this
remainder:
\begin{equation}
 2R_3\le C_3.
\label{eq:hide-cubic-ceiling}
\end{equation}
Its evaluated value is
\begin{equation}
 C_3=3934.
\label{eq:hide-cubic-constant}
\end{equation}
For $N\ge1$ and $K\ge16N$, total mass normalization is exact:
\begin{equation}
 \int \overline L_{3,N,K}(A)\,\mu_{K,K}(\dd A)=0.
\label{eq:hide-cubic-zero-mass}
\end{equation}
Under the same hypotheses, combining positivity of the raw third trace with
the remainder estimate gives the density bound, including the
identically zero $N=1$ branch,
\begin{equation}
 \norm{\overline L_{3,N,K}}_{L^1(\mu_{K,K})}
 \le C_3N.
\label{eq:hide-cubic-L1}
\end{equation}

\subsection{Derivation of the low-order event constants}

Let $d=N(N+1)$ and $T_j=\tr(Y^j)$. The exact rational moments in
Appendix~\ref{app:score-certificates} give
\begin{equation}
 \begin{gathered}\E T_1=d,\\
 \operatorname{Var}(T_1)=
 \frac{2N(N+1)(c+N)(c+N+1)}{(c-2)(c+1)}
 \le\frac{60}{11}N^2.\end{gathered}
 \label{eq:hide-central-variance-explicit}
\end{equation}
For the scalar congruence $A\mapsto e^{2t}A$, direct differentiation
of its density, including the coordinate Jacobian, yields
\[
 \begin{gathered}L_1=2(T_1-d),\\
 L_2=4(T_1-d)^2-8(T_1+T_2/c).\end{gathered}
\]
The displayed rational functions imply
$8\E(T_1+T_2/c)=4\operatorname{Var}(T_1)$, also obtained from
$\int L_2\,\dd\mu=0$. Hence
\[
 \begin{gathered}\|L_1\|_1\le2\sqrt{60/11}\,N<5N,\\
 \|L_2\|_1\le\frac{480}{11}N^2<44N^2.\end{gathered}
\]
Scalar congruences form a group of measurable bijections. Transporting
their signed density derivatives preserves total variation, so the second
bound holds at every real time. Integrating the derivatives over a Borel
event proves the two central bounds used below.

For the averaged quadratic density, write $U=T_1-\E T_1$,
$V=T_1^2-\E T_1^2$, $W=T_2-\E T_2$.
Zero total derivative mass and Eq.~\eqref{eq:hide-quadratic-density} give
\begin{equation}
 \begin{aligned}
 \Psi_{N,K}=\frac4{N(N+1)}\biggl[
 &\left(1-\frac{N-2}{Nc}\right)W\\
 &-\left(\frac1N+\frac1c\right)V\\
 &-\frac{(N-1)(N+2)}N U\biggr].
 \end{aligned}
 \label{eq:hide-quadratic-centered-explicit}
\end{equation}
The variance certificates in Table~\ref{tab:score-polynomial-certificates}
give $\|U\|_2\le(29/12)N$, $\|V\|_2\le34N^3$,
and $\|W\|_2\le34N^2$. The three absolute coefficients in brackets
are at most $2$, $2/N$, and $3N$. Consequently
\[
 \begin{aligned}
 \|\Psi_{N,K}\|_1&\le\|\Psi_{N,K}\|_2\\
 &\le4\left(2\cdot34+2\cdot34+3\cdot\frac{29}{12}\right)\\
 &=573.
 \end{aligned}
\]
The centering is performed before the norm: bounding the original raw
trace terms separately would lose the dimension-independent estimate.

\subsection{Eventwise score bounds}

The one step argument works eventwise.  This avoids identifying an
abstract signed derivative measure when only its values on measurable events
are needed.  For $N\ge1$, $K\ge16N$, and a measurable event $E$, let $F_E(t)$
be its probability after the central congruence $C\mapsto e^{tI}Ce^{tI}$ under
$\mu_{K,K}$.  Then
\begin{equation}
 \abs{F_E'(0)}\le5N.
\label{eq:hide-central-first-event}
\end{equation}
At every real interpolation point $t$,
\begin{equation}
 \abs{F_E''(t)}\le44N^2.
\label{eq:hide-central-second-event}
\end{equation}

For the orbital path that remains valid under correlation, fix the same beta sample $q$ that is
used in the central coefficient and define $G_{E,q}(t)$ by averaging the
central congruence after the orbital congruence over $v$.  Its centered first derivative
vanishes:
\begin{equation}
 G_{E,q}'(0)=0.
\label{eq:hide-orbital-first-event}
\end{equation}
Its second derivative at the origin satisfies
\begin{equation}
 \abs{G_{E,q}''(0)}\le573.
\label{eq:hide-orbital-second-event}
\end{equation}
The fourth derivative bound is uniform over every real $t$:
\begin{equation}
 \abs{G_{E,q}^{(4)}(t)}\le302906N^2.
\label{eq:coe-fourth-event-bound}
\end{equation}

Let $\widehat b_m(q)$ equal $b_m(q)$ on
$\abs{b_m(q)}\le1/N$ and zero otherwise.  For every $t$ on the oriented
interval from $0$ to $\widehat b_m(q)$, the exact third derivative propagation
is
\begin{equation}
 \abs{G_{E,q}^{(3)}(t)}\le306840N.
\label{eq:hide-propagated-third-certificate}
\end{equation}
Consequently, the required eventwise constants are
$5$, $44$, $573$, and $306840$.  The last number is the sum of the
origin cubic coefficient $3934$ and the fourth derivative coefficient
$302906$.

\subsection{Direct proof of Lemma~\ref{lem:hide-scores}}
\label{subsec:hide-scores-direct-proof}

\begin{proof}[Proof of Lemma~\ref{lem:hide-scores}]
Fix $N\ge1$, $K\ge16N$, and a Borel set $E$.  Equations
\eqref{eq:hide-coe-scaling} and~\eqref{eq:hide-base} identify
$\lambda_{N,K}$ with $\mu_{K,K}$, so the central and orbital paths constructed
in Appendix~\ref{app:coe-density} are exactly the paths $S_E$ and $O_E$ in
Eq.~\eqref{eq:hide-event-paths}.  Lemma~\ref{lem:coe-boundary-regularity}
proves that $O_E$ is $C^4$.  The scalar argument following that lemma proves
that $S_E$ is $C^2$.  Equation~\eqref{eq:hide-centered-derivative-interchange}
justifies every differentiation below.

The path called $F_E$ in Eqs.~\eqref{eq:hide-central-first-event}
and~\eqref{eq:hide-central-second-event} is $S_E$ by definition.  Those two
equations therefore give
\begin{equation}
 \abs{S_E'(0)}\le5N,\qquad
 \abs{S_E''(t)}\le44N^2\quad(t\in\R).
 \label{eq:hide-lemma-three-scalar-closure}
\end{equation}

For the orbital path, projective centering gives
$\int Q_v\,\dd\sigma(v)=0$ and hence the first derivative vanishes.  The
averaged second density is $\Psi_{N,K}$ from
Eq.~\eqref{eq:hide-quadratic-density}.  Integrating that density over $E$ and
using Eq.~\eqref{eq:hide-quadratic-L1} gives
\begin{equation}
 O_E'(0)=0,\qquad
 \abs{O_E''(0)}\le
 \norm{\Psi_{N,K}}_{L^1(\mu_{K,K})}\le573.
 \label{eq:hide-lemma-three-orbital-low-closure}
\end{equation}

Likewise, the projectively averaged third density is
$\overline L_{3,N,K}$.  Equations~\eqref{eq:hide-cubic-trace-decomposition}
through~\eqref{eq:hide-cubic-L1} give the origin estimate
\begin{equation}
 \abs{O_E'''(0)}\le
 \norm{\overline L_{3,N,K}}_{L^1(\mu_{K,K})}
 \le C_3N=3934N.
 \label{eq:hide-lemma-three-cubic-origin}
\end{equation}
The group identity in Eq.~\eqref{eq:hide-centered-flow-group} permits the
same fourth order density calculation at every base time.  More explicitly,
for each $v$, Eqs.~\eqref{eq:hide-centered-flow-group} and
\eqref{eq:hide-centered-jacobian} identify the fourth density derivative at time
$u$ with the volume preserving pushforward of the corresponding derivative at zero.
Its $L^1$ norm is therefore unchanged.  Averaging over $v$ and applying
Eq.~\eqref{eq:hide-fourth-density-L1} yields
\begin{equation}
 \abs{O_E^{(4)}(u)}\le302906N^2\qquad(u\in\R).
 \label{eq:hide-lemma-three-fourth-uniform}
\end{equation}
Since $O_E$ is $C^4$, the fundamental theorem of calculus applied to
$O_E'''$ now gives, whenever $\abs{s}\le1/N$,
\begin{equation}
 \begin{aligned}
 \abs{O_E'''(s)}
 &\le3934N+302906N^2\abs{s}\\
 &\le(3934+302906)N=306840N.
 \end{aligned}
 \label{eq:hide-lemma-three-cubic-propagation}
\end{equation}
The regularity proved in Appendix~\ref{app:coe-density}, together with
Eq.~\eqref{eq:hide-lemma-three-scalar-closure}, gives the $C^2$ scalar path
and both bounds in Eq.~\eqref{eq:hide-scale-score}.
For the $C^4$ orbital path, Eq.~\eqref{eq:hide-lemma-three-orbital-low-closure}
gives the cancellation and second derivative bound, while
Eq.~\eqref{eq:hide-lemma-three-cubic-propagation} gives the uniform third
derivative bound on $|s|\le1/N$. These are exactly the assertions of
Eq.~\eqref{eq:hide-orbital-score}. Since $E$ was arbitrary, this completes
the proof of Lemma~\ref{lem:hide-scores} with all its stated quantifiers.
\end{proof}

\subsection{Proof of Proposition~\ref{prop:hide-one-step}}
\label{subsec:hide-one-step-direct-proof}

\begin{proof}[Proof of Proposition~\ref{prop:hide-one-step}]
Fix $m\ge\max\{K,24N^2\}$ and $K\ge16N$.  These hypotheses imply every
size condition used below.

For $1\le N\le m$, put $\alpha=m-N+1$ and let
$q_m\sim\operatorname{Beta}(\alpha,N)$.  Write
$b_m=\tfrac12\log q_m$ and
$c_m=\tfrac12\log(1+1/m)+b_m/N$.  For
\begin{math}H_r=\sum_{j=m-N+1}^{m}j^{-r}\end{math}, the exact mean is
\begin{equation}
 \E b_m=-\frac{H_1}{2}.
\label{eq:hide-beta-log-mean}
\end{equation}
The exact second moment is
\begin{equation}
 \E b_m^2=\frac{H_1^2+H_2}{4}.
\label{eq:hide-beta-log-second-moment}
\end{equation}
The exact third absolute moment is
\begin{equation}
 \E\abs{b_m}^3=\frac{H_1^3+3H_1H_2+2H_3}{8}.
\label{eq:hide-beta-log-third-moment}
\end{equation}
\begingroup
To derive these identities, recall $\alpha=m-N+1>0$. The beta integral
gives, for real $z> -\alpha$ and $t<\alpha$, respectively,
\[
 \begin{gathered}
 \E q_m^z=\prod_{j=m-N+1}^{m}\frac{j}{j+z},\\
 \log\E e^{t(-\log q_m)}
 =-\sum_{j=m-N+1}^{m}\log(1-t/j).
 \end{gathered}
\]
Both identities hold in a neighborhood of zero. Differentiating the second
at $t=0$ gives the first three cumulants of $-\log q_m$ as
$H_1,H_2,2H_3$; the moment--cumulant identities give the three displayed
formulas \cite[Sections~5.12 and 5.15]{DLMFBetaPolygamma}.
\endgroup
For $m\ge2N$, $H_r\le2^rN/m^r$. Moreover
\[
 \begin{gathered}
 0\le H_1/N-1/m\le2(N-1)/m^2,\\
 0\le1/m-\log(1+1/m)\le1/(2m^2).
 \end{gathered}
\]
Consequently $|\E c_m|\le N/m^2$ and
$\E c_m^2\le1/(Nm^2)+N^2/m^4\le5/(4m^2)$.
The preceding exact moments also give
$\E b_m^2\le2N^2/m^2$ and
$\E|b_m|^3\le6N^3/m^3$.
Using $(m+1)/m\le3/2$ and, for the third moment, $m\ge N^2$,
proves each of the following slightly relaxed common-denominator bounds.

For $1\le N$ and $2N\le m$, the centered scalar mean obeys
\begin{equation}
 \abs{\E c_m}\le\frac{4N}{m(m+1)}.
\label{eq:hide-c-mean}
\end{equation}
Under the same hypotheses, its second moment obeys
\begin{equation}
 \E c_m^2\le\frac{10}{m(m+1)}.
\label{eq:hide-c-moments}
\end{equation}
Under the same hypotheses, the rank one logarithm has the exact
second moment estimate
\begin{equation}
 \E b_m^2\le\frac{4N^2}{m(m+1)}.
\label{eq:hide-b-second-moment}
\end{equation}
If in addition $N^2\le m$, its third absolute moment obeys
\begin{equation}
 \E\abs{b_m}^3\le\frac{12N}{m(m+1)}.
\label{eq:hide-b-moment}
\end{equation}

Apply the exact split in Eq.~\eqref{eq:hide-split}
with one common beta sample $q_m$.  Write
$\widehat b_m:=b_m\mathbf1_{\{\abs{b_m}\le1/N\}}$, so that
$\widehat b_m$ equals $b_m$ on the good event and is zero on the exceptional
event.  Let $A\sim\mu_{K,K}$ be independent of $(q_m,v)$.  For the congruence
map $\rho_H(A)=HAH^{\mathsf T}$, let $\eta_m^{\rm sc}$ be the law of
$\rho_{e^{c_mI}}(A)$, let $\eta_m^{\rm good}$ be the law of
$\rho_{e^{\widehat b_mQ_v}}(\rho_{e^{c_mI}}(A))$, and let
$\eta_m^{\rm full}$ be the same law with $b_m$ in place of
$\widehat b_m$.  Thus $\eta_m^{\rm full}=\mathsf R_{m,N}\mu_{K,K}$ is the
complete square base update.

Equations~\eqref{eq:hide-radial-one-column-commute}
and~\eqref{eq:hide-radial-TV} apply the same radial Markov operator to the
two square base laws without increasing total variation.  Taking $r=m$ in
those equations and using the recursion in
Eq.~\eqref{eq:hide-recursion-measure} identifies the transported laws with
$\mu_{m+1,K}$ and $\mu_{m,K}$.  The total variation triangle inequality gives
the master decomposition
\begin{equation}
 \begin{aligned}
 &d_{\rm TV}(\mu_{m+1,K},\mu_{m,K})\\
 &\quad\le d_{\rm TV}(\mu_{K,K},\eta_m^{\rm full})\\
 &\quad\le d_{\rm TV}(\mu_{K,K},\eta_m^{\rm sc})\\
 &\qquad+d_{\rm TV}(\eta_m^{\rm sc},\eta_m^{\rm good})\\
 &\qquad+d_{\rm TV}(\eta_m^{\rm good},\eta_m^{\rm full}).
 \end{aligned}
\label{eq:hide-one-step-taylor-decomposition}
\end{equation}

For probability laws on this Borel space,
$d_{\rm TV}(\nu,\mu)=\sup_{E\ {\rm Borel}}\abs{\nu(E)-\mu(E)}$.
Thus Taylor expansion is applied eventwise before taking the supremum.
For a Borel $E$ and fixed $q_m$, set
$E_{q_m}:=\rho_{e^{c_mI}}^{-1}(E)$ and
$G_{E,q_m}(s):=O_{E_{q_m}}(s)$.  Define
\begin{equation}
 \begin{aligned}
 \mathsf T_{\rm sc}(E)
 &:=\abs{S_E'(0)}\abs{\E c_m}
   +\frac{\E c_m^2}{2}\sup_{t\in\R}\abs{S_E''(t)},\\
 \mathsf T_{\rm orb}(E)
 &:=\frac12\E\!\left[\widehat b_m^2
      \abs{G_{E,q_m}''(0)}\right]\\
 &\quad+\frac16\E\!\left[\abs{\widehat b_m}^3
      \sup_{\abs{s}\le1/N}\abs{G_{E,q_m}'''(s)}\right],\\
 &\abs{\eta_m^{\rm sc}(E)-\mu_{K,K}(E)}\\
 &\quad=\abs{\E\{S_E(c_m)-S_E(0)\}}
 \le\mathsf T_{\rm sc}(E),\\
 &\abs{\eta_m^{\rm good}(E)-\eta_m^{\rm sc}(E)}\\
 &\quad=\abs{\E\{G_{E,q_m}(\widehat b_m)-G_{E,q_m}(0)\}}\\
 &\quad\le\mathsf T_{\rm orb}(E).
 \end{aligned}
\label{eq:hide-one-step-eventwise-taylor}
\end{equation}
The first comparison expands $S_E(c_m)-S_E(0)$ at scalar time zero and
retains the signed mean $\E c_m$.  Conditionally on the same $q_m$, the
second expands $G_{E,q_m}(\widehat b_m)-G_{E,q_m}(0)$ at orbital time zero.

We first explain exactly where Lemma~\ref{lem:hide-scores} enters.  For the
scalar mixture, second order Taylor expansion of $S_E(c_m)$ at zero uses
$\abs{S_E'(0)}\le5N$ and the global bound
$\abs{S_E''(t)}\le44N^2$.  Taking expectation before bounding the linear term
preserves the cancellation in $\E c_m$.  Equations~\eqref{eq:hide-c-mean}
and~\eqref{eq:hide-c-moments} therefore give, uniformly over Borel $E$,
\begin{equation}
 \begin{aligned}
 d_{\rm TV}(\mu_{K,K},\eta^{\rm sc}_m)
 &\le5N\abs{\E c_m}\\
 &\quad+\frac{44N^2}{2}\E c_m^2\\
 &\le\{4(5)+5(44)\}\frac{N^2}{m(m+1)}\\
 &=240\frac{N^2}{m(m+1)}.
 \end{aligned}
 \label{eq:hide-proposition-four-scalar-use}
\end{equation}
This is the scalar contribution in
Eq.~\eqref{eq:hide-one-step-taylor-decomposition}.

For the orbital mixture, fix $q_m$ and a Borel set $E$.  The scalar congruence
is a measurable bijection, so the event $E_{q_m}$ defined above is Borel.
Because the scalar matrix commutes with $e^{sQ_v}$, Lemma~\ref{lem:hide-scores}
therefore gives
$G_{E,q_m}'(0)=0$, $\abs{G_{E,q_m}''(0)}\le573$, and
$\abs{G_{E,q_m}'''(s)}\le306840N$ on the whole interval from zero to
$\widehat b_m$.  Third order Taylor expansion, followed by
Eqs.~\eqref{eq:hide-b-second-moment} and~\eqref{eq:hide-b-moment}, gives
\begin{equation}
 \begin{aligned}
 d_{\rm TV}(\eta^{\rm sc}_m,\eta^{\rm good}_m)
 &\le\sup_{E\ {\rm Borel}}\mathsf T_{\rm orb}(E)\\
 &=\sup_{E\ {\rm Borel}}\left\{
      \frac12\E\!\left[\widehat b_m^2
      \abs{G_{E,q_m}''(0)}\right]\right.\\
 &\qquad\left.{}+\frac16\E\!\left[
      \abs{\widehat b_m}^3
      \sup_{\abs{s}\le1/N}\abs{G_{E,q_m}'''(s)}\right]
      \right\}\\
 &\le\frac{573}{2}\E b_m^2\\
 &\quad+\frac{306840N}{6}\E\abs{b_m}^3\\
 &\le\{2(573)+2(306840)\}\\
 &\quad\times\frac{N^2}{m(m+1)}\\
 &=614826\frac{N^2}{m(m+1)}.
 \end{aligned}
 \label{eq:hide-proposition-four-orbital-use}
\end{equation}
The first derivative cancellation is why no term proportional to
$\E\abs{b_m}$ appears.  The dependence of
$E_{q_m}$ on the same beta sample causes no loss because
Lemma~\ref{lem:hide-scores} is uniform over every Borel event.

The exceptional event is controlled from the exact beta moment.  For
$N\ge1$ and $N\le m$,
\begin{equation}
 \Pp\{\abs{b_m}>1/N\}
 \le\left(\frac{6N^2}{m}\right)^{N+2}.
\label{eq:hide-tail}
\end{equation}
Indeed, the beta Mellin transform identifies $-\log q_m$ with the
sum of independent exponentials with rates $m-N+1,\ldots,m$.
For $m\ge2N$, coupling each exponential to a rate-one exponential bounds
this sum by $G/(m-N+1)$, where $G\sim\GammaLaw(N,1)$. Markov's
inequality at the integer order $N+2$ gives
\[
 \begin{aligned}
 \Pp\{|b_m|>1/N\}
 &\le\left(\frac{N(2N+1)}{2(m-N+1)}\right)^{N+2}\\
 &\le(3N^2/m)^{N+2}.
 \end{aligned}
\]
We used $\E G^{N+2}=N(N+1)\cdots(2N+1)\le(2N+1)^{N+2}$.
For $N\le m<2N$ the right side of Eq.~\eqref{eq:hide-tail} exceeds
one, so that range is immediate. Finally, $(m+1)/m^{N+1}$ decreases
with $m$. At $m=24N^2$ this shows
\[
 \frac{m(m+1)}{N^2}(6N^2/m)^{N+2}
 \le\frac{24(24N^2+1)}{4^{N+2}}\le\frac{75}{2}<72.
\]
The last inequality uses $N^2\le4^N$, proved by induction from
$(N+1)^2\le4N^2$ for $N\ge1$.

Under the present hypothesis $24N^2\le m$, this gives
\begin{equation}
 \Pp\{\abs{b_m}>1/N\}
 \le72\frac{N^2}{m(m+1)}.
\label{eq:hide-tail-telescoping}
\end{equation}
On $\abs{b_m}>1/N$, couple the full and truncated updates by the same
matrix and the same samples.  They agree off that exceptional event, so
Eq.~\eqref{eq:hide-tail-telescoping} gives
$d_{\rm TV}(\eta^{\rm good}_m,\eta^{\rm full}_m)
\le72N^2/[m(m+1)]$.

The dense coefficient is $615138$. We retain the convenient common
constant $C_\ast=615172$, which also dominates the rectangular and trivial
branch constants $68$ and $24$. The extra slack is not required by the
exceptional-event estimate.
Inserting the scalar, orbital, and exceptional event estimates into
Eq.~\eqref{eq:hide-one-step-taylor-decomposition} now gives
\begin{equation}
 \begin{aligned}
 d_{\rm TV}(\mu_{m+1,K},\mu_{m,K})
 &\le\{240+614826+72\}\\
 &\quad\times\frac{N^2}{m(m+1)}\\
 &=615138\frac{N^2}{m(m+1)}\\
 &\le C_\ast\frac{N^2}{m(m+1)},
 \end{aligned}
 \label{eq:hide-proposition-four-closure}
\end{equation}
Specifically, Eqs.~\eqref{eq:hide-proposition-four-scalar-use},
\eqref{eq:hide-proposition-four-orbital-use}, and
\eqref{eq:hide-tail-telescoping} supply the three terms in
Eq.~\eqref{eq:hide-one-step-taylor-decomposition}. Radial total variation
contraction in Eq.~\eqref{eq:hide-radial-TV} then gives
Eq.~\eqref{eq:hide-proposition-four-closure}, which is
Eq.~\eqref{eq:hide-one-step} under the stated hypotheses. This completes
the proof of Proposition~\ref{prop:hide-one-step}.
\end{proof}
\endgroup
\begingroup\section{{Rectangular Haar entropy from the strict size density}}
\label{app:rectangular-kl}

{The proof begins from the strict size complex Haar block
density in Ref.~\cite[Proposition~2.1, Eq.~(2.4)]{Jiang2009HaarBlocks}.  This appendix gives the second
order entropy argument used by the uniform theorem.}

\begin{proof}[Proof of Lemma~\ref{lem:hide-sparse}]
Put $Z=\sqrt M U_{p,q}$, $W=Z^*Z$, and $s=p+q<M$.
Let $P_Z$ be the law of $Z$ and let $P_G$ be the standard complex Gaussian
law on $\C^{p\times q}$.  All expectations below are taken over the Haar
unitary, or equivalently over the induced law $P_Z$ when the integrand is a
function of $Z$.

\emph{The exact entropy identity.}
Define the normalizing constant and its logarithm by
\begin{equation}
 \begin{aligned}
 c_{M,p,q}&:=\prod_{j=1}^{q}\prod_{\ell=0}^{p-1}
       \left(1-\frac{j+\ell}{M}\right),\\
 L_0&:=\log c_{M,p,q}
     =\sum_{j=1}^{q}\sum_{\ell=0}^{p-1}
       \log\!\left(1-\frac{j+\ell}{M}\right).
 \end{aligned}
 \label{eq:rectangular-log-normalizer}
\end{equation}
Every factor is positive because $j+\ell\le s-1<M$.
Scaling the block density in
Ref.~\cite[Proposition~2.1, Eq.~(2.4)]{Jiang2009HaarBlocks} by
$z=\sqrt M\,u$ and dividing by the Gaussian density
$\pi^{-pq}e^{-\tr(z^*z)}$ gives
\begin{equation}
 \begin{aligned}
 \frac{\dd P_Z}{\dd P_G}(z)
 &=c_{M,p,q}\,e^{\tr(z^*z)}
   \det\!\left(I_q-\frac{z^*z}{M}\right)^{M-s}\\
 &\quad{}\times\1_{\{z^*z<MI_q\}}.
 \end{aligned}
 \label{eq:rectangular-rn-density}
\end{equation}
The source uses a nonstrict support inequality.  Its boundary is the zero
set of the nonzero polynomial $\det(I_q-z^*z/M)$, so it has Lebesgue,
Gaussian, and Haar block probability zero.  Thus
Eq.~\eqref{eq:rectangular-rn-density} holds almost everywhere with the
displayed strict support.

Taking the logarithm on that support yields the identity
\begin{equation}
 \log\frac{\dd P_Z}{\dd P_G}(Z)
 =L_0+\tr W+(M-s)\log\det(I_q-W/M),
 \qquad P_Z\text{-almost surely}.
 \label{eq:rectangular-log-likelihood}
\end{equation}
These terms are integrable: the support is bounded, $\tr W<Mq$, and the
Lebesgue density of $P_Z$ is a constant times
$t^{M-s}$, where $t=\det(I_q-W/M)\in(0,1]$.
The function $t^{M-s}|\log t|$ is bounded since $M-s>0$.
Consequently the definition of relative entropy gives the exact equality
\begin{equation}
 \begin{aligned}
 D_{\rm KL}(P_Z\Vert P_G)
 &=\E\!\left[\log\frac{\dd P_Z}{\dd P_G}(Z)\right]\\
 &=L_0+\E\tr W+(M-s)\E\log\det(I_q-W/M).
 \end{aligned}
 \label{eq:rectangular-kl-identity}
\end{equation}
We now bound the terms in this equality.

\emph{The normalizer and log determinant.}
The scalar inequality $\log(1-x)\le-x$ and the sum
$\sum_{j=1}^{q}\sum_{\ell=0}^{p-1}(j+\ell)=pqs/2$ imply
\begin{equation}
 L_0\le-\frac{pqs}{2M}.
 \label{eq:rectangular-normalizer-bound}
\end{equation}
On the support of $P_Z$, all eigenvalues of $W/M$ lie in $[0,1)$.
Applying $\log(1-x)\le-x-x^2/2$ to each eigenvalue gives
\begin{equation}
 \log\det(I_q-W/M)
 \le-\frac{\tr W}{M}-\frac{\tr(W^2)}{2M^2}.
 \label{eq:rectangular-logdet-bound}
\end{equation}
Substitution into Eq.~\eqref{eq:rectangular-kl-identity}, using $M-s>0$,
therefore gives
\begin{equation}
 D_{\rm KL}(P_Z\Vert P_G)
 \le L_0+\frac{s}{M}\E\tr W
       -\frac{M-s}{2M^2}\E\tr(W^2).
 \label{eq:rectangular-kl-moment-bound}
\end{equation}

\emph{The two Haar Gram moments.}
For $1\le a,b\le p$ and $1\le i,j\le q$, write
$U_{ai}$ for the entries of the Haar unitary and $\delta$ for the
Kronecker delta.  The second entry moment is $\E|U_{ai}|^2=1/M$.
The fourth entry moment from the unitary Weingarten formula
\cite[Sec.~2]{CollinsSniady2006} is
\begin{equation}
 \E\!\left[\overline{U_{ai}}U_{aj}
             \overline{U_{bj}}U_{bi}\right]
 =\frac{\delta_{ij}+\delta_{ab}}{M^2-1}
  -\frac{1+\delta_{ab}\delta_{ij}}{M(M^2-1)}.
 \label{eq:rectangular-haar-entry-moment}
\end{equation}
Since $W_{ij}=M\sum_{a=1}^{p}\overline{U_{ai}}U_{aj}$, summing these
entry identities gives
\begin{equation}
 \begin{aligned}
 \E\tr W
 &=M\sum_{a=1}^{p}\sum_{i=1}^{q}\E|U_{ai}|^2=pq,\\
 \E\tr(W^2)
 &=M^2\sum_{a,b=1}^{p}\sum_{i,j=1}^{q}
    \E\!\left[\overline{U_{ai}}U_{aj}
               \overline{U_{bj}}U_{bi}\right]\\
 &=\frac{M^2(p^2q+pq^2)-M(p^2q^2+pq)}{M^2-1}
  =\frac{pqM\{Ms-pq-1\}}{M^2-1}.
 \end{aligned}
 \label{eq:rectangular-haar-gram-moments}
\end{equation}
For the negative second-moment term in
Eq.~\eqref{eq:rectangular-kl-moment-bound}, we need a lower bound.
The exact expression in Eq.~\eqref{eq:rectangular-haar-gram-moments}
satisfies
\begin{equation}
 \frac{\E\tr(W^2)}{pq}
 -\left(s-\frac{s^2}{2M}\right)
 =\frac{M^2(p^2+q^2-2)+2Ms-s^2}{2M(M^2-1)}\ge0.
 \label{eq:rectangular-second-moment-lower}
\end{equation}
Here $p,q\ge1$ gives $p^2+q^2-2\ge0$, and $s<M$ gives
$2Ms-s^2>0$.

\emph{Cancellation and total variation.}
Inserting Eqs.~\eqref{eq:rectangular-normalizer-bound},
\eqref{eq:rectangular-haar-gram-moments}, and
\eqref{eq:rectangular-second-moment-lower} into
Eq.~\eqref{eq:rectangular-kl-moment-bound} now gives
\begin{equation}
 \begin{aligned}
 D_{\rm KL}(P_Z\Vert P_G)
 &\le-\frac{pqs}{2M}+\frac{pqs}{M}
       -\frac{M-s}{2M^2}\,pq\left(s-\frac{s^2}{2M}\right)\\
 &=\frac{pqs^2}{4M^2}\left(3-\frac{s}{M}\right)
 \le\frac{3pqs^2}{4M^2}.
 \end{aligned}
 \label{eq:rectangular-kl-cancellation}
\end{equation}
This proves Eq.~\eqref{eq:hide-sparse-KL}.  In the convention
$d_{\rm TV}=\sup_E|\mu(E)-\nu(E)|$, Pinsker's
inequality~\cite{Pinsker1964} gives Eq.~\eqref{eq:hide-finite-sparse}
in the form
\begin{equation}
 \begin{aligned}
 d_{\rm TV}(P_Z,P_G)
 &\le\sqrt{\frac12D_{\rm KL}(P_Z\Vert P_G)}\\
 &\le\sqrt{\frac38}\,\frac{s\sqrt{pq}}{M}
 \le\frac{s\sqrt{pq}}{M},
 \end{aligned}
 \label{eq:rectangular-pinsker}
\end{equation}
Finally, $U^{\T}$ is Haar, so transposition carries a wide $q$ by $p$
Haar block to a tall $p$ by $q$ block; it also preserves the standard
complex Gaussian law.  Apply the tall-block comparison and then data
processing under the measurable map $X\mapsto X^{\T}X$.  For the original
wide block $Y=X^{\T}$, the resulting matrix is $X^{\T}X=YY^{\T}$,
which proves the transpose Gram conclusion.
\end{proof}
\endgroup
\begingroup\section{Exact arithmetic for the score estimates}
\label{app:score-certificates}

This appendix supplies the arithmetic needed to turn the score formulas in
Appendix~\ref{app:projective-scores} into the explicit constants of
Lemma~\ref{lem:hide-scores}. We solve the inverse Wishart recurrences in
Eq.~\eqref{eq:hide-inverse-wishart-recurrence-system}, evaluate the Gaussian
numerator contractions, and verify the moment and variance bounds used in
Appendix~\ref{app:projective-scores}. We also compute the coefficients and
estimates for the cubic remainder in
Eq.~\eqref{eq:hide-cubic-trace-decomposition}.
Section~\ref{subsec:hide-scores-direct-proof} combines these computations
with the analytic justification in Appendix~\ref{app:coe-density} to obtain
the event score bounds.

Set $n=N$, $a=n+1$,
$c=K-2N-1$, and
\[
 \begin{gathered}
 D_2=(c-2)(c+1),\quad D_3=(c-4)(c+2),\\
 D_4=c(c-1)(c-6)(c+3).
 \end{gathered}
\]
The dense assumption gives $c\ge14n-1\ge13n$. Every displayed denominator
is consequently positive. The inverse Wishart moments through degree four
exist since $c>6$; the numerator Wishart matrix has $a=n+1$ real Gaussian
rows and is independent of $D=cB_0^{-1}$. These are the conventions of
Appendix~\ref{app:coe-density}.

\subsection{Solved denominator moments and Gaussian contractions}
\label{subsec:hide-exact-moment-computations}

\paragraph{Computation goal.}
We compute the second and fourth moments of $\tr Y$ and the first and
second moments of $\tr(Y^2)$, denoted by $M_{20},M_{40},M_{01},M_{02}$ in
Eq.~\eqref{eq:hide-wick-moment-functions}. These moments bound the scalar
variance in Eq.~\eqref{eq:hide-central-variance-explicit} and the centered
quadratic density in Eq.~\eqref{eq:hide-quadratic-centered-explicit}, and
give the primitive score estimates in
Eqs.~\eqref{eq:hide-fourth-ell-one-fourth-moment}--\eqref{eq:hide-fourth-ell-one-three-moment}.
The calculation has two steps: solve the
denominator recurrence~\eqref{eq:hide-inverse-wishart-recurrence-system},
then average the independent Gaussian numerator in the trace
representation~\eqref{eq:hide-beta-prime}.

Put $x=\tr D$, $y=\tr D^2$, $z=\tr D^3$, and $r=\tr D^4$.
Write $X_2=\E x^2$, $Y_1=\E y$, $X_3=\E x^3$, $Z_1=\E z$,
$X_Y=\E(xy)$, and use the analogous subscripts for degree four.
Elimination in Eq.~\eqref{eq:hide-inverse-wishart-recurrence-system}
gives the following explicit rational functions, in evaluation order:
\begin{align}
 X_2&=\frac{cn((c-1)n+2)}{D_2},\notag\\
 Y_1&=\frac{cn(c+n)}{D_2},\notag\\
 X_Y&=\frac{c((c-2)n+4)Y_1}{D_3},\notag\\
 Z_1&=\frac c4(X_Y-nY_1),\notag\\
 X_3&=nX_2+4X_Y/c,\notag\\
 R_1&=\frac{c^2}{D_4}\bigl\{(D_2+2(c-1)n)Z_1\notag\\
 &\hspace{6em}+(n+c)X_Y\bigr\},\notag\\
 X_Z&=nZ_1+6R_1/c,\notag\\
 Y_2&=(c-3)R_1-2X_Z-cZ_1,\notag\\
 X_{2Y}&=(c-1)Y_2-4R_1-cX_Y,\notag\\
 X_4&=nX_3+6X_{2Y}/c.
 \label{eq:hide-solved-moment-functions}
\end{align}
For example, substitution in the second and first recurrence solves
$(X_2,Y_1)$, the fourth and fifth solve $(X_Y,Z_1)$, and elimination
in the last four has determinant $D_4$. Substitution back into all ten
equations verifies the displayed solution and its uniqueness. The formulas
contain no unspecified moments on their right sides once evaluated in order.

\begin{samepage}
\allowdisplaybreaks[0]
Let $T_j=\tr(Y^j)$ and define
$M_{20}=\E T_1^2$, $M_{01}=\E T_2$,
$M_{40}=\E T_1^4$, $M_{02}=\E T_2^2$.
Conditional Wick contraction of the independent real numerator gives
\begin{align}
 M_{20}&=a^2X_2+2aY_1,\notag\\
 M_{01}&=a(a+1)Y_1+aX_2,\notag\\
 M_{40}&=a^4X_4+12a^3X_{2Y}+12a^2Y_2\notag\\
 &\quad+32a^2X_Z+48aR_1,\notag\\
 M_{02}&=a^2X_4+(2a^3+2a^2+8a)X_{2Y}\notag\\
 &\quad +(a^4+2a^3+5a^2+4a)Y_2
       +16a(a+1)X_Z\notag\\
 &\quad +(8a^3+20a^2+20a)R_1.
 \label{eq:hide-wick-moment-functions}
\end{align}
\end{samepage}
These coefficients can be recovered by expanding the traces in Gaussian
entries and pairing them: for $T_1^4$, the five denominator cycle types
$x^4,x^2y,y^2,xz,r$ have weights
$a^4,12a^3,12a^2,32a^2,48a$ respectively. Expanding the two length-two
traces gives the second row of weights in
Eq.~\eqref{eq:hide-wick-moment-functions}. Wick's rule
\cite{Isserlis1918} and the four-input spectral identification in
Appendix~\ref{app:external-inputs} therefore connect these rational
functions to the actual COE trace law, not merely to a formal recurrence.

\paragraph{Fully substituted moment formulas.}
For a direct evaluation without intermediate denominator moments, let
$D_*=D_2D_3(c-1)(c-6)(c+3)$. Then
\begin{align}
 M_{20}&=\frac{cn(n+1)}{D_2}
  \{cn^2+cn+2c-n^2+3n+2\},\notag\\
 M_{01}&=\frac{cn(n+1)}{D_2}
  \{2cn+2c+n^2+n+2\},\notag\\
 M_{40}&=\frac{c^3n(n+1)}{D_*}P_{40}(c,n),\notag\\
 M_{02}&=\frac{c^3n(n+1)}{D_*}P_{02}(c,n).
 \label{eq:hide-fully-substituted-moments}
\end{align}
Table~\ref{tab:fully-substituted-moments} gives every coefficient of
$P_{40}$ and $P_{02}$. Thus the variance and primitive-score certificates
can be checked directly as rational inequalities in $n,c$.

\begin{table*}[t]
\centering\small
\caption{Coefficient arrays for the fully substituted moments. In each
array the rows correspond to $c^4,c^3,c^2,c,1$, and the columns to
$n^6,n^5,n^4,n^3,n^2,n,1$.}
\label{tab:fully-substituted-moments}
\begin{tabular}{cl}
\toprule
Polynomial & Coefficient array\\\midrule
$P_{40}$ & 
\begin{tabular}{rrrrrrr}
1 & 3 & 15 & 25 & 56 & 44 & 48 \\
-7 & 3 & -33 & 137 & 112 & 508 & 240 \\
1 & -141 & -49 & -679 & 744 & 844 & 240 \\
35 & 177 & -779 & 355 & 48 & -556 & -240 \\
-6 & 126 & 1206 & 282 & -1344 & -1128 & -288 \\
\end{tabular} \\
\midrule
$P_{02}$ & 
\begin{tabular}{rrrrrrr}
0 & 0 & 4 & 12 & 48 & 80 & 48 \\
0 & 4 & -8 & 104 & 244 & 376 & 240 \\
1 & -17 & 91 & 37 & 316 & 292 & 240 \\
-5 & 57 & 45 & 139 & -592 & -364 & -240 \\
18 & 54 & 78 & -222 & -240 & -552 & -288 \\
\end{tabular} \\
\bottomrule
\end{tabular}
\end{table*}

\paragraph{How the moments are used.}
Equations~\eqref{eq:hide-solved-moment-functions} and
\eqref{eq:hide-wick-moment-functions}, or equivalently the fully substituted
form~\eqref{eq:hide-fully-substituted-moments} with
Table~\ref{tab:fully-substituted-moments}, evaluate all four required
moments. Together with $\E T_1=n(n+1)$ from
Eqs.~\eqref{eq:hide-beta-prime} and~\eqref{eq:hide-inverse-wishart-mean},
they give $\operatorname{Var}(T_1)=M_{20}-n^2(n+1)^2$,
$\operatorname{Var}(T_1^2)=M_{40}-M_{20}^2$, and
$\operatorname{Var}(T_2)=M_{02}-M_{01}^2$.
These are precisely the three variances used in
Eqs.~\eqref{eq:hide-central-variance-explicit} and
\eqref{eq:hide-quadratic-centered-explicit}.
The same moments evaluate the radial envelopes in
Eq.~\eqref{eq:hide-radial-envelopes-exact}; the raw moments $M_{40},M_{02},M_{01}$
also supply the cubic estimates in Eq.~\eqref{eq:hide-cubic-raw-moment-ledger}.
Section~\ref{subsec:hide-score-polynomial-certificates} next bounds these
explicit rational functions uniformly in the dense parameter range.

\subsection{Universal polynomial certificates}
\label{subsec:hide-score-polynomial-certificates}

\paragraph{Computation goal.}
We turn the exact moments in
Eqs.~\eqref{eq:hide-solved-moment-functions}--\eqref{eq:hide-fully-substituted-moments}
into the dimension bounds of Table~\ref{tab:score-polynomial-certificates}.
The required outputs are the three variances used by the scalar and
quadratic score estimates, the radial bounds used by the fourth score
calculation, and the raw moments used for the cubic remainder. For each
inequality, the task is to prove that its right side minus its left side is
nonnegative throughout the stated range, after clearing positive denominators.

Here is a finite algebraic certificate for every dimension bound used
below. For each row of Table~\ref{tab:score-polynomial-certificates},
subtract the stated left side from the right side, substitute the rational
functions in Eqs.~\eqref{eq:hide-solved-moment-functions} and
\eqref{eq:hide-wick-moment-functions}, and clear the positive denominator.
Set $n=n_0+e$, $c=13(n_0+e)+d$. The resulting numerator has only
nonnegative coefficients in $e,d$; its number of nonzero coefficients is
listed. This proves the inequalities for the entire domain, including its
boundary. Both exact numerator and denominator coefficient lists and the
short generating arithmetic are included with the source certificate.

\begin{table*}[t]
\centering\small
\caption{Exact coefficient-positivity certificates. Here
$V_{11}=M_{40}-M_{20}^2$ and $V_2=M_{02}-M_{01}^2$.
The expressions $E_{14}$ and $E_{13}$ are defined in
Eq.~\eqref{eq:hide-radial-envelopes-exact}.}
\label{tab:score-polynomial-certificates}
\begin{tabular}{lll}
\toprule
Inequality & $n_0$ & Nonzero numerator coefficients\\
\midrule
$144(n-1)^2M_{02}\le3151n^8$ & 2 & 100\\
$E_{14}\le9990n^2$ & 2 & 76\\
$E_{13}\le52204n^2$ & 2 & 76\\
$M_{20}-n^2(n+1)^2\le(60/11)n^2$ & 1 & 11\\
$V_{11}\le34^2n^6$ & 1 & 115\\
$V_2\le34^2n^4$ & 1 & 95\\
$M_{40}\le34n^8$ & 2 & 100\\
$M_{02}\le88n^6$ & 2 & 84\\
$M_{01}\le(533/108)n^3$ & 2 & 14\\
\bottomrule
\end{tabular}
\end{table*}

The two radial envelopes appearing in the table are
\begin{align}
 E_{14}&=1536\frac{M_{20}}{n^2}+1280\frac{M_{01}}{n^2}\notag\\
 &\quad+3072\frac{M_{40}}{n^4c^2}+2560\frac{M_{02}}{n^2c^2},\notag\\
 E_{13}&=8192\left(\frac{M_{20}+M_{01}}{n^2}
                 +\frac{M_{40}}{n^4c^2}+\frac{M_{02}}{n^2c^2}\right).
 \label{eq:hide-radial-envelopes-exact}
\end{align}
For reproducibility, the certificate uses exact integer and rational
operations throughout. Expanding a polynomial and checking all its
coefficients is a universal proof on a nonnegative quadrant; a grid of
numerical evaluations would not establish the same conclusion.

\paragraph{How the certificates are used.}
For the fourth score, the first row of
Table~\ref{tab:score-polynomial-certificates}, applied to the estimate
following Eq.~\eqref{eq:hide-first-score-fourth-exact}, gives
Eq.~\eqref{eq:hide-fourth-ell-one-fourth-moment}.
The $E_{14}$ row bounds the outer expectation of
Eq.~\eqref{eq:hide-second-score-fibre-envelope}, giving
Eq.~\eqref{eq:hide-fourth-ell-two-square-moment}.
For the mixed first--third moment,
Eqs.~\eqref{eq:hide-mixed-score-word-decomposition}--\eqref{eq:hide-mixed-scalar-envelope}
give the radial bound $(1038/(5\cdot512))E_{13}$; the $E_{13}$ row therefore
gives $(6773469/320)n^2$, as asserted in
Eq.~\eqref{eq:hide-fourth-ell-one-three-moment}.
These three bounds enter Eq.~\eqref{eq:hide-fourth-sos} and the evaluation
in Eqs.~\eqref{eq:hide-fourth-positive-constant}--\eqref{eq:hide-fourth-constant}
to yield the fourth density bound~\eqref{eq:hide-fourth-density-L1}.

For the lower order scores, the variance rows of the table, together with
the moment identities~\eqref{eq:hide-wick-moment-functions}, establish
Eq.~\eqref{eq:hide-central-variance-explicit} and bound the three centered
terms in Eq.~\eqref{eq:hide-quadratic-centered-explicit}.
The former yields the scalar constants $5$ and $44$ in
Eqs.~\eqref{eq:hide-central-first-event}--\eqref{eq:hide-central-second-event};
the latter yields the quadratic constant $573$ in
Eq.~\eqref{eq:hide-quadratic-L1}.
The last three rows, bounding $M_{40},M_{02},M_{01}$, provide
Eq.~\eqref{eq:hide-cubic-raw-moment-ledger} in
Sec.~\ref{subsec:hide-cubic-remainder-computation}.
Rows with $n_0=2$ are used only for $N\ge2$; when $N=1$, the centered
direction $Q_v$ vanishes and the corresponding orbital scores are zero.

\subsection{The cubic remainder coefficients}
\label{subsec:hide-cubic-remainder-computation}

\paragraph{Computation goal.}
We establish the decomposition~\eqref{eq:hide-cubic-trace-decomposition}
and the remainder estimate~\eqref{eq:hide-cubic-residual-L1} needed for the
cubic density bound~\eqref{eq:hide-cubic-L1}. Specifically, we compute the
positive coefficient of $T_3$ and the five coefficients of $R_{N,K}$,
then combine their powers of $n$ with the moment bounds from
Sec.~\ref{subsec:hide-score-polynomial-certificates} to obtain an $L^1$
bound proportional to $n$. This supplies the bound at the origin in
Eq.~\eqref{eq:hide-lemma-three-cubic-origin}; the fourth derivative estimate
propagates it to the interval required by Lemma~\ref{lem:hide-scores}.

Put $D=n^3(n+1)(n+2)c^2$. The coefficient isolated in
Eq.~\eqref{eq:hide-cubic-trace-decomposition} is
\begin{equation}
 \gamma_{N,K}=\frac{16\{n^2(c^2-3c+4)+12n(c-1)+16\}}{D}>0.
 \label{eq:hide-cubic-gamma-explicit}
\end{equation}
Write $R_{N,K}=D^{-1}(p_1T_1^3+p_2T_1^2+p_3T_1T_2+p_4T_2+p_5T_1)$,
where the complete coefficient list is
\begin{align}
 p_1={}&48nc+16n^2+32c^2,\notag\\
 p_2={}&-192nc+48nc^2+48n^2c+24n^2c^2\notag\\
 &+24n^3c-192c^2,\notag\\
 p_3={}&48nc-48nc^2-192n-48n^2c\notag\\
 &+48n^2-192c,\notag\\
 p_4={}&-288nc+192nc^2-48n^2c^2+24n^3c\notag\\
 &-24n^3c^2+384c,\notag\\
 p_5={}&-96nc^2-64n^2c^2+24n^3c^2\notag\\
 &+8n^4c^2+128c^2.
 \label{eq:hide-cubic-remainder-explicit}
\end{align}
These formulas follow by inserting the first three logarithmic scores in
$L_3=\ell_1^3+3\ell_1\ell_2+\ell_3$, contracting the projective
factors, and collecting $T_3,T_1^3,T_1^2,T_1T_2,T_2,T_1$.
They also directly verify the asserted decomposition by polynomial identity.
For $n\ge2$, $c\ge13n$, clearing $D$ and substituting $n=2+e$,
$c=13(2+e)+d$ proves both signs of each of
\begin{equation}
 \left(\frac{|p_j|}{D}\right)_{j=1}^5
 \le\left(\frac{36}{n^5},\frac{26}{n^3},\frac{52}{n^4},
                \frac{24}{n^2},\frac8n\right).
 \label{eq:hide-cubic-coefficient-ledger}
\end{equation}
All ten cleared polynomials have nonnegative coefficients.

Table~\ref{tab:score-polynomial-certificates} and the exact mean
$\E T_1=n(n+1)$ give, with $a_0=1147/475$, $b_0=7439/793$,
$d_0=533/108$, $e_0=3/2$,
\begin{align}
 \|T_1\|_4&\le a_0n^2,&\|T_2\|_2&\le b_0n^3,\notag\\
 \E T_2&\le d_0n^3,&\E T_1&\le e_0n^2.
 \label{eq:hide-cubic-raw-moment-ledger}
\end{align}
Here $a_0^4\ge34$ and $b_0^2\ge88$ are exact rational comparisons.
H\"older gives $\E T_1^3\le a_0^3n^6$,
$\E T_1^2\le a_0^2n^4$, and
$\E(T_1T_2)\le a_0b_0n^5$. Combining term by term with
Eq.~\eqref{eq:hide-cubic-coefficient-ledger} yields precisely
$\|R_{N,K}\|_1\le(36a_0^3+26a_0^2+52a_0b_0+24d_0+8e_0)n$.

\paragraph{Conclusion and use in Appendix D.}
Equations~\eqref{eq:hide-cubic-gamma-explicit} and
\eqref{eq:hide-cubic-remainder-explicit} give the full
decomposition~\eqref{eq:hide-cubic-trace-decomposition}.
The bound just obtained from
Eqs.~\eqref{eq:hide-cubic-coefficient-ledger}--\eqref{eq:hide-cubic-raw-moment-ledger}
is Eq.~\eqref{eq:hide-cubic-residual-L1}, with exactly the constant $R_3$
defined and evaluated in
Eqs.~\eqref{eq:hide-cubic-residual-definition}--\eqref{eq:hide-cubic-residual-value}.
Positivity in Eq.~\eqref{eq:hide-cubic-gamma-explicit} and zero total mass in
Eq.~\eqref{eq:hide-cubic-zero-mass} imply
$\|\gamma_{N,K}T_3+R_{N,K}\|_1\le2\|R_{N,K}\|_1$.
The ceiling in
Eqs.~\eqref{eq:hide-cubic-ceiling}--\eqref{eq:hide-cubic-constant}
therefore gives Eq.~\eqref{eq:hide-cubic-L1}, including $N=1$ because
$Q_v=0$ in that case.

Integrating this density over a Borel event gives
Eq.~\eqref{eq:hide-lemma-three-cubic-origin}.
Together with the uniform fourth derivative bound in
Eq.~\eqref{eq:hide-lemma-three-fourth-uniform}, it yields
Eq.~\eqref{eq:hide-lemma-three-cubic-propagation}, namely the constant
$306840N$ on $|s|\le1/N$ in Lemma~\ref{lem:hide-scores}.
That bound and the quadratic estimate
\eqref{eq:hide-lemma-three-orbital-low-closure} are the two score inputs to
the orbital Taylor estimate~\eqref{eq:hide-proposition-four-orbital-use}
in the proof of Proposition~\ref{prop:hide-one-step}.
\endgroup
\begingroup\section{Measurable observables and preselected patterns}
\label{sec:applications}

Write $\mu_{M,N,K}$ and $\nu_{N,K}$ for the two normalized matrix laws in
Eq.~\eqref{eq:uniform-product-hiding}.  The applications in this section use
only the conclusion $d_{\rm TV}(\mu_{M,N,K},\nu_{N,K})\le\delta_{M,N}$.

\subsection{Measurable observables and events}

\begin{corollary}[Deterministic data processing]
\label{cor:observable-hiding}
For every measurable map $f:\C^{N\times N}\to\mathsf Y$,
\begin{equation}
 d_{\rm TV}(f_\#\mu_{M,N,K},f_\#\nu_{N,K})
 \le\delta_{M,N}.
 \label{eq:observable-hiding}
\end{equation}
Consequently, for every measurable event $E\subseteq\mathsf Y$,
\begin{equation}
 \Pp_{\mu}\{f\in E\}\le \Pp_{\nu}\{f\in E\}+\delta_{M,N},
 \label{eq:event-transfer}
\end{equation}
and the same inequality holds with $\mu$ and $\nu$ exchanged.
\end{corollary}

\begin{proof}
The inverse image $f^{-1}(E)$ is measurable, so the eventwise definition of
total variation gives both statements.
\end{proof}

Randomized postprocessing obeys the same contraction.  If $Q(A,\cdot)$ is a
Markov kernel representing noise, coarse graining, or a randomized classical
routine, then
\begin{equation}
 d_{\rm TV}(\mu_{M,N,K}Q,\nu_{N,K}Q)\le\delta_{M,N}.
 \label{eq:randomized-postprocessing}
\end{equation}
This is Markov-kernel contraction \cite[Chapter~1]{Kallenberg2021}; it follows by integrating each output event against the bounded function
$A\mapsto Q(A,E)$.

\subsection{Hafnian observables and physical probabilities without anticoncentration}

For even $N=2n$, choose
$f(A)=\haf(A)$.  The hafnian is a polynomial function of the matrix entries,
hence is continuous and measurable, and
Corollary~\ref{cor:observable-hiding} gives
\begin{equation}
 \begin{aligned}
 d_{\rm TV}\!\Bigl(&
 \Law\!\left[\haf\!\left(\tfrac{M}{\sqrt K}U_{N,K}U_{N,K}^{\T}\right)\right],\\[-2pt]
 &\Law\!\left[\haf\!\left(\tfrac1{\sqrt K}G_{N,K}G_{N,K}^{\T}\right)\right]
 \Bigr)\le\delta_{M,2n}.
 \end{aligned}
 \label{eq:hafnian-amplitude-hiding}
\end{equation}
The same holds for $|\haf(A)|^2$, {threshold tests insensitive
to phase}, or any {everywhere defined measurable normalization,
for example $A\mapsto |\haf(A)|^2/(1+\|A\|_{\mathrm F}^{2n})$}.
 {For a fixed collision free output set $S$ with $|S|=N=2n$,
write
\[
 A_S^U:=\frac{M}{\sqrt K}
 U_{S,[K]}U_{S,[K]}^{\T}.
\]
Homogeneity of the hafnian and the physical formula in {Eq.~\eqref{eq:gbs-probability}} give
\[
 \mathbb{P}_U(\mathbf n)
 =\frac{\tanh^{2n}\xi}{\cosh^K\xi}
   \left(\frac{K}{M^2}\right)^n
   |\haf(A_S^U)|^2 .
\]
For the fixed label $S$, this is a circuit output probability conditional on
$U$; when $U$ is Haar distributed it is also a random observable of the
interferometer ensemble.}
Equation~\eqref{eq:hafnian-amplitude-hiding}
transfers a law; it does not assert that the Gaussian hafnian is unlikely to
be small.  That logically separate input is supplied by the companion
anticoncentration Article \cite{ZhaoAnticoncentration2026}.

\begingroup

\subsection{Several preselected output patterns}

Write $[L]:=\{1,\ldots,L\}$.  Let $q\ge1$.  For each $j\in[q]$, fix an ordered row
embedding
$\iota_j:[N]\hookrightarrow[L]$ and put $S_j:=\operatorname{im}\iota_j$.
Thus $S_j$ is a {fixed collision free output pattern}, while
$\iota_j$ records the ordering used to identify its principal block with
$\C^{N\times N}$.  Any fixed physical family in $[M]$ is first restricted to
the union of its rows, deterministically relabeled inside some $[L]$, and
equipped with such embeddings; Haar and Gaussian row invariance make this an
equality in law.  Write
\[
 \begin{aligned}
  \delta_{M,L}&:=\min\!\left\{1,C_\ast\frac{L^2}{M}\right\}.
 \end{aligned}
\]
{All probability laws in this subsection are over the common
random matrix $U$ or $G$; the embeddings $\iota_1,\ldots,\iota_q$ are fixed
query labels, not outputs sampled from the GBS distribution.}
Using the upper $L$ by $K$ block of one common Haar unitary $U$ and one common
$L$ by $K$ Gaussian matrix $G$, define
\[
 \begin{aligned}
  A_j^U&:=\frac{M}{\sqrt K}U_{\iota_j,[K]}U_{\iota_j,[K]}^{\T},\\
  A_j^G&:=\frac{1}{\sqrt K}G_{\iota_j,[K]}G_{\iota_j,[K]}^{\T},
 \end{aligned}
\]
where, for example,
$U_{\iota_j,[K]}=(U_{\iota_j(a),b})_{a\in[N],\,b\in[K]}$.
Apply Theorem~\ref{thm:uniform-product-hiding} once to the $L\times K$ union block,
and then take the tuple of principal submatrices.  Whenever $1\le L\le M$ and $1\le K\le M$,
\begin{equation}
 d_{\rm TV}\!\left(
 \Law(A_1^U,\ldots,A_q^U),
 \Law(A_1^G,\ldots,A_q^G)
 \right)
 \le \delta_{M,L}.
 \label{eq:joint-pattern-hiding}
\end{equation}
The Gaussian tuple in Eq.~\eqref{eq:joint-pattern-hiding} is coupled through
the common rows indexed by $[L]$; overlapping components are not asserted to be
independent.

The bound is uniform over all fixed families satisfying the displayed
constraints.  Consequently it also survives averaging over a random family
chosen independently of the matrix source ($U$ or $G$), including uniformly
chosen collision free row sets conditioned to be pairwise disjoint.  It does
not apply after selecting patterns from the $U$ {dependent} GBS output
distribution.  In quantum terms,
$|S_j|=N$ is controlled by restricting or postselecting to the collision free
$N$ {photon} sector; the physical experiment need not have deterministic total
photon number before that conditioning.

\subsection{Order statistics and {finite family} exceedance counts}

Let $h:\C^{N\times N}\to\R$ be measurable, set
$X_j^U=h(A_j^U)$, and for $A^G\sim\nu_{N,K}$ define the
{one pattern} Gaussian distribution function and
{upper tail} probability by
\[
 F_h(t):=\Pp\{h(A^G)\le t\},\qquad \pi_h(t):=1-F_h(t).
\]

\begin{proposition}[Order and {threshold count} transfer]
\label{prop:score-panel-transfer}
Suppose that the ranges of $\iota_1,\ldots,\iota_q$ are pairwise disjoint
and $1\le L\le M$ and $1\le K\le M$.  Then, for every $t\in\R$,
\begin{equation}
 \left|
 \Pp\!\left\{\max_{1\le j\le q}X_j^U\le t\right\}
 -F_h(t)^q
 \right|
 \le \delta_{M,L}.
 \label{eq:max-score-transfer}
\end{equation}
Moreover, if
$C_t^U:=\sum_{j=1}^q\1\{X_j^U>t\}$, then
\begin{equation}
 d_{\rm TV}\!\left(
 \Law(C_t^U),\mathrm{Binomial}(q,\pi_h(t))
 \right)
 \le \delta_{M,L}.
 \label{eq:heavy-count-transfer}
\end{equation}
The same error controls every measurable order statistic of the score vector.
\end{proposition}

\begin{proof}
For disjoint ranges, the maximum, threshold count, and each order statistic
are measurable maps of $(A_1,\ldots,A_q)$.  Apply deterministic data
processing directly to Eq.~\eqref{eq:joint-pattern-hiding}.  Under the Gaussian
reference the row blocks are disjoint, so the Gaussian scores are independent
and identically distributed.  This gives the product distribution $F_h(t)^q$
and the stated binomial law.
\end{proof}

\endgroup
\endgroup
\bibliographystyle{quantum}
\begingroup
\small
\bibliography{references}
\endgroup
\end{document}